\documentclass[11pt]{article}
\usepackage[margin=1in]{geometry}
\usepackage{amsmath}
\usepackage{graphicx}
\usepackage{mathtools}
\usepackage{comment}
\usepackage{enumerate}
\usepackage{booktabs}
\usepackage{threeparttable}

\usepackage{url} 
\newcommand{\blind}{1}

\usepackage{amsmath,amsthm, amssymb, bbm}
\usepackage{blindtext}
\usepackage{enumerate}
\usepackage{listings}
\usepackage{color, colortbl}
\usepackage{xcolor}
\usepackage[nottoc]{tocbibind}
\usepackage[utf8]{inputenc}
\usepackage{caption}
\usepackage{tabularx}
\usepackage{array}
\usepackage{rotating}
\usepackage{wrapfig}
\usepackage{graphicx}
\usepackage{subcaption}
\definecolor {processblue}{cmyk}{0.96,0,0,0}
\usepackage{mathrsfs}
\usepackage{hyperref}
\hypersetup{
	colorlinks   = true,
	linkcolor = red,
	citecolor    = blue
}
\usepackage{fancyvrb}
\usepackage{enumitem}
\definecolor{LightCyan}{rgb}{0.88,1,1}
\usepackage{parskip}
\usepackage{float}
\usepackage{titling}
\usepackage{varwidth}
\usepackage{multirow}
\usepackage{breqn}
\usepackage{bm}
\usepackage{titlesec}
\titleformat{\subsection}    
{\bfseries\itshape}{\thesubsection}{1em}{}
\usepackage{xr}
\usepackage{setspace}

\newtheorem{assumption}{Assumption}
\newtheorem{lemma}{Lemma}
\newtheorem{proposition}{Proposition}

\theoremstyle{remark}

\newtheorem{remark}{Remark}
\usepackage{amssymb,mathtools}

\makeatletter
\newcommand*{\addFileDependency}[1]{
	\typeout{(#1)}
	\@addtofilelist{#1}
	\IfFileExists{#1}{}{\typeout{No file #1.}}
}
\makeatother

\def\indist{\rightsquigarrow}
\def\ind{\perp\!\!\!\perp}

\newcommand{\var}{\text{var}}

\newcommand{\Pb}{\mathbb{P}}

\newcommand{\Pn}{\mathbb{P}_n}

\newcommand{\E}{\mathbb{E}}
\newcommand{\R}{\mathbb{R}}

\newcommand{\muhat}{\widehat\mu}

\newcommand{\mhat}{\widehat{m}}

\newcommand{\rhat}{\widehat{r}}

\def\logit{\text{logit}}
\def\expit{\text{expit}}

\usepackage[backend=bibtex, style=nature, doi=true, url=true, citestyle=authoryear-comp, intitle=true, maxbibnames=99]{biblatex}
\bibliography{ref.bib}

\begin{document}

	\def\spacingset#1{\renewcommand{\baselinestretch}%
		{#1}\small\normalsize} \spacingset{1}

	
	\if1\blind
	{
		\title{\bf On regression with estimated covariates and conditional effects given the propensity score}
	\author{Jiaqi Wu\thanks{Department of Statistics, Rutgers, The State University of New Jersey. Email: jw1726@stat.rutgers.edu} \and Matteo Bonvini\thanks{Corresponding author. Department of Statistics, Rutgers, The State University of New Jersey. Email: mb1662@stat.rutgers.edu.} \and Edward H. Kennedy\thanks{Department of Statistics \& Data Science, Carnegie Mellon University. Email: edward@stat.cmu.edu.} \and Jennie E. Brand \thanks{Department of Sociology, University of California, Los Angeles. Email: brand@soc.ucla.edu} \and Yu Xie \thanks{Department of Sociology, Princeton University. Email: yuxie@princeton.edu}}
		\date{ \today \\ \medskip \textit{Preliminary. Comments welcome.}}
		\maketitle
	} \fi
	
	\if0\blind
	{
		\bigskip
		\bigskip
		\bigskip
		\begin{center}
			{\bf Title}
		\end{center}
		\medskip
	} \fi
	
	\begin{abstract}
Motivated by the study of heterogeneous returns to education in \cite{brand2010benefits}, which considers how the effect of completing college on earnings varies with the (unknown) probability of completing college, we analyze the problem of estimating a nonparametric regression function when certain covariates are estimated in a first step. Plug-in estimators that treat the estimated covariates as known generally suffer from first-stage estimation error. To mitigate this issue, we analyze two debiasing approaches within a framework that is agnostic to the choice of the first-stage estimation method and relies on either local-smoothing or sieve-based methods for the second-stage regression. In particular, we consider: (i) influence function-based estimators of pathwise differentiable parameters that approximate the target estimand, and (ii) a variant of plug-in estimators that directly aims to correct their bias. For each method, we upper bound the estimation error and characterize conditions under which oracle rates can be approached, highlighting the possible gains in terms of convergence rates relative to the plug-ins. Simulation studies illustrate the finite-sample behavior of the methods. We apply our methodology to data from the National Longitudinal Survey of Youth 1997 and find evidence that completing college yields the largest reductions in unemployment for individuals least likely to do so, consistent with earlier findings in the literature (\cite{brand2010benefits}; \cite{brand2023overcoming}).
	\end{abstract}
\section{Introduction}
In this work, we consider the problem of estimating a nonparametric regression function in which certain covariates are not directly observed but need to be estimated in a first stage. Our work is motivated by the study of returns to education on certain outcomes of interest, such as labor-related outcomes. In particular, \cite{brand2010benefits} and \cite{brand2023overcoming} have shown that the causal effect of completing college on wages, among other outcomes, is greater for those students less likely to do so. Their finding provides evidence in favor of the \textit{negative selection hypothesis} whereby students who are least likely to  receive treatment (in this case college completion) are also the ones who would benefit the most from it. This contrasts with the \textit{positive selection hypothesis} which views students as rational agents completing college with larger probability to yield a greater benefit from doing so. Negative selection may occur because attending and completing college is a complex social phenomenon governed by many mechanisms in addition to pure economic incentives. For instance, for students from an advantaged background, completing college may be the norm whether or not it can provide a significant boost to their wages. We refer the reader to the original works and references therein.

The approach taken in \cite{brand2010benefits} and here is to estimate the conditional average treatment effect, identified under the no-unmeasured-confounding assumption (in addition to consistency and positivity), conditional on the unobserved but estimable probability of treatment (the propensity score). Mathematically, our goal is to estimate $\tau(t) := \E(Y^1 - Y^0 \mid r(X) = t)$, where $Y^a$ denotes the potential outcome had treatment been set to $A = a$, $A$ denotes the indicator for treatment and $r(X) = \Pb(A = 1 \mid X)$ denotes the propensity score. As the propensity score is unknown in observational studies, our work contributes to the literature on nonparametric regression with generated regressors, which our motivating application allows us to connect to the literature on conditional average treatment effect estimation (e.g., \cite{foster2023orthogonal}, \cite{kennedy2020optimal} and \cite{nie2021quasi}). Even sidestepping for a moment that $\tau(t)$ depends on the partially observed outcome $Y^1 - Y^0$ (for each unit, at most one potential outcome can be observed), our target is a nonparametric regression function on an estimated covariate. As such, even if the regressor was known exactly, our target is not root-$n$ estimable in nonparametric models. Because $\tau(t)$ is a one-dimensional curve while $r(X)$ is a multi-dimensional surface, we aim to be as agnostic as possible with respect to the estimators of $r$ while considering local smoothing and sieve-based estimators of the second-stage regression of $Y$ on $r(X)$.

Several studies have examined the effects of estimated covariates on downstream tasks. A canonical example occurs when a nonparametric regression on estimated covariates enters the construction of an estimator of a finite-dimensional parameter of interest. In this case, the estimand would be a functional of the regression on estimated covariates, which would thus act as a nuisance parameter rather than as the main estimand as in our setting \parencite{hahn2013asymptotic, hahn2018nonparametric, mammen2016semiparametric, escanciano2014uniform}. The main focus of this literature is on deriving asymptotic normality (at the root-$n$ scale) and establishing the contribution of the generated regressors to the limiting variance. The intermediate estimators of the nonparametric regression on estimated covariates considered are of the \textit{plug-in} variety. That is, the outcome is simply regressed on the estimated covariate without further adjustments, see, e.g., Section 5 in \cite{hahn2018nonparametric}, Section 2.2 in \cite{mammen2016semiparametric}. Without further modifications or committing to a specific estimator for the generated regressors, such plug-in estimators exhibit errors that are essentially first-order in the error incurred in estimating the covariates; see, e.g., Theorem 2 in \cite{mammen2016semiparametric} or Appendix D in \cite{hahn2018nonparametric}. More recently, \cite{escanciano2023automatic} derive automatically debiased estimators of parameters defined by finite-dimensional moment conditions that rely on nuisances that can depend on generated regressors, extending the literature on double machine learning \parencite{chernozhukov2018double, chernozhukov2021automatic, kennedy2022semiparametric} to cover settings with generated covariates.

There are relatively fewer works that consider a nonparametric regression on estimated covariates as the main target of inference. Our work builds directly on \cite{mammen2012nonparametric}, which present a detailed analysis of the error incurred by estimating the regressor(s) in a first step by analyzing a second-stage, plug-in local linear estimator that regresses the observed outcome onto the estimated covariate. As discussed in Section \ref{sec:mammen}, the convergence rates they derive for the plug-in estimator cast a rather pessimistic light on the accuracy with which one can recover the estimated regression function with generated regressors. Their work is rather general as they do not assume that the unobserved covariate is itself a regression function, although this setting is one of the several examples considered in their work. In our work, the unobserved regressor is a regression function (the propensity score). One of the main goals of our paper is to investigate whether this extra structure can be exploited to construct estimators with more favorable statistical properties than plug-in estimators. 

Other examples of plug-in estimators of two-stage regression with estimated covariates can be found in \cite{song2008uniform} and \cite{andrews1995nonparametric}. \cite{sperlich2009note} derives the asymptotic distribution for a local smoothing estimator regressing the outcome on estimated covariates under conditions on the bias and variance of the estimated regressors; see also \cite{scholz2016nonparametric}. This formulation of the problem bears resemblance to that of the error-in-variables regression \parencite{fan1993nonparametric}, where the ``covariate measured with error'' consists of the true covariate plus bias and variance terms. One key difference with the literature on error-in-variables models is that consistency of the estimator of the unknown regressors means that the error in ``measuring'' the covariates vanishes as the sample size increases. In order to flexibly incorporate black-box, machine learning methods for estimating the unobserved regressors, we prefer specifying only high-level, rate conditions on the covariates' estimators and thus we build directly on the approach taken in \cite{mammen2012nonparametric}. 

Furthermore, \cite{wuHanTongLi2024propensity} considers estimating the conditional average treatment effect given a low-dimensional vector of effect modifiers, conditioning on  the propensity score to deconfound the treatment-outcome association. However, they operate under restrictive conditions on the estimated propensity score effectively ensuring that its error, assumed to be of order $n^{-1/2}$, is asymptotically negligible (see their Assumption 1). Finally, \cite{tian2026bracketing} propose regressing conditional average treatment effects on estimated propensity scores---the same type of estimand we consider in Section~\ref{sec:application}---as a diagnostic tool for assessing monotonicity assumptions used to bracket different causal effects. Their theoretical analysis, however, does not focus on the estimation of these regressions and, in particular, does not account for the additional error induced by using estimated propensity scores. Further, their implementation is limited to low-dimensional parametric models for both the conditional effects and the propensity scores. Our work therefore provides a range of estimators for implementing this diagnostic approach, together with theoretical guarantees allowing for more flexible nuisance estimation.

Our manuscript proceeds as follows. We describe the setup and notation in Sections \ref{sec:setup} and \ref{sec:notation}, while Sections \ref{sec:mammen} and \ref{sec:contributions} summarize the results of \cite{mammen2012nonparametric} most relevant to our setting and our main contributions, respectively. To simplify the exposition and isolate what we view as the key challenges of the problem we consider, we first introduce our proposed estimators for the closely related problem of estimating $m(t) = \E(Y \mid r(X) = t)$, where $r(X) = \E(A \mid X)$, using $n$ iid copies of $(Y, A, X) \sim \Pb$. While we strive to be agnostic with respect to the estimation of $r(x)$, we focus on local-smoothing-based estimators (Section \ref{sec:local_smoothing}) and sieve-based estimators (Section \ref{sec:sieves}) for the second-stage regression.  
Sections \ref{sec:sims} and \ref{sec:da} contain our simulation experiments and our empirical application, respectively. 

\subsection{Setup}\label{sec:setup}
We first consider estimating $m(t) = \E\{Y \mid r(X) = t\}$, where $r(X) = \E(A \mid X)$, using $n$ iid copies of $Z \equiv (Y, A, X) \sim \Pb$. We assume that $Y$ and $A$ are scalar random variables while $X \in \R^p$. The regression model is 
\begin{align*}
    Y = \mu(X) + \epsilon, \quad \text{ where } \quad \E(\epsilon \mid X) = 0.
\end{align*}
There is a wide range of possible structural assumptions one may impose on this problem. For example, suppose $X \sim \mathcal{N}_p(\theta, \Sigma)$, $\mu(X) = \beta^\intercal X$ and $r(X) = \expit(X^\intercal \gamma)$. Then, $m(t)$ has the following closed-form solution:
\begin{align*}
    m(t) = \beta^\intercal \E\left\{X \mid X^\intercal \gamma = \logit(t)\right\} = \beta^\intercal\theta + \frac{\beta^\intercal\Sigma \gamma}{\gamma^\intercal \Sigma \gamma}\left\{\logit (t) - \gamma^\intercal \theta\right\}.
\end{align*}
Thus, in this case, the parameters can be estimated by maximum likelihood. The focus of this paper, however, is on estimating $m(t)$ in nonparametric models for both $t \mapsto m(t)$ and $x \mapsto r(x)$.

As noted in \cite{mammen2012nonparametric}, one key factor dictating the difficulty of the problem is whether the regressor $r(X)$ is enough to capture all the influence of $X$ on $Y$. That is, whether the random variable 
\begin{align}\label{eq:rho}
    \rho(X) \equiv \E(Y \mid X) - \E\{Y \mid r(X)\}
\end{align}
is zero almost surely or not. In our application, this assumption would impose that the effect heterogeneity induced by $X$ can be fully captured by the heterogeneity induced by $r(X)$, i.e., $\E(Y^1 - Y^0 \mid X) = \E\{Y^1 - Y^0 \mid r(X)\}$. This is likely not the case in practice. However, it satisfies $\E\{\rho(X) \mid r(X)\} = 0$ by construction, which also implies $\E\{\rho(X)\} = 0$ by the law of iterated expectation. Thus, although $\rho(X)$ may be nonzero, the remaining heterogeneity not captured by $r(X)$ averages to zero within each stratum of $r(X)$ and across the population. In deriving our results, we keep the dependence of the estimators' errors on $\rho(X)$ explicit. Based on the conditional mean-zero property that $\E\{\rho(X) \mid r(X)\} = 0$, one may hope that $\E\{\rho(X) \mid \rhat(X)\}$ is converging to zero, possibly at a rate comparable to $\sup_{x \in \R^p}|\rhat(x) - r(x)|$. Relaxing this assumption is important avenue for future work. We refer to \cite{mammen2012nonparametric} for a discussion of empirical applications where $\rho(X) = 0$ almost surely is a plausible condition.\footnote{While in many applications it may be overly restrictive, the parametric Gaussian model above does not generally impose the single index structure $\mu(X) = m(r(X))$ (corresponding to the condition $\rho = 0$ in our notation); this can be seen by direct computation since 
\begin{align*}
    \mu(X) - m(r(X)) = (X - \theta)^\intercal\left(\beta - \gamma \frac{\beta^\intercal \Sigma \gamma}{\gamma^\intercal \Sigma \gamma}\right),
\end{align*}
which would be zero only in special cases (such as when $\gamma$ is aligned with $\beta$.)} 

\medskip 

\begin{assumption}\label{main:assumption} 
Throughout we make the following assumptions:
\begin{enumerate}

    \item 
    $r(X)$ is continuously distributed with compact support $I_r$. Its density function $f_r(\cdot)$ is Lipschitz continuous and bounded above and away from zero on $I_r$. To simplify the exposition, throughout, we take $I_r = [0, 1]$.

    \item 
    The function $m(t) = \E\{Y \mid r(X) = t\}$ is twice differentiable on $I_r$, with bounded first and second derivatives.

    \item The random variables $Y$ and $A$ are bounded.
\end{enumerate}
\end{assumption}
Assumption \ref{main:assumption} imposes mild regularity conditions on the data generating process. In particular, the first condition allows us to apply standard arguments to derive the asymptotic normality of the oracle estimator with access to the true covariate $r(X)$. The second condition imposes mild smoothness on the true regression function, which we leverage both to derive the convergence rate of the oracle estimator and to bound the contribution of having to estimate $r(x)$ to the final rate. Notice that if $m(t)$ possesses additional smoothness, one could employ higher order polynomial regression to better track it. Finally, the third condition is a commonly invoked boundedness condition (it can be found, e.g., \cite{kennedy2022minimax, robins2017higher}), which we expect can be relaxed with a more careful analysis.

\subsection{Notation}\label{sec:notation}
Throughout, we employ the notation $\Pb f = \int f(z) d\Pb(z)$ and $\Pn f = n^{-1} \sum_{i=1}^n f(Z_i)$. Further, we define the $L_q$- norm $\|f\|^q_{q, \Pb} = \int |f|^q(z) d\Pb(z)$ and the sup-norm $\|f\|_\infty = \sup_z|f(z)|$. For a vector \(v \in \mathbb R^k\), we write the Euclidean norm $\|v\|^2_2 = \sum_{j=1}^k v_j^2$. We let $a \lesssim b$ denote $a \leq C b$ for some constant $C$ not depending on the sample size. We let $a \asymp b$ if $a \lesssim b$ and $b \lesssim a$. For a $k \times k$ matrix $M$, we let $\|M \|_{\rm op} = \sup_{v: \|v\|_2 \neq 0} \|M v\|_2 / \|v\|_2$ denote the operator norm of $M$. 

When the context is clear, we abbreviate the notation to represent a random variable that is a function of $X$ as follows: $\rhat \equiv \rhat(X)$ and $\rhat_i \equiv \rhat(X_i)$. We have $\mu(X) = \E(Y \mid X)$, $\rho = \mu(X) - \E\{Y \mid r(X)\} \equiv \mu - m(r)$. Thus, $Y = m(r) + \rho + \epsilon$, where $\E(\epsilon \mid X) = 0$ and $\E(\rho \mid r) = 0$. We assume that the sample consists of $n$ iid observations and let $D^n$ denote a separate, auxiliary iid sample of size $n$ used to estimate $r(X)$ (and all other nuisance functions).

To describe the local smoothing estimators, we introduce the following notation. We let $K(u)$ denote a symmetric twice continuously differentiable density function, with derivative $K'(u)$ and compact support. We let $K_{ht}(r) = h^{-1}K((r-t)/h)$, $K'_{ht}(r) = h^{-2}K'((r-t)/h)$ and $g_{ht}(r) = \begin{bmatrix} 1 & (r-t)/h\end{bmatrix}^\intercal$. We also let $e_1$ and $e_2$ denote the standard basis of $\R^2$. 

Finally, to describe the estimators based on sieves, we let $u \mapsto \Phi_k(u)$ denote  a $k$-dimensional vector of basis functions and $u \mapsto \dot\Phi_k(u)$ denote the vector of their derivatives. For instance, if $\Phi_k(u)$ denotes the cosine basis, we have $\Phi_0(u) = 1$, $\Phi_j(u) = \sqrt{2} \cos(\pi j u)$ and $\dot\Phi(u) = - \sqrt{2}\pi j \cdot \sin(\pi j u)$. Notice that $\Phi_0(u), \Phi_1(u), \ldots, \Phi_{k-1}(u)$ is orthonormal in $L_2[0, 1]$ and, further, that $\int_0^1 \dot\Phi(u)\dot\Phi(u)^\intercal du$ is a diagonal matrix with non-zero entries equal to $\pi^2 \cdot j^2$.
\subsection{Previous work}\label{sec:mammen}
To the best of our knowledge, \cite{mammen2012nonparametric} provide the most recent comprehensive analysis of the problem of estimating a regression with estimated covariates in a nonparametric setting (i.e., both $m(t)$ and $r(X)$ are estimated nonparametrically).\footnote{The setting covered by \cite{mammen2012nonparametric} is more general than ours since they do not assume that $r$ is a regression function, and they allow for the possibility of multiple unobserved covariates. Our understanding is that the setup considered here captures the main subtleties and difficulties of the problem well.}  Their first result (Theorem 1) bounds the error incurred by the plug-in estimator that regresses $Y$ on $\rhat(X)$ via local linear smoothing. To simplify our exposition and avoid imposing certain empirical process conditions, we consider a variant of their strategy by which the generated regressor $\rhat(x)$ is computed using the separate independent sample $D^n$. For a vanishing bandwidth $h$, let $W_{it}(X^n, \rhat; h)$ denote the weight of the local linear regression for observation $i$, i.e., 
\begin{align*}
W_{it}(X^n, \rhat; h) = e_1^\intercal \cdot [\Pn\{K_{ht}(\rhat) g_{ht}(\rhat)g_{ht}(\rhat)^\intercal\}]^{-1} \cdot K_{ht}(\rhat_i) \cdot  g_{ht}(\rhat_i).
\end{align*}
The estimator analyzed in \cite{mammen2012nonparametric} can be written as $\widehat{m}(t) = n^{-1}\sum_{i=1}^n W_{it}(X^n, \rhat; h) Y_i$. Its error satisfies the following decomposition. Consider a first-order Taylor expansion of $m(r_i)$ around $m(t)$:
\begin{align*}
m(r_i) = g_{ht}(\rhat_i)^\intercal \begin{bmatrix}
    m(t) \\ h \cdot m'(t)
\end{bmatrix} + m'(t)(r_i - \rhat_i) + \{m'(\overline{t}_i)  - m'(t)\}(r_i - t),
\end{align*}
where $\overline{t}_i$ is an intermediate value between $r_i$ and $t$. Since $Y_i = m(r_i) + \rho_i + \epsilon_i$, we have
\begin{align*}
    \widehat{m}(t) - m(t) = \frac{1}{n}\sum_{i=1}^n W_{it}(X^n, \rhat; h) \left[m'(t)(r_i - \rhat_i) + \rho_i +  \{m'(\overline{t}_i)  - m'(t)\}(r_i - t) + \epsilon_i \right].
\end{align*}
If $m'(t)$ is Lipschitz continuous and the weights $W_{it}(X^n, \rhat; h)$ are localized, we have
\begin{align*}
    \left| W_{it}(X^n, \rhat; h)\{m'(\overline{t}_i)  - m'(t)\}(r_i - t)\right| & \lesssim   \left| W_{it}(X^n, \rhat; h)(r_i - t)^2\right| \\
    & \lesssim |W_{it}(X^n, \rhat; h)|\{h^2 + (r_i - \rhat_i)^2\}.
\end{align*}
The error term $\epsilon_i$ is mean-zero given $(D^n, X^n)$; in this respect, a standard local smoothing argument yields that $n^{-1}\sum_{i=1}^n W_{it}(X^n, \rhat; h)\epsilon_i = O_{\Pb}((nh)^{-1/2})$ under mild conditions. Therefore, we may write
\begin{align*}
  \widehat{m}(t) - m(t) & = \frac{1}{n}\sum_{i=1}^n W_{it}(X^n, \rhat; h) \{m'(t) \cdot (r_i - \rhat_i) + \rho_i\} + O_{\Pb}\left(h^2 + \frac{1}{\sqrt{nh}} + \|\rhat - r\|_\infty^2\right).
\end{align*}
Notice that, unless the conditional mean of $Y$ given $X$ is only a function of $r(X)$, $\rho(X)$ is non-zero and may contribute first-order to the error $\widehat{m}(t) - m(t)$. Even if $\rho(X)$ is not zero, we have, by construction, $\E\{\rho(X) \mid r(X)\} = 0$. In this light, the expression above further reduces to
\begin{align*}
  \widehat{m}(t) - m(t) & = m'(t) \cdot \frac{1}{n}\sum_{i=1}^n W_{it}(X^n, \rhat; h) (r_i - \rhat_i) + \frac{1}{n}\sum_{i=1}^n \{W_{it}(X^n, \rhat; h) - W_{it}(X^n, r; h)\}\rho_i \nonumber \\
  & \hphantom{=} + O_{\Pb}\left(h^2 + \frac{1}{\sqrt{nh}} + \|\rhat - r\|_\infty^2\right).
\end{align*}
Under mild smoothness assumptions on the kernel function, by the mean-value-theorem, one may show that the second term can be bounded as
\begin{align*}
\left| \frac{1}{n}\sum_{i=1}^n \{W_{it}(X^n, \rhat; h) - W_{it}(X^n, r; h)\}\rho_i \right|  \lesssim \frac{\|\rhat - r\|_\infty \sup_{t_1, t_2}|\E(\rho \mid \rhat = t_1, r = t_2, D^n)|}{h} + O_{\Pb}\left(\frac{1}{\sqrt{nh}}\right).
\end{align*}
The fact that $\E\{\rho(X) \mid r(X)\} = 0$ motivates \cite{mammen2012nonparametric} to consider an assumption akin to $|\E\{\rho(X)\mid\rhat(X), D^n\}| \ \lesssim \|\rhat - r\|_\infty$ almost surely (their Assumption 4). Under the stronger assumption that $|\E(\rho \mid \rhat, r, D^n)|\lesssim \|\rhat - r\|_\infty$ almost surely, this informal analysis of the plug-in estimator yields that
\begin{align}\label{eq:plugin_expansion_local}
    \widehat{m}(t) - m(t) = m'(t) \cdot \frac{1}{n}\sum_{i=1}^n W_{it}(X^n, \rhat; h) (r_i - \rhat_i) + O_{\Pb}\left(h^2 + \frac{1}{\sqrt{nh}} + \frac{\|\rhat - r\|_\infty^2}{h}\right)
\end{align}
The oracle rate for estimating $m(t)$ had $r(X)$ been known is of order $h^2 + (nh)^{-1/2} \asymp n^{-2/5}$ when $h \asymp n^{-1/5}$. As pointed out in \cite{mammen2012nonparametric}, the non-stochastic component of the first term would typically be first-order in $r - \rhat$. Interestingly, the decrease in variance from the fact that this first term is a kernel weighted average of $r_i - \rhat_i$ means that it might be negligible relative to the oracle rate $n^{-2/5}$ even if $\rhat$ converges to $r$ at a slower rate. See their Corollaries 1-4 for the case when $\rhat$ itself is a local polynomial estimator. Informally, removing from their Corollary 1's error decomposition the portion related to the first-stage function class' complexity (in the spirit of using sample splitting), the rate obtained (applied to our settings) is of order
\begin{align*}
    \frac{1}{\sqrt{nh}} + h^2 + \frac{\|\rhat - r\|^2_\infty}{h} + h\|\rhat - r\|_\infty + \frac{1}{n g^p} + g^{\alpha},
\end{align*}
if $r(x)$ is $\alpha$-smooth and it is estimated by a local polynomial with bandwidth $g$. When $h \asymp n^{-1/5}$, the oracle rate is achieved if $\|\rhat - r\|_\infty = o_\Pb(n^{-3/10})$. Ignoring log terms, this means that $(ng^p)^{-1/2} = o(n^{-3/10}) \iff g \gg n^{-2/(5p)}$. To make $g^\alpha$ negligible at $n^{-2/5}$-scale, one needs $g \ll n^{-2/(5\alpha)}$, which leads to the requirement that $\alpha > p$. 

This work aims to study new estimators of $m(t)$ with more favorable error decompositions than Eq. \eqref{eq:plugin_expansion_local}, while remaining agnostic with respect to the choice of the estimator of $r$. Following the principles of semiparametric efficiency theory, plug-in estimators of pathwise differentiable parameters can often be substantially improved through debiasing based on the parameter's influence function \parencite{bickel1993efficient, tsiatis2006semiparametric, chernozhukov2018double, kennedy2022semiparametric, liu2017semiparametric, robins2017higher}. Related ideas have recently been considerably generalized to cover parameters, such as conditional average treatment effects and dose-response functions, that are not root-$n$ estimable in nonparametric models \parencite{foster2023orthogonal, kennedy2020optimal, kennedy2022minimax, nie2021quasi, kennedy2017nonparametric, luedtke2024one, zhang2026higher, bonvini2022fast}. In these settings, the estimands are typically defined as regressions of unobserved but estimable pseudo-outcomes on observed covariates; the resulting convergence rates often take the form ``oracle rate plus second-order nuisance errors,'' where the oracle rate is the rate that would be achieved by an estimator with access to the true pseudo-outcome. Our analysis suggests that estimating a regression function with estimated covariates is, in general, a substantially different and more delicate statistical problem than estimating a regression function with an estimated pseudo-outcome. In particular, one should not expect the pure second-order nuisance errors that may be obtainable when the outcome is estimated and the covariates are known. A formal negative result would require minimax lower bounds, which we leave as an important direction for future work. Section~\ref{appendix:sec_estimated_outcomes_vs_covariates} presents a stylized example highlighting a key distinction between unobserved covariates and unobserved outcomes.
\subsection{Summary of our methodological contributions}\label{sec:contributions}
The estimation strategy summarized above is rather general since it does not require the unobserved covariate $r(X)$ to be a regression function. It is natural to ask whether better rates can be achieved when $r(X) = \E(A \mid X)$, as is the case for our motivating example. Our attempt to answer this question hinges on a couple of observations. First, for a fixed bandwidth $h$, the local linear moment condition:
\begin{align*}
    \E\left(K_{ht}\{r(X)\}g_{ht}\{r(X)\}\left[\mu(X) - g_{ht}\{r(X)\}^\intercal \beta\right]\right) = 0
\end{align*}
is a pathwise differentiable parameter with influence function $\varphi(Z; \beta) = \varphi_1(Z) - \varphi_2(Z) \beta$, where $\varphi_1$ and $\varphi_2$ are given in \eqref{eq:varphi_1} and \eqref{eq:varphi_2}, respectively. This motivates the estimator $\mhat_{\text{bc}}(t) = g_{ht}(t)^\intercal \widehat\beta$, where $\widehat\beta$ solves $n^{-1} \sum_{i=1}^n \widehat\varphi(Z_i, \widehat\beta) = 0$. The subscript  ``lbc" stands for \underline{b}ias-\underline{c}orrection of a plug-in estimator based on \underline{l}ocal smoothing. A similar estimator based on sieves can be derived by solving the influence function-based estimator of the following moment condition:
\begin{align*}
    \E\left(\Phi_k\{r(X)\}\left[\mu(X) - \Phi_k\{r(X)\}^\intercal \beta\right] \right) = 0,
\end{align*}
yielding a bias-corrected estimator $\mhat_{\rm sbc}(t; k) = \Phi_k(t)^\intercal\widehat\beta_{\rm bc}$. The subscript ``sbc" stands for \underline{b}ias-\underline{c}orrection of a plug-in estimator based on \underline{s}ieves. Our analysis reveals that, while the bias of these estimators can be lower than that of their plug-in counterparts, their variances may in fact be larger. This is because the presence of the derivatives of the kernel function or the basis vector involve an additional division by $h$ and multiplication by $k$, respectively, which inflates our bounds on the variances when $h\to 0$ and $k \to \infty$. When $h$ or $k$ are fixed constants, the inflation in the variance does not show up in the rates and the bias-correction works in the usual way as for pathwise differentiable parameters. In particular, for a fixed bandwidth or number of basis terms, the bias is second-order in $(\rhat - r)$. 

Our bound on the estimation errors of $\widehat{m}_{\text{lbc}}(t; h)$ and $\mhat_{\rm sbc}(t; k)$ are such that, depending on the data generating process, the reduction in the bias from estimating $r(X)$ might not be enough to offset the increase in variance so that the overall convergence rate is in fact slower than that of the plug-in estimators. For example, consider the estimator $\mhat_{\rm {lbc}}(t; h)$ based on local smoothing and suppose that the smoothing bias is of order $h^2$. The extra division by $h$ coming from the derivative of the kernel function leads to a bound on the variance of order $(nh^3)^{-1}$, which then appears to preclude attainment of the oracle rate, which is obtained by balancing the oracle standard error and bias, of orders $(nh)^{-1/2}$ and $h^2$, respectively. We find that this issue is substantially mitigated if one is willing to assume that $\rho = 0$ since, under this condition, our inflated bound on the standard error, $(nh^3)^{-1/2}$, gets multiplied by $\|\muhat - \mu\|_\infty + \|\rhat - r\|_\infty$, which goes to zero at a certain rate. However, in the empirical application we consider, the condition $\rho = 0$ is highly implausible.

Such potential shortcomings of these bias-corrected estimators motivate us to consider other strategies to improve the plug-ins. In this respect, the second estimation procedure that we study consists of the plug-in estimator minus an estimate of the first terms in \eqref{eq:plugin_expansion_local} and \eqref{eq:plugin_expansion_sieves}. To refer to these estimators, we rely on the subscript ``lcpi" and ``scpi" to highlight that these are \underline{c}orrected \underline{p}lug-\underline{i}ns based on \underline{l}ocal smoothing and \underline{s}ieves. The rationale is that the first term in either \eqref{eq:plugin_expansion_local} and \eqref{eq:plugin_expansion_sieves} is one of the two leading terms involving the first-order difference $(\rhat - r)$. The other term involves $\rho$ and it appears more difficult to correct for. One may hope that this term is negligible, since $\E(\rho \mid r) = 0$ by definition. 

Finally, inspired by the calibrated debiased machine learning framework in \cite{van2024automatic}, we also consider an isotonic calibration step for $\rhat$. In our numerical experiments, we have found that calibrating $\rhat(X)$ yields more stable numerical performance. In particular, we estimate $\rhat(X)$ on $D^n$ and calibrate it on the main sample to get a final estimate $\rhat_{\text{cali}}(X)$ satisfying $\sum_{i = 1}^n g(\rhat_i)(A_i - \rhat_i) = 0$ for any function $g$. The rationale for calibrating $\rhat(X)$ is clear from Eq. \eqref{eq:plugin_expansion_local} since, by writing $r_i - \rhat_i = r_i - A_i + A_i - \rhat_i$ in the first term, one has
\begin{align*}
\widehat{m}(t) - m(t) & = m'(t) \cdot \frac{1}{n}\sum_{i=1}^n W_{it}(X^n, \rhat; h) (A_i - \rhat_i) + O_{\Pb}\left(h^2 + \frac{1}{\sqrt{nh}} + \frac{\|\rhat - r\|_\infty^2}{h}\right).
\end{align*}
If $\rhat(X)$ is calibrated, the first term would vanish. However, this informal analysis would need to be adjusted because $\rhat_{\text{cali}}$ now depends on the main sample as well as $D^n$. While we explore calibrating $\rhat$ in our simulations and empirical application, we leave a detailed analysis of the calibrated procedure for future work.

\section{Estimators based on local smoothing}\label{sec:local_smoothing}
In this section, we analyze the properties of the local smoothing estimators. To start, let us define the smoothing-based approximation to $m(t)$: 
\begin{align*}
    m(t; h) = e_1^\intercal Q_{ht}^{-1}(r) \E\{K_{ht}(r)g_{ht}(r)\mu\}, \quad \text{ where } \quad Q_{ht}(r) = \E\{K_{ht}(r) g_{ht}(r) g_{ht}(r)^\intercal\}.
\end{align*}
A standard smoothing argument yields that $|m(t; h) - m(t)| \ \lesssim h^2$ when $m'(t)$ is Lipschitz continuous; see, e.g., Proposition 1.13 in \cite{tsybakov2008introduction}. Sections \ref{sec:IFbased_smoothing} and \ref{sec:dpi_smoothing} describe and analyze our proposed local smoothing-based estimators of $m(t)$.
\subsection{Influence function-based estimator of the fixed-bandwidth approximating target}\label{sec:IFbased_smoothing}
The first estimator is motivated by solving the influence function-based estimate of 
$$\E[K_{ht}(r)g_{ht}(r)\{\mu - g_{ht}(r)^\intercal \beta\}],$$ 
in a nonparametric model when treating the bandwidth $h$ as fixed. The derivation of the influence function follows by standard calculations and is omitted; we refer to the recent review by \cite{kennedy2022semiparametric} on this topic. Here, we propose and analyze the estimator $\mhat_{\text{lbc}}(t; h) = e_1^\intercal \cdot  \{\Pn \varphi_2(Z,\rhat)\}^{-1} \cdot \Pn \varphi_1(Z,\rhat, \widehat\mu)$, where
\begin{align}
    \varphi_1(Z,r, \mu) & = \left[K'_{ht}(r)g_{ht}(r) + K_{ht}(r) g_{ht}'(r)\right] \cdot (A - r) \cdot \mu + K_{ht}(r)g_{ht}(r)Y, \label{eq:varphi_1} \\
    \varphi_2(Z,r) & = \left[K'_{ht}(r)g_{ht}(r)g_{ht}(r)^\intercal + K_{ht}(r) g_{ht}'(r)g_{ht}(r)^\intercal + K_{ht}(r) g_{ht}(r)g_{ht}'(r)^\intercal\right] \cdot (A - r) \nonumber \\
    & \hphantom{=} + K_{ht}(r)g_{ht}(r)g_{ht}(r)^\intercal. \label{eq:varphi_2}
\end{align}
Crucially, while the estimator is motivated by treating the bandwidth as fixed, our analysis allows the bandwidth to vanish with the sample size; this is one key difference with prior work studying finite-dimensional functionals, including the recent pre-print by \cite{escanciano2023automatic}. The terms multiplying $(A - r)$ represent the component-wise derivatives of $u \mapsto K_{ht}(u) g_{ht}(u)$ and $u \mapsto K_{ht}(u) g_{ht}(u) g_{ht}(u)^\intercal$, respectively. Notice that $g_{ht}'(t) = \begin{bmatrix} 0 & 1/h \end{bmatrix}^\intercal$. The following proposition bounds the pointwise error.
\medskip

\begin{proposition}\label{proposition:bc_local}
In addition to Assumption \ref{main:assumption}, assume that $t$ is an interior point, $nh^3 \to \infty$ and $\|\rhat - r\|_\infty = o_{\Pb}(h)$. Then it holds that
\begin{align*}
 & \mhat_{\rm lbc}(t; h) - m(t)  - \frac{1}{2} m''(t) h^2 \int K(u) u^2 du \\
 & = \left(\frac{1}{f_r(t)} \Pn \left[K_{ht}(r)\{\epsilon+\rho - m'(t) (A - r)\} + 
 K_{ht}'(r)(A - r) \rho\right] \right)\{1 + o_{\Pb}(1)\} + o(h^2)\\
& \hphantom{=} +  O_{\Pb}\left(\frac{\|r-\rhat\|_\infty+\|\mu-\muhat\|_\infty}{\sqrt {nh^3}}+\frac{\|\rhat - r\|_\infty\cdot\{\sup_{t_1,t_2}\E(\rho^2 \mid \rhat = t_1, r = t_2, D^n)\}^{1/2} }{\sqrt{nh^5}} \right) \\
&\hphantom{=} +O_{\Pb}\left(\frac{\|\rhat - r\|_\infty\left\{\sup_{t_1,t_2}|\E(\rho\mid\rhat = t_1,r = t_2, D^n) | +\|r-\rhat\|_\infty+\|\mu-\muhat\|_\infty\right\}}{h}\right)
\end{align*}
where $f_r(\cdot)$ denotes the density of $r(X)$.
\end{proposition}
Proposition \ref{proposition:bc_local} decomposes the error incurred by $\mhat_{\rm lbc}(t; h)$ into three components: (i) a smoothing bias term of order $h^2$; (ii) a sample average of $n$ iid observations (dependent by the bandwidth $h$) that is amenable to the application of a triangular array central limit theorem (CLT); and (iii) nuisance bias terms that depend on the bandwidth $h$, the term $\rho$ as well as the accuracy with which $r(x)$ and $\mu(x)$ are estimated. The first term is a standard smoothing bias term that would enter the error decomposition even if $r(X)$ were known. The second term is a CLT term whose scaling depends on whether $\rho$ is zero ($(nh)^{1/2}$ scale) or not ($(nh^3)^{1/2}$ scale). The third nuisance bias terms are summarized and interpreted below under simplifying assumptions. 

We note that all our results based on local smoothing rest on the assumption that $\|\widehat{r} - r\|_\infty = o(h)$; such assumption is also invoked in previous works, including \cite{mammen2012nonparametric} (Assumption 2) and we invoke a similar condition also for our results based on sieves estimation. Intuitively, the second-stage regression averages outcomes $Y_i$ such that $\rhat_i$ is within $h$ from $t$. If $\rhat_i - r_i$ remains large relative to $h$ as the sample size increases, the second-stage regression would localize in the wrong spot and thus our estimator would not approximate the oracle estimator with knowledge of $r$.

In the following, we expand on the implications of Proposition \ref{proposition:bc_local}. To simplify the rates and derive sufficient conditions on $\|r-\rhat\|_\infty$ to ensure asymptotic normality of $\mhat_{\rm lbc}(t; h)$, we further require the following rate conditions for $\widehat\mu$ and $\rho$. 

\medskip 
\begin{assumption}\label{ass:muhat}
The nuisance estimator satisfies $ \|\muhat-\mu\|_\infty \lesssim\|\rhat-r\|_\infty.$
\end{assumption}

\medskip

\begin{assumption}\label{ass:rho}
The conditional expectation of $\rho$ given $(\rhat, r)$ satisfies:
    \begin{align*}
      \sup_{t_1}|\E(\rho\mid\rhat = t_1, D^n) |\leq\sup_{t_1,t_2}|\E(\rho\mid\rhat = t_1,r = t_2, D^n)     |\lesssim\|r-\rhat\|_\infty.
    \end{align*}
\end{assumption}
Assumption \ref{ass:muhat} is made solely to ease the presentation of the results as it allows us to focus on the conditions needed on $\widehat{r}$; in our settings, it is a plausible condition as both $r(x)$ and $\mu(x)$ are regression functions on the same domain. However, if $\|\mu - \widehat\mu\|_\infty \gg \|\widehat{r} -r\|_\infty$, then one can derive similar conditions as those below, where the relevant requirements would need to be explicitly stated in terms of $\|\mu - \widehat\mu\|_\infty$. Assumption \ref{ass:rho} plays a more important role in deriving the rates below in the sense that if $\sup_{t_1, t_2} | \E(\rho \mid \rhat = t_1, r = t_2, D^n)|$ converges to zero slowly, then the nuisance bias might effectively be first-order in $\|\rhat - r\|_\infty$, rendering the bias-correction ineffective. The plausibility of Assumption \ref{ass:rho} comes from the condition $\E(\rho \mid r) = 0$ that holds by design; we consider investigating whether this condition can be relaxed as an important avenue for future work.

Under Assumptions \ref{ass:muhat} and \ref{ass:rho}, in the likely case that $\rho \neq 0$, the rate from Proposition \ref{proposition:bc_local} is
    \begin{align*}
        \widehat{m}_{\text{lbc}}(t; h) - m(t) = O_\Pb\left(h^2 +  \frac 1{\sqrt{nh^3}} + \frac{\|\rhat -r\|^2_\infty}{h}\right).
    \end{align*}
If $\|\rhat - r\|_\infty \gtrsim n^{-3/14}$, choosing $h \asymp \|\rhat - r\|_\infty^{2/3}$ yields the rate $O_{\Pb}(\|\rhat - r\|_\infty^{4/3} \ + \ n^{-1/2} \cdot \|\rhat - r\|_\infty^{-1})$. If $\|\rhat - r\|_\infty \lesssim n^{-3/14}$, choosing $h \asymp n^{-1/7}$ yields the rate $n^{-2/7}$. Further, if $\|\rhat - r\|_\infty = o_{\Pb}(n^{-3/14})$, then we have, under regularity conditions ensuring non-degeneracy and continuity of the conditional variance function, and with $h \asymp n^{-1/7}$:
    \begin{align*}
        \sqrt{nh^3}\left(\widehat{m}_{\text{lbc}}(t; h) - m(t) - \frac{1}{2}m''(t)h^2\int K(u)u^2\,du\right) \xrightarrow{d} N\left(0,\sigma^2(t)\right),
    \end{align*}
where $\sigma^2(t) = f_r(t)^{-1} \cdot \int K'(u)^2\,du \cdot \E\left[\left\{(A-r)\rho\right\}^2 \,\Big|\, r(X) = t\right].$

Recall that the oracle estimator with access to $r$ achieves a rate $n^{-2/5}$. Our analysis thus suggests that this rate is not attainable by $\widehat{m}_{\text{lbc}}(t; h)$ in the general case where $\rho \neq 0$. On the other hand, if $\rho = 0$, asymptotic Gaussianity is attainable (at rate $n^{-2/5}$ with $h \asymp n^{-1/5}$) if $\|\rhat - r\|_\infty = o_\Pb(n^{-3/10})$. 
\subsection{Direct bias correction of plug-in estimator}\label{sec:dpi_smoothing}
In this section, we consider directly estimating and subtracting off one of the two leading terms in the expansion $\mhat_{\text{lpi}}(t; h) - m(t)$ in \eqref{eq:plugin_expansion_local}, namely $m'(t) \cdot n^{-1} \sum_{i=1}^n W_{it}(X^n, \rhat;h)\{r(X_i) - \rhat(X_i)\}$. We propose estimating $m'(t)$ via local smoothing with a bandwidth $b$ that may differ from $h$. To simplify the analysis, we consider a leave-one-out type estimator in which the same observation is not used to estimate both the derivative term and the kernel smoothed average of $r - \rhat$. This naturally leads to a second-order U-statistic. In the simulations, we find that simply regressing $Y - \widehat{m}'(\rhat; b) \cdot (A - \rhat)$ onto $\rhat$ via local linear regression yields improvements in performance relative to the plug-in estimator $\mhat_{\rm lpi}(t; h)$.

To describe the estimator $\mhat_{\rm lcpi}(t; h, b)$, we need to introduce some additional notation. For a vanishing bandwidth $b$, we define
\begin{align*}
    & \widehat{Q}_{bt, -j} = \frac{1}{n-1} \sum_{i = 1, i \neq j}^n K_{bt}(\rhat_i) g_{bt}(\rhat_i) g^T_{bt}(\rhat_i), \quad  W_{2, it, -j}(X^n,\rhat; b) = e_2^T\widehat{Q}_{bt, -j}^{-1} b^{-1} K_{bt}(\rhat_i) g_{bt}(\rhat_i), \\
    & \widehat\varphi(Z_i, Z_j) = W_{2, it, -j}(X^n,\rhat; b)Y_i \times W_{jt}(X^n,\rhat; h)(A_j - t).
\end{align*}
The estimator considered is
\begin{align*}
    \widehat{m}_{\text{lcpi}}(t; h, b) = \mhat_{\text{lpi}}(t; h) - \frac{1}{n(n-1)}\mathop{\sum\sum}_{1 \leq i \neq j \leq n} \widehat\varphi(Z_i, Z_j),
\end{align*}
which satisfies the following error decomposition.
\medskip

\begin{proposition}\label{proposition:cpi_local}
In addition to Assumption \ref{main:assumption}, assume that $t$ is an interior point, $n \min(h, b) \to \infty$ and  $\|r-\rhat\|_\infty = o(\min(h, b))$. Then it holds that:
\begin{align*}
    & \widehat{m}_{\text{lcpi}}(t; h, b) - m(t)- \frac{1}{2} m''(t) h^2 \int K(u)u^2 du\\
    & = \frac{1}{f_r(t)} \cdot \left[\frac{1}{n} \sum_{i = 1}^n K_{ht}(r_i) \{\epsilon_i + \rho_i - m'(t) (A_i - r_i)\}\right] \{1 + o_{\Pb}(1)\} + o(h^2)  \\
   & \hphantom{=} + O_{\Pb}\bigg(\frac{\|\rhat - r\|_\infty}{h} \cdot \left\{\frac{1}{\sqrt{nh}} + 
\sup_{t_1,t_2}|\E(\rho\mid\rhat = t_1,r = t_2, D^n) |\right\}+\frac{1}{\sqrt{n^2 h b^3}} \\
& \hphantom{\quad\quad = O_{\Pb}\bigg(} + \frac{\|\rhat - r\|_\infty + \sup_t |\E(\rho\mid \rhat = t, D^n)| + b^2}{b \sqrt{n h}} \\
    &\hphantom{\quad\quad = O_{\Pb}\bigg(}+ \frac{\|\rhat - r\|_\infty}{\sqrt{nb^3}} + b \cdot \|\rhat - r\|_\infty + \frac{\|\rhat - r\|_\infty\{\|\rhat - r\|_\infty + \sup_t |\E(\rho\mid \rhat = t, D^n)|\}}{b}\bigg).
\end{align*}
\end{proposition}
Relative to Proposition \ref{proposition:bc_local}, the error decomposition from Proposition \ref{proposition:cpi_local} presents two key differences: (i) the leading CLT term does not involve $P_n\{K_{ht}'(r)(A - r) \rho \}$, which is of order $O_{\Pb}((nh^3)^{-1/2})$ when $\rho \neq 0$; (ii) the nuisance bias terms involve both the bandwidth $h$ used in the second-stage regression and the bandwidth $b$ used for estimating the derivative $m'(t)$. Throughout, we assume Assumption \ref{ass:rho} is satisfied.
If $\rho \neq 0$, choosing $b_* \asymp \max(h, \|\rhat - r\|_\infty^{1/2})$, the rate simplifies to
    \begin{align*}
        \widehat{m}_{\text{lcpi}}(t; h, b_*) - m(t) & = O_\Pb\left(h^2 + \frac{1}{\sqrt{nh}} + b_*^{-1}\|\rhat - r\|^2_\infty  +\frac{\|\rhat - r\|_\infty^2}{h} \right).
    \end{align*}
    In this light, if $\|\rhat - r\|_\infty \gtrsim n^{-3/10}$, choosing $h \asymp \|\rhat - r\|_\infty^{2/3}$ yields the rate
    \begin{align*}
        \widehat{m}_{\text{lcpi}}(t; h, b) - m(t) & = O_\Pb\left(\frac{1}{\sqrt{n} \|\rhat - r\|_\infty^{1/3}} + \|\rhat - r\|^{4/3}_\infty \right),
    \end{align*}
    while, if $\|\rhat - r\|_\infty \lesssim n^{-3/10}$, choosing $h \asymp b \asymp n^{-1/5}$ yields the oracle rate $n^{-2/5}$. If $\rho =0$, the oracle rate $n^{-2/5}$ is obtained by choosing $h\asymp n^{-1/5}$ if $\|r-\rhat\|_\infty\lesssim n^{-4/15}$. 

Finally, if either (i) $\|\rhat - r\|_\infty = o_{\Pb}( n^{-3/10})$ or (ii) $\|r-\rhat\|_\infty = o_{\Pb}(n^{-4/15})$ and $\rho = 0$, then we have
    \begin{align*}
            \sqrt{nh}\left\{\widehat{m}_{\text{lcpi}}(t;h,b)-m(t) -\frac 12 m''(t)h^2\int K(u)u^2du \right\} &\xrightarrow{d} N\left(0,\sigma^2(t)\right),
        \end{align*}
    where $ \sigma^2(t) = f_r(t)^{-1} \cdot \int K(u)^2 du  \cdot \E\left[\left\{\epsilon+\rho-m'(t)(A-r)\right\}^2\mid r(X)  =t\right]$.
Table~\ref{tab:lpoly_rates} summarizes the discussion above by restating the conditions derived on the first-stage error under which a CLT applies. Notice that, in the likely scenario where $\rho \neq 0$, our analysis for $\widehat{m}_{\text{lcpi}}(t;h,b)$ requires $\|\widehat{r} - r\|_\infty = o_{\Pb}(n^{-3/10})$ for inference at rate $n^{-2/5}$ while that for $\widehat{m}_{\text{lbc}}(t; h)$ requires the weaker condition $\|\widehat{r} - r\|_\infty = o_{\Pb}(n^{-3/14})$ for inference at the slower rate $n^{-2/7}$.  

\begin{table}[ht]
\centering
\caption{Summary of the conditions for inference derived from Propositions \ref{proposition:bc_local} and \ref{proposition:cpi_local}.}
\label{tab:lpoly_rates}
\begin{tabular}{lcccc}
\toprule
& \multicolumn{2}{c}{$\widehat{m}_{\text{lcpi}}(t;h,b)$} & \multicolumn{2}{c}{$\widehat{m}_{\text{lbc}}(t; h)$} \\
\cmidrule(lr){2-3} \cmidrule(lr){4-5}
& $\rho \neq 0$ & $\rho = 0$ & $\rho \neq 0$ & $\rho = 0$ \\
\midrule

\textbf{Achievable rate}
& $n^{-2/5}$ & $n^{-2/5}$
& $n^{-2/7}$ & $n^{-2/5}$ \\[6pt]
\textbf{Standard error scale}
& $(nh)^{-1/2}$ & $(nh)^{-1/2}$
& $(nh^3)^{-1/2}$ & $(nh)^{-1/2}$ \\[6pt]
\textbf{Condition on $\|\hat{r} - r\|_\infty$}
& $o_{\Pb}(n^{-3/10})$ & $o_{\Pb}(n^{-4/15})$
& $o_{\Pb}(n^{-3/14})$ & $o_{\Pb}(n^{-3/10})$ \\
\bottomrule
\end{tabular}

\end{table}

\section{Estimators based on sieves} \label{sec:sieves}
In this section, we consider estimating the projection of $m(r)$ onto the space spanned by $\Phi_k(r)$, which is a dictionary of $k$ basis functions. That is, the target is $\beta_k$ defined as
\begin{align*}
\beta_k = Q_{\Phi,k}(r)^{-1}\E\{\Phi_k(r)Y\}, \quad \text{ where } \quad Q_{\Phi,k}(r) :=\Pb\{\Phi_k(r)\Phi_k(r)^\intercal\}.
\end{align*}
We also define the corresponding empirical Gram matrix $\widehat Q_{\Phi,k}(r) := \Pn\{\Phi_k(r)\Phi_k(r)^\intercal\}$.

In deriving the results below, we allow for the possibility that $k$ increases with the sample size $n$, yielding an increasingly accurate approximation to the regression function $m(r)$. We also allow for the misspecification of the model for $m(r)$; this regime is natural when considering only a fixed number of basis terms. 

Let $\Delta_k(t) = m(t) - \Phi_k(t)^\intercal \beta_k$ and $\dot\Delta_{ k}(t) = m'(t)-\dot\Phi_k(t)^\intercal \beta_k $. By Taylor's expansion, recalling that $\rho \equiv \mu - m(r)$, we have, for some intermediate value $\overline{r}$ between $r$ and $\rhat$
\begin{align*}
    \mu(x) = \rho + \Delta_{k}(\rhat) + \Phi_k(\rhat)^\intercal\beta_k + m'(\rhat)(r - \rhat) + \frac 12m''(\overline{r})(r - \rhat)^2.
\end{align*}
Let us first consider the plug-in estimator of $\beta_k$, denoted $\widehat\beta_{\rm spi}$, and with corresponding estimator $\mhat_{\rm spi}(t; k) = \Phi_k(t)^\intercal \widehat\beta_{\rm spi, k}$. We proceed as in the analysis from Theorem 1 in \cite{mammen2012nonparametric}. We have
\begin{align}\label{eq:plugin_expansion_sieves}
\widehat\beta_{\text{spi}, k} - \beta_k & = \widehat{Q}^{-1}_{\Phi, k}(\rhat) \cdot \Pn\{\Phi_k(\rhat) Y\} - \beta_k \nonumber \\
& =  \widehat{Q}^{-1}_{\Phi, k}(\rhat) \cdot \Pn\left[\Phi_k(\rhat) \left\{m'(\rhat)(r - \rhat) + \rho + \Delta_k(\rhat) + \frac 12m''(\overline{r})(r-\rhat)^2 + \epsilon \right\}\right].
\end{align}
As expected, $\widehat\beta_{\text{spi}, k} - \beta_k$, and thus also $\mhat_{\rm spi}(t; k) - m(t)$, presents two leading terms: one involving the projection of $m'(\rhat)(r- \rhat)$ and one involving the projection of $\rho$. Similarly to the local smoothing setting, we therefore consider two alternative estimators to $\mhat_{\rm spi}(t; k)$. The first one, denoted $\mhat_{\rm sbc}(t; k)$, is based on the influence function of the projection parameter $\beta_k$ treating the number of basis elements $k$ as fixed (in the same spirit as $\mhat_{\rm lbc}(t; h)$). The second one, denoted $\mhat_{\rm scpi}(t; k, q)$, subtracts off an estimate of the projection of $m'(\rhat)(r-\rhat)$ onto the space spanned by $\Phi_k(\rhat)$ from the plug-in estimator $\mhat_{\rm spi}(t; k)$ (in the same spirit as $\mhat_{\rm lcpi}(t)$). The derivative term $m'(\rhat)$ is estimated as $\mhat'(\rhat; q) = \dot\Phi_q(\rhat)^\intercal \widehat\beta_{\text{spi}, q}$. Similarly to $\mhat_{\rm lcpi}(t; k, q)$, we estimate the correction term in $\mhat_{\rm scpi}(t; k, q)$ via a second-order $U$-statistic.

As estimating $m(t)$ based on sieves is a ``global method" approximating the target function over its entire domain, we derive bounds on $\|\widehat\beta_{\text{sbc}, k} - \beta_k\|_2$ and $\|\widehat\beta_{\text{scpi}, k, q} - \beta_k\|_2$. In turn, under mild conditions, these bounds translate into bounds on $\|\Phi_k(r)^\intercal (\widehat\beta_{\text{sbc}, k} - \beta_k)\|_{2, \Pb}$ and $\|\Phi_k(r)^\intercal (\widehat\beta_{\text{scpi}, k, q} - \beta_k)\|_{2, \Pb}$, respectively. Adding the approximation error $\Delta_{k}(t)$ to these bounds would then yield bounds on the $L_2$ norm of $\mhat_{\rm sbc}(r; k) - m(r)$ and $\mhat_{\rm scpi}(r; k, q) - m(r)$.

Our bounds on the norm of $\widehat\beta_{\text{sbc}, k} - \beta_k$ and $\widehat\beta_{\text{scpi}, k, q} - \beta_k$ depend on the operator norm of Gram matrices involving first and second derivatives of the components of $\Phi_k$. To better quantify the order of magnitude of the bounds, we instantiate these matrices taking the cosine basis as an example. Let $\phi_0(u) = 1$, $\dot\phi_0(u) = \ddot\phi_0(u) = 0$ and, for $j \geq 1$, $\phi_j(u) = \sqrt{2} \cos(\pi j u)$, $\dot\phi_j(u) = - \sqrt{2} \pi j \sin(\pi j u)$ and $\ddot\phi_j(u) = -\sqrt{2}\pi^2 j^2 \cos(\pi j u)$. Define $\Phi_k(u) = \{\phi_0(u), \ldots, \phi_{k-1}(u)\}^\intercal$ and $\dot\Phi_k$ and $\ddot\Phi_k$ analogously. In this respect, assuming $r$ is uniformly distributed, we have
\begin{align*}
    & \int_0^1 \Phi_k(u)\Phi_k^\intercal(u) du = I_{k \times k} \implies \left\|Q_{\Phi, k}(r) \right\|_{\rm op} = 1, \\
    & \int_0^1 \dot\Phi_k(u)\dot\Phi_k(u)^\intercal du = \text{diag}(0, \pi^2, 4\pi^2, \ldots, (k-1)^2 \pi^2) \implies \left\|Q_{\dot\Phi, k}(r)\right\|_{\rm op} \asymp k^2, \\
    & \int_0^1 \ddot\Phi_k(u)\ddot\Phi_k(u)^\intercal du = \text{diag}(0, \pi^4, 16\pi^4, \ldots, (k-1)^4 \pi^4) \implies \left\|Q_{\ddot\Phi, k}(r)\right\|_{\rm op} \asymp k^4.
\end{align*}

We assume that the Gram matrix $Q_{\Phi,k}(r)$ is well-conditioned and invertible uniformly in $k$, which is commonly imposed in the sieve literature; see, for example, \cite{belloni2015some} and \cite{semenova2017debiased}.

\medskip

\begin{assumption}\label{assump:sieve_eigen}
The eigenvalues of $Q_{\Phi,k}(r)$ are bounded above and away from zero uniformly over $k$.
\end{assumption}

For the cosine basis, we have
\begin{align*}
\| Q_{\dot\Phi, k}(r)\|_{\rm op} = \left\| 2\pi^2 \begin{bmatrix} 0 & 0 \\ 0 & V \Lambda_{k}(r)  V \end{bmatrix} \right\|_{\rm op},
\end{align*}
where $V = \text{diag}(1, 2, \ldots, k-1)$ and $\Lambda_{k}$ has $(i, j)$-entry equal to $\E[\sin \{ \pi i r(X)\}\sin\{\pi j r(X)\}]$. Under Assumption \ref{main:assumption}, for any unit-norm $v \in \R^k$, we have
\begin{align*}
    & v^\intercal \Lambda_k(r) v = \int_0^1 \left\{\sum_{j =1}^k v_j \sin(\pi j u) \right\}^2 f_r(u) du \in \left[\frac{1}{2} \inf_u f_r(u), \ \frac{1}{2} \sup_u f_r(u)\right].
\end{align*}
so that $1 \lesssim \lambda_\min(\Lambda_k) \lesssim \lambda_\max(\Lambda_k) \lesssim 1$.
Thus, $ \| Q_{\dot\Phi, k}(r)\|_{\rm op} \asymp k^2$. By direct calculation, we also have $\|Q_{\ddot\Phi, k}(r)\|_{\rm op} \asymp k^4$.

In deriving the bounds below, we rely repeatedly on first and second order Taylor expansions applied component wise. Based on an integral representation of Taylor's remainders, the upper bound on the error then also depends on the operator norms of 
\begin{align*}
    & \widetilde{Q}_{\Phi,k}(r,\rhat) = \Pb\left\{\int_0^1 \Phi_k(\rhat + u(r - \rhat)) \Phi_k(\rhat + u(r - \rhat))^\intercal du \right\}\\
    & \widetilde{Q}_{\dot\Phi,k}(r,\rhat) = \Pb\left\{\int_0^1 \dot\Phi_k(\rhat + u(r - \rhat)) \dot\Phi_k(\rhat + u(r - \rhat))^\intercal du \right\} \\
    & \widetilde{Q}_{\ddot\Phi,k}(r,\rhat) = \Pb\left\{ \int_0^1 \ddot\Phi_k(\rhat + u(r - \rhat)) \ddot\Phi_k(\rhat + u(r - \rhat))^\intercal du \right\}.
\end{align*}
We leave the bounds with explicit dependence on the operator norms of the three matrices above. For many commonly employed bases, we expect these norms to be of orders $O(1)$, $O(k^2)$, and $O(k^4)$, respectively, under mild conditions. For example, these orders hold for the cosine basis if, conditional on $D^n$, the random variable
$S=r(X)+U\{\widehat r(X)-r(X)\}$, where $U\sim\operatorname{Unif}[0,1]$ is independent of $(X,D^n)$, admits a density that is uniformly bounded above and away from zero. Finally, following the notation in \cite{belloni2015some}, we define $\xi_k = \sup_{u} \|\Phi_k(u)\|_2$. We refer to their Section 3 for a description of available bounds on $\xi_k$ depending on the choice of the basis vector. When the basis is bounded, such as the cosine basis, one has $\xi_k \lesssim \sqrt{k}$. We also define $\eta_k = \sup_{u} \|\dot\Phi_k(u)\|_2$ and $\zeta_k = \sup_u\|\ddot\Phi_k(u)\|_2$. For example, for the cosine basis, we have $\eta_k \lesssim k^{3/2}$ and $\zeta_k \lesssim k^{5/2}$. 
\subsection{Influence function-based estimator of the finite-dimensional approximation}
We begin by bounding the error of $\widehat\beta_{\text{sbc}, k} = \{\Pn \varphi_2(Z, \rhat)\}^{-1} \Pn \varphi_1(Z,\rhat,\muhat)$, where
\begin{align*}
    & \varphi_1(Z,r,\mu) = \dot\Phi_k(r)(A - r) \mu + \Phi_k(r) Y, \\
    & \varphi_2(Z,r) = \left\{\dot\Phi_k(r) \Phi_k(r)^\intercal + \Phi_k(r) \dot\Phi_k(r)^\intercal\right\}(A - r) + \Phi_k(r)\Phi_k^\intercal(r).
\end{align*}
The estimator $\widehat\beta_{\text{sbc}, k}$ is the influence function-based estimator of $\beta_k$ when $k$ is fixed and the model for the data generating distribution is nonparametric. In the following proposition, we bound the norm of $\widehat\beta_{\text{sbc}, k} - \beta_k$, allowing for both the possibility of growing $k$ and of a misspecified model $\Phi_k^\intercal \beta$. 

\medskip
\begin{proposition}\label{proposition:bc_seive}
Suppose Assumptions \ref{main:assumption} and \ref{assump:sieve_eigen} hold and that $\|\dot\Phi_k^\intercal\beta_k\|_\infty \lesssim 1$ and $\|\ddot\Phi_k^\intercal \beta_k\|_\infty \lesssim 1$. Further assume that

\begin{enumerate}
\item 
$\begin{aligned}[t]
\log k \cdot \Big(
    \eta_k \xi_k + \xi_k^2 + \eta_k^2 
    + \xi_k^2 \|Q_{\dot\Phi,k}(r)\|_{\rm op}
\Big) = o(n).
\end{aligned}$

\item 
$\begin{aligned}[t]
&\|\widetilde{Q}_{\dot\Phi,k}(r,\rhat)\|_{\rm op}^{1/2}
\cdot \min\Big\{
    \|r-\rhat\|_{\infty}, \;
    \xi_k \|r-\rhat\|_{2,\Pb}
\Big\} \\
&\quad
+ \min\Big\{
    \|r-\rhat\|_\infty^2 \|\widetilde Q_{\dot\Phi,k}(r,\rhat)\|_{\rm op}, \;
    \eta_k^2 \|r-\rhat\|_{2,\Pb}^2
\Big\}
= o_{\Pb}(1).
\end{aligned}$

\item 
$\begin{aligned}[t]
&\|\widetilde{Q}_{\ddot\Phi,k}(r,\rhat)\|_{\rm op}^{1/2}
\cdot \min\Big\{
    \|r-\rhat\|_{\infty} \|Q_{\dot\Phi,k}(r)\|_{\rm op}^{1/2}, \;
    \eta_k \|r-\rhat\|_{2,\Pb}
\Big\} \\
&\quad
+ \min\Big\{
    \|r-\rhat\|_\infty^2 \|\widetilde Q_{\ddot\Phi,k}(r,\rhat)\|_{\rm op}, \;
    \zeta_k^2 \|r-\rhat\|_{2,\Pb}^2
\Big\}
= o_{\Pb}\big(\|Q_{\dot\Phi,k}(r)\|_{\rm op}\big).
\end{aligned}$
\item $\begin{aligned}
    & \min\left[
        \|\rhat - r\|_\infty^2 \{
            \|\widetilde{Q}_{\dot\Phi,k}(r,\rhat)\|_{\rm op}
            + 
              \|\widetilde{Q}_{\ddot\Phi,k}(r,\rhat)\|_{\rm op}^{1/2} \|\widetilde{Q}_{\Phi,k}(r,\rhat)\|_{\rm op}^{1/2}\},\right.  \\
        & \hphantom{\min\left[
        \right.} \quad \left. \vphantom{\min\left[
        \|\rhat - r\|_\infty^2 \{
            \|\widetilde{Q}_{\dot\Phi,k}(r,\rhat)\|_{\rm op}
            + 
              \|\widetilde{Q}_{\ddot\Phi,k}(r,\rhat)\|_{\rm op}^{1/2} \|\widetilde{Q}_{\Phi,k}(r,\rhat)\|_{\rm op}^{1/2}\},\right.} \|r-\rhat\|_{2,\Pb}^2(\eta_k^2+\xi_k\zeta_k)
    \right]
 = o_{\Pb}(1)
\end{aligned}$
\end{enumerate}

Then the following statements hold:
\begin{enumerate}
    \item The estimation error in $\beta$ is bounded as
\begin{align*}
    \|\widehat\beta_{\rm sbc, k} - \beta_k\|_2 = O_\Pb\left(\frac{\xi_k+\eta_k\sqrt{\E\{\rho^2(A-r)^2\}}}{\sqrt n}+ \|R_n\|_2\right);
\end{align*}
    \item For any $k$-dim unit vector $\alpha$, the estimator is approximately linear:
\begin{align*}
& \alpha^\intercal\left(\widehat\beta_{\rm sbc, k} - \beta_k \right)\\
&= \alpha^\intercal Q_{\Phi, k}(r)^{-1}
\underbrace{\Pn\left[\Phi_k(r)\left\{\Delta_k(r) + \epsilon - \dot\Phi_k(r)^\intercal\beta_k (A - r) + \rho\right\} + \dot\Phi_k(r)(A - r)\rho\right]}_{:= T_n}\\
& \hphantom{=} + \alpha^\intercal (R_n + S_n)
\end{align*}
\end{enumerate}
where $S_n := \left[ \{\Pn \varphi_2(Z, \rhat)\}^{-1} - Q_{\Phi, k}(r)^{-1}\right] T_n $ with
\begin{align*}
& \|S_n\|_2 \\
& = O_\Pb\left(\left\{\frac{\xi_k+\eta_k\sqrt{\E\{\rho^2(A-r)^2\}}}{\sqrt n} \right\} \right. \\
& \hphantom{O_\Pb\left(\left\{\right.\right.} \quad \cdot\left[\vphantom{\sqrt{\frac{\xi_k^2 \log k}{n}}} \min\left\{
        \|\rhat - r\|_\infty^2 
        \Big(
            \|\widetilde{Q}_{\dot\Phi,k}(r,\rhat)\|_{\rm op} 
            + \|\widetilde{Q}_{\Phi,k}(r,\rhat)\|_{\rm op}^{1/2}
              \|\widetilde{Q}_{\ddot\Phi,k}(r,\rhat)\|_{\rm op}^{1/2}\Big), \|r-\rhat\|_{2,\Pb}^2(\eta_k^2+\xi_k\zeta_k) \right\} \right. \\
& \hphantom{O_\Pb\left(\left\{\right.\right.} \left. \left. \qquad \ + \frac{(\xi_k \eta_k + \xi_k^2)\log k}{n} + \sqrt{\frac{\left\{\eta_k^2
    + \xi_k^2 (1 + \|Q_{\dot\Phi, k}(\rhat)\|_{\rm op})\right\}\log k}{n}} \right] \right) \\
& \|R_n\|_2=  O_{\Pb}\Bigg(
  \left\{\eta_k + \xi_k + \zeta_k \cdot \sup_{t_1, t_2} |\E(\rho^2\mid r = t_1,\rhat = t_2, D^n)|^{1/2}\right\} \cdot \frac{\|r - \rhat\|_{2,\Pb}}{\sqrt{n}} \\
    &\hphantom{\|R_n\|_2=O_{\Pb}\Bigg( } + \|r-\rhat\|_{4,\Pb}^2  + \|\widetilde Q_{\dot\Phi,k}(r,\rhat)\|_{\rm op}^{1/2}\|(r-\rhat)\Delta_k(r)\|_{2,\Pb} \\
&\hphantom{\|R_n\|_2=O_{\Pb}\Bigg( } + \| Q_{\dot\Phi,k}(r)\|_{\rm op}^{1/2}\left\{
    \|(r-\rhat)(\mu-\muhat)\|_{2,\Pb}
    + \|(r-\rhat)\Delta_k(\rhat)\|_{2,\Pb}
    + \|r-\rhat\|_{4,\Pb}^2 
\right\} \\
&\hphantom{\|R_n\|_2=  O_{\Pb}\Bigg(} + \frac{\eta_k}{\sqrt n}\left\{
    \|\mu-\muhat\|_{2,\Pb}
    + \|\Delta_k(\rhat)\|_{2,\Pb}
\right\}  + \|\widetilde{Q}_{\ddot\Phi,k}(r,\rhat)\|_{\rm op}^{1/2}
    \|(r-\rhat)^2\E(\rho\mid r,\rhat, D^n)\|_{2,\Pb}
\Bigg).
\end{align*}
\end{proposition}
Before describing the rate and its implications for inference in more detail, we briefly unpack the four conditions needed in the statement. We do so by introducing the following two simplifying assumptions.

\medskip

\begin{assumption} \label{ass:cosine1}
    The following bounds hold: 
    \begin{itemize}
        \item $\xi_k \lesssim \sqrt{k}$, $\eta_k \lesssim k^{3/2}$ and $\zeta_k\lesssim k^{5/2}$;
        \item $\|Q_{\Phi, k}(r)\|_{\rm op} \asymp 1$ and $\|Q_{\dot\Phi, k}(r) \|_{\rm op} \asymp k^2$.
    \end{itemize}
\end{assumption}
\medskip

\begin{assumption}\label{ass:cosine2}
    The following bounds hold:
    \begin{itemize}
        \item $\|\widetilde{Q}_{\Phi,k}(r,\rhat)\|_{\rm op} \ \lesssim 1$, $\|\widetilde{Q}_{\dot\Phi,k}(r,\rhat)\|_{\rm op} \ \lesssim k^2$ and $\|\widetilde{Q}_{\ddot\Phi,k}(r,\rhat)\|_{\rm op} \ \lesssim k^4$.
    \end{itemize}
\end{assumption}

Under Assumption \ref{ass:cosine1}, the first condition reduces to $k^3 (\log k) / n \to 0$. Under Assumptions \ref{ass:cosine1} and \ref{ass:cosine2}, the second, third and fourth conditions reduce to 
\begin{align}\label{eq:rrhat}
    k^2 \min\{\|\rhat - r\|^2_\infty, \ k \|\rhat - r\|_{2, \Pb}^2\} = o_\Pb(1).
\end{align}

Next, based on Proposition \ref{proposition:bc_seive}, we derive sufficient conditions for inference. We introduce the following assumption, which, analogously to Assumption \ref{ass:muhat}, is made solely to ease the exposition.

\medskip

\begin{assumption}\label{ass:muhat2}
It holds that $\|(r-\rhat)(\muhat - \mu)\|_{2,\Pb} \ \lesssim \|r-\rhat\|_{4,\Pb}^2 $ and $\|\mu-\muhat\|_{2,\Pb}\lesssim\|r-\rhat\|_{2,\Pb}$ 
\end{assumption}
The next assumption, akin to Assumption \ref{ass:rho}, plays a more important role and ensures that the nuisance bias is of smaller order even when $\rho$ is nonzero.

\medskip

\begin{assumption}\label{ass:rho2}
The following bounds hold:
\begin{itemize}
    \item $\|\E(\rho \mid \rhat,  D^n)\|_{2,\Pb}\ \lesssim\|\E(\rho \mid \rhat, r, D^n)\|_{2,\Pb}\ \lesssim \|\rhat - r\|_{2,\Pb}$;
    \item $\|(r-\rhat)\E(\rho\mid r,\rhat, D^n)\|_{2,\Pb}\lesssim\|r-\rhat\|_{4,\Pb}^2 $.
\end{itemize}
\end{assumption}
In the results below, we assume that the choice of the basis is appropriate for approximating $m(t)$. 

\begin{assumption}\label{ass:approximation_error_basis}
It holds that $\|\Delta_k\|_\infty \lesssim k^{-s}$.
\end{assumption}
Assumption \ref{ass:approximation_error_basis} holds, for example, when $m(r)$ is H\"{o}lder smooth of order $s$ and $\Phi(r)$ denotes an appropriate basis, such as wavelet or B-spline series (see, e.g., Section 3 in \cite{belloni2015some} and Appendix A in \cite{mcgrath2026nuisance}). Under Assumptions \ref{ass:cosine1}---\ref{ass:approximation_error_basis}, we have
\begin{align*}
    \|R_n\|_2 = O_{\Pb}\left(\sqrt{\frac{k^5}{n}}\|r-\rhat\|_{2,\Pb}+\sqrt{\frac{k^3}{n}}k^{-s}+k\|r-\rhat\|_{4,\Pb}^2 + \frac{\|r-\rhat\|_{2,\Pb}}{k^{s-1}}+k^2\|r-\rhat\|_\infty\|r-\rhat\|_{4,\Pb}^2\right)
\end{align*}
If $s = 2$, choosing $k\asymp n^{1/7}$ yields that $\|\widehat{m}_{\rm sbc, k} - m\|_{2, \Pb} = O_\Pb(n^{-2/7})$ if $\|r-\rhat\|_{4,\Pb} \ = O_{\Pb}(n^{-3/14})$ and $\|r-\rhat\|_\infty\|\rhat - r\|^2_{4, \Pb} = O_\Pb(n^{-4/7})$. If $\rho = 0$, then, with $k \asymp n^{1/5}$, $\|\widehat{m}_{\rm sbc, k} - m\|_{2, \Pb} = O_\Pb(n^{-2/5})$ if $\|\rhat -r \|_{4, \Pb} = O_\Pb(n^{-3/10}).$ Furthermore, define
\begin{align*}
& \Omega =  Q^{-1}_{\Phi,k}(r) \E\left[\left\{\Phi_k(r) \delta_1 + \dot\Phi_k(r)\delta_2\}\{\Phi_k(r) \delta_1 + \dot\Phi_k(r)\delta_2\right\}^\intercal\right]  Q^{-1}_{\Phi,k}(r), \text{ where } \\
& \delta_1 = \Delta_k(r) + \epsilon + \rho - \dot\Phi_k(r)^\intercal\beta_k (A - r) \quad \text{ and } \quad \delta_2 = (A-r)\rho.
\end{align*}
Under the condition that 
\begin{align*}
    |\Phi_k(t)^\intercal (S_n + R_n)| = o_\Pb\left(\sqrt{\frac{\Phi_k(t)^\intercal \Omega \Phi_k(t)}{n}}\right),
\end{align*}
we have
\begin{align*}
     \frac{\sqrt n\{\Phi_k(t)^\intercal(\widehat{\beta}_{\rm sbc,k}-\beta_k)\}}{\sqrt{\Phi_k(t)^\intercal\Omega\Phi_k(t)}}\xrightarrow{d} N\left(0,1 \right).
    \end{align*}
In the most likely scenario that $\rho \neq 0$, it is reasonable to expect that $\Phi_k(t)^\intercal \Omega \Phi_k(t) \asymp k^3$ since the term $\E\{\dot\Phi_k \dot\Phi_k^\intercal (A - r)^2 \rho^2\}$ would dominate. If this is the case, a sufficient condition for the remainder negligibility is that $\|S_n\|_2 + \|R_n\|_2 = o_\Pb(k / \sqrt{n})$. With $s = 2$ and $k \asymp n^{1/7}$, this condition would hold if $\|r - \rhat\|_{4, \Pb} = o_\Pb(n^{-1/4})$ and $\|r-\rhat\|_\infty = o_{\Pb}(n^{-5/28})$. If $\rho = 0$, sufficient conditions to ensure asymptotic normality (at rate $n^{-2/5}$) are $\|\rhat - r\|_{4, \Pb} = o_\Pb(n^{-7/20})$ and $\|\rhat - r\|_\infty = o_\Pb(n^{-1/4})$.
\medskip

 \begin{remark}
    Suppose that one is interested in specifying a finite-dimensional approximation to $m(r)$ so that $k$ is treated as a fixed constant. Then, since $\widehat\beta_{\rm sbc, k}$ solves $\Pn(\varphi_1(Z, \rhat, \muhat) - \varphi_2(Z, \rhat)\widehat\beta_{\rm sbc, k}) = 0$, we have
    \begin{align*}
  \|\Phi_k(r)^\intercal (\widehat\beta_{\rm sbc, k} - \beta_k)\|_{2,\Pb} = O_\Pb\left(\frac{1}{\sqrt{n}} + \|\rhat - r\|_{2, \Pb}\|\muhat - \mu\|_{2, \Pb} + \|r-\rhat\|_{2,\Pb}^2  \right),
    \end{align*}
    which involves only second-order nuisance errors. This is expected since, for finite $k$, $\beta_k$ is a pathwise differentiable parameter. Furthermore, by sample splitting and the conditions that $\|\muhat - \mu\|_{2, \Pb} = o_\Pb(1)$ and
    \begin{align*}
        \|\rhat - r\|_{2, \Pb}\|\muhat - \mu\|_{2, \Pb} + \|r-\rhat\|_{2,\Pb}^2  = o_{\Pb}(n^{-1/2}),
    \end{align*}
    we have 
    \begin{align*}
        \sqrt{n}(\widehat{\beta}_{\rm sbc,k} - \beta_k) \indist N\left(0, Q_{\Phi,k}(r)^{-1}\var\{\varphi_1(Z, r) - \varphi_2(Z, r) \beta_k\}Q_{\Phi,k}(r)^{-1}\right).
    \end{align*}
\end{remark}
\subsection{Direct correction of plug-in bias}
In this section, we study a correction to a plug-in estimator based on sieves that directly subtracts off the term involving the first order bias $m'(\rhat)(r-\rhat)$. We view this estimator as the sieve-based counterpart to $\mhat_{\text{lcpi}}(t; h)$ (Section \ref{sec:dpi_smoothing}):
\begin{align*}
\widehat\beta_{\text{scpi}, k, q} = \widehat\beta_{\text{spi}, k} - \frac{1}{n(n-1)}\mathop{\sum\sum}_{1 \leq i \neq j \leq n} \dot\Phi_q(\rhat_j)^\intercal \widehat{Q}^{-1}_{\Phi, q, -j}(\rhat)\Phi_q(\rhat_i) Y_i \times \widehat{Q}_{\Phi, k}^{-1} \Phi_k(\rhat_j)(A_j - \rhat_j)
\end{align*}
Here, both $k$ and $q$ index the dimension of the approximating sieve space. In particular, the derivative $m'(\rhat)$ is estimated by first projecting $Y$ onto the $q$-dimensional space spanned by $\Phi_q(\rhat)$ and then computing $\mhat_q'(\rhat) = \dot\Phi_q(\rhat)^\intercal \widehat\beta_{\text{spi}, q}$. Notice that (modulo the removal of observation $j$ in computing the derivative estimate), summing over $i$ in the double sum above returns exactly $n^{-1} \sum_{j=1}^n \mhat_q'(\rhat_j) \times \widehat{Q}^{-1}_{\Phi, k}(\rhat) \Phi_k(\rhat_j)(A_j - \rhat_j)$.
We derive the following proposition bounding the error of $\widehat\beta_{\text{scpi}, k, q}$, which is analogous to Proposition \ref{proposition:cpi_local}.
\medskip
\begin{proposition}
    \label{proposition:cpi_sieve}
    Suppose Assumptions \ref{main:assumption} and \ref{assump:sieve_eigen} hold and that $\|\dot\Phi_q^\intercal\beta_q\|_\infty \lesssim 1$ and $\|\ddot\Phi_q^\intercal \beta_q\|_\infty \lesssim 1$. Further assume that
    \begin{enumerate}
        \item $\xi_k^2\log k+ \xi_q^2\log q = o(n)$
        \item$\begin{aligned}[t]
&\|\widetilde{Q}_{\dot\Phi,k}(r,\rhat)\|_{\rm op}^{1/2}
\cdot \min\Big\{
    \|r-\rhat\|_{\infty}, \;
    \xi_k \|r-\rhat\|_{2,\Pb}
\Big\} \\
&\quad
+ \min\Big\{
    \|r-\rhat\|_\infty^2 \|\widetilde Q_{\dot\Phi,k}(r,\rhat)\|_{\rm op}, \;
    \eta_k^2 \|r-\rhat\|_{2,\Pb}^2
\Big\}
= o_{\Pb}(1).
\end{aligned}$
\item 
$\begin{aligned}[t]
&\|\widetilde{Q}_{\dot\Phi,q}(r,\rhat)\|_{\rm op}^{1/2}
\cdot \min\Big\{
    \|r-\rhat\|_{\infty}, \;
    \xi_q \|r-\rhat\|_{2,\Pb}
\Big\} \\
&\quad
+ \min\Big\{
    \|r-\rhat\|_\infty^2 \|\widetilde Q_{\dot\Phi,q}(r,\rhat)\|_{\rm op}, \;
    \eta_q^2 \|r-\rhat\|_{2,\Pb}^2
\Big\}
= o_{\Pb}(1).
\end{aligned}$

\item 
$\begin{aligned}[t]
&\|\widetilde{Q}_{\ddot\Phi,q}(r,\rhat)\|_{\rm op}^{1/2}
\cdot \min\Big\{
    \|r-\rhat\|_{\infty} \|Q_{\dot\Phi,q}(r)\|_{\rm op}^{1/2}, \;
    \eta_q \|r-\rhat\|_{2,\Pb}
\Big\} \\
&\quad
+ \min\Big\{
    \|r-\rhat\|_\infty^2 \|\widetilde Q_{\ddot\Phi,q}(r,\rhat)\|_{\rm op}, \;
    \zeta_q^2 \|r-\rhat\|_{2,\Pb}^2
\Big\}
= o_{\Pb}\big(\|Q_{\dot\Phi,q}(r)\|_{\rm op}\big).
\end{aligned}$

    \end{enumerate}

Then the following statements hold:
\begin{enumerate}
    \item The estimation error in $\beta$ is bounded as
\begin{align*}
    \|\widehat\beta_{\rm scpi,k,q} - \beta_k\|_2 = O_\Pb\left(\frac{\xi_k}{\sqrt n}+ \|R_n\|_2\right);
\end{align*}
    \item For any unit vector $\alpha \in \R^k$, the estimator is approximately linear:
    \begin{align*}
\alpha^\intercal(\widehat\beta_{\rm scpi, k, q} - \beta_k) & = \alpha^\intercal Q_{\Phi, k}(r)^{-1} \underbrace{\Pn \left[ \Phi_k(r)\left\{\Delta_k(r) + \dot\Phi_q(r)^\intercal\beta_q(r - A) + \epsilon + \rho\right\} \right]}_{:= T_n} \\
& \hphantom{=} \quad + \alpha^\intercal (S_n + R_n),
\end{align*}
\end{enumerate}
where $S_n := \{\widehat{Q}_{\Phi, k}(\rhat)^{-1} - Q_{\Phi, k}(r)^{-1}\} T_n$ and 
\begin{align*}
\|S_n\|_2 & = \ O_\Pb\left(\frac{\xi_k}{\sqrt{n}} \cdot \left\{\sqrt{\frac{\xi_k^2 \log k}{n}} +\|\widetilde{Q}_{\dot\Phi,k}(r,\rhat)\|_{\rm op}^{1/2}\cdot\min\left\{\|r-\rhat\|_{\infty}\|Q_{\Phi,k}(r)\|_{\rm op}^{1/2}, \xi_k\|r-\rhat\|_{2,\Pb}\right\} \right. \right. \\
& \qquad \left. \left. \qquad\qquad\qquad+\min\left\{\|r-\rhat\|^2_\infty \|\widetilde Q_{\dot\Phi,k}(r,\rhat)\|_{\rm op},\eta_k^2\|r-\rhat\|_{2,\Pb}^2\right\}\vphantom{\sqrt{\frac{\xi_k^2 \log k}{n}}} \right\} \right) \\
\|R_n\|_2 & = \ O_{\Pb}\Bigg(
\frac{\xi_k + \eta_k}{\sqrt n}\|r-\rhat\|_{2,\Pb} \\
& \hphantom{ = \quad O_{\Pb}\Bigg[\,} + \|\widetilde Q_{\dot\Phi,k}(r,\rhat)\|_{\rm op}^{1/2}\|(r-\rhat)\left\{\Delta_k(r)+\E(\rho|r,\rhat,D^n)\right\}\|_{2,\Pb} \\
&\hphantom{ = \quad O_{\Pb}\Bigg[\,}
+ \left(1 + \frac{\xi_k}{\sqrt n} \right) \left\{\|(r-\rhat)\dot\Delta_q(\rhat)\|_{2,\Pb}+\|\dot\Delta_k\|_\infty\|r-\rhat\|_{2,\Pb}+\|(r-\rhat)^2\|_{2,\Pb}\right\}\\
&\hphantom{ = \quad O_{\Pb}\Bigg[\,}
+ \frac{\xi_k \eta_q}{n} + \frac{\xi_k}{\sqrt{n}} \cdot \|Q_{\dot\Phi, q}(\rhat)\|^{1/2}_{\rm op} \cdot \{\|\Delta_q(\rhat)\|_{2, \Pb} + \|\rhat - r\|_{2, \Pb} + \|E(\rho \mid \rhat, D^n)\|_{2, \Pb}\}\\
&\hphantom{ = \quad O_{\Pb}\Bigg[\,}
+ \frac{\xi_q}{\sqrt{n}} \cdot \min\left\{\|Q_{\dot\Phi, q}(\rhat)\|^{1/2}_{\rm op} \|\rhat - r\|_\infty, \eta_q \|\rhat - r\|_{2, \Pb}\right\}
\\
&\hphantom{ = \quad O_{\Pb}\Bigg[\,}
+ \min \left\{ \|\rhat - r\|_\infty \|Q_{\dot\Phi, q}(\rhat)\|^{1/2}_{\rm op}, \|\rhat - r\|_{2, \Pb} \eta_q\right\} \\
& \left. \vphantom{\frac{\xi_k}{\sqrt{n}}} \hphantom{ = \quad O_{\Pb}\Bigg[\,} \qquad \qquad \cdot \{\|\Delta_q(\rhat)\|_{2, \Pb} + \|\rhat - r\|_{2, \Pb} + \|E(\rho \mid \rhat, D^n)\|_{2, \Pb}\} \right)
\end{align*}
\end{proposition}
The interpretation of the conditions and the result of Proposition \ref{proposition:cpi_sieve} is similar to that of Proposition \ref{proposition:cpi_local}. For the corrected plug-in estimator, to further simplify the discussion for the rates, we impose the following assumption.
\medskip
\begin{assumption}\label{ass:first-step error} The first-step error satisfies that:
\begin{itemize}
    \item $\|r-\rhat\|^2_\infty\lesssim \min\{k,q\}\cdot\|r-\rhat\|_{2,\Pb}^2$
\end{itemize}
\end{assumption}
The purpose of writing the bound on $\|R_n\|_2$ above in terms of $\min\{\|\rhat - r\|_\infty, \ \sqrt q \|\rhat - r\|_{2, \Pb}\}$ is because $L_2$ error rates are available for many function classes and, depending on the function class where $r$ reside, can be much faster than $L_\infty$ rates; for H\"{o}lder-smooth classes, however, the two rates match up to $\log n$ terms \parencite{tsybakov2008introduction}. For example, if $k\asymp n^\alpha$ and $q\asymp n ^{\beta}$ for any fixed $\alpha,\beta>0$, Assumption~\ref{ass:first-step error} is satisfied when r is H\"{o}lder-smooth and estimated in a minimax optimal sense (e.g., by using local polynomials). In particular, under Assumptions \ref{ass:cosine1}, \ref{ass:cosine2} and \ref{ass:first-step error}, the four conditions simplify to, for $v = \max(k, q)$:
\begin{align*}
v \log(v) / n \to 0 \quad \text{ and } \quad v\|r - \rhat\|_\infty = o_\Pb(1).
\end{align*} 
Under Assumption \ref{ass:cosine1}--\ref{ass:first-step error}, and if $\|\dot\Delta_l\|_\infty \lesssim l^{-s + 1}$ for $l \in \{k, q\}$, we have
\begin{align*}
\|R_n\|_2 & = \ O_{\Pb}\Bigg(
\sqrt{\frac{k^3}{n}}\|r-\rhat\|_{2,\Pb} + \frac{\|\rhat - r\|_{2, \Pb}}{k^{s-1}} + \frac{\|\rhat - r\|_{2, \Pb}}{q^{s-1}} + k\|\rhat - r\|_{4, \Pb}^2 + \frac{\sqrt{k q^3}}{n} \\
&\hphantom{ = \quad O_{\Pb}\Bigg[\,}
\left. + \sqrt{\frac{k q^2}{n}}(q^{-s} + \|\rhat - r\|_{2, \Pb}) + q\|\rhat - r\|_\infty\left(\sqrt{\frac{q}{n}}  + q^{-s} + \|\rhat - r\|_{2, \Pb}\right) \right).
\end{align*}
Therefore, if $s = 2$, $k \asymp n^{1/5}$ and one could ensure $q_* \asymp \min\{k, \|r - \rhat\|_{2,\Pb}^{-1/2}\}$, then we have
\begin{align*}
\|R_n\|_2
=
O_{\Pb}\Bigg(
&\frac{\|\widehat r-r\|_{2,\Pb}}{k}
+k\|\widehat r-r\|_{4,\Pb}^2
+\frac{k^2}{n} +\frac{q_*^{3/2}}{\sqrt n}
\|\widehat r-r\|_\infty
+q_*
\|\widehat r-r\|_\infty
\|\widehat r-r\|_{2,\Pb} \\
&+\|\widehat r-r\|_\infty
\max\left\{
\frac{1}{k},
\|\widehat r-r\|_{2,\Pb}^{1/2}
\right\}
\Bigg).
\end{align*}
In this respect, recalling that throughout we have assumed that $\|\rhat - r\|_\infty = o_\Pb(1/k)$, we have $\|\widehat m_{\rm scpi, k, q} - m\|_{2, \Pb} = O_\Pb(n^{-2/5})$ if $\|r - \rhat\|_{4, \Pb} = O_\Pb(n^{-3/10})$ and $\|r - \rhat\|_\infty \|\rhat - r\|^{1/2}_{2, \Pb} = O_\Pb(n^{-2/5})$. If $\rho = 0$, the condition assumed throughout that $\|\rhat -r\|_\infty = o_\Pb(n^{-1/5})$ together with $\|r - \rhat\|_\infty \|\rhat - r\|^{1/2}_{2, \Pb} = O_\Pb(n^{-2/5})$ suffice since the term $k\|\rhat - r\|_{4, \Pb}^2$ in the preceding display can be replaced by $\|\rhat -r\|_{4, \Pb}^2$.

To achieve the asymptotic Gaussian distribution (at rate $\sqrt{k / n} \asymp n^{-2/5}$), a sufficient condition is that $\|S_n\|_2 + \|R_n\|_2 = o_\Pb(n^{-1/2})$, which holds if $\|\rhat - r\|_\infty = o_\Pb(n^{-7/20})$. Under such condition on $\|\rhat - r\|_\infty$, the choice $q \asymp k \asymp n^{1/5}$ also works for inference. In this case, 
\begin{align*}
\frac{\sqrt n\{\Phi_k(t)^\intercal(\widehat{\beta}_{\rm scpi,k,q}-\beta_k)\}}{\sqrt{\Phi_k(t)^\intercal\Omega\Phi_k(t)}}\xrightarrow{d} N\left(0,1 \right),
\end{align*}
where $\Omega = Q^{-1}_{\Phi,k}(r)\cdot \E\left[\Phi_k(r)\Phi_k(r)^\intercal\left\{\Delta_k(r)+\epsilon + \rho -(A-r)\dot\Phi_q(r)^\intercal\beta_q\right\}^2\right]\cdot Q^{-1}_{\Phi,k}(r)$. 

Table~\ref{tab:sieve_rates} summarizes the rate conditions on the first-step estimation error needed to obtain valid inference.
Notice that these two estimators exhibit a similar pattern to that of the local linear. Although $\widehat{m}_{\text{sbc}}(t)$ converges slower than $\widehat{m}_{\text{scpi}}(t)$, it has weaker requirements for the first-step estimation, namely $\|r-\rhat\|_{4,\Pb} = o_{\Pb}(n^{-1/4})$ and $\|r-\rhat\|_\infty = o_{\Pb}(n^{-5/28})$.
\begin{table}[ht]
\centering
\caption{Summary of sufficient conditions for inference derived from Propositions \ref{proposition:bc_seive} and \ref{proposition:cpi_sieve}.}
\label{tab:sieve_rates}
\begin{tabular}{lcccc}
\toprule
& \multicolumn{2}{c}{$\widehat{m}_{\text{scpi}}(t)$} & \multicolumn{2}{c}{$\widehat{m}_{\text{sbc}}(t)$}\\
\cmidrule(lr){2-3} \cmidrule(lr){4-5}
& $\rho \neq 0$ & $\rho = 0$ & $\rho \neq 0$ & $\rho = 0$ \\
\midrule
\textbf{Achievable rate}
& $n^{-2/5}$ & $n^{-2/5}$
& $n^{-2/7}$ & $n^{-2/5}$ \\[6pt]
\textbf{Standard error scale}
& $(k/n)^{1/2}$ & $(k/n)^{1/2}$
& $(k^3/n)^{1/2}$ & $(k/n)^{1/2}$ \\[6pt]

\textbf{Condition on $\|\hat{r}-r\|_{4,\Pb}$}
& -- & --
 & $o_{\Pb}(n^{-1/4})$ & $o_{\Pb}(n^{-7/20})$ \\[6pt]
\textbf{Condition on $\|\hat{r}-r\|_{\infty}$}
& $o_{\Pb}(n^{-7/20})$ & $o_{\Pb}(n^{-7/20})$
& $o_{\Pb}(n^{-5/28})$ &  $o_{\Pb}(n^{-1/4})$ \\
\bottomrule
\end{tabular}

\end{table}

\section{Simulation Studies}\label{sec:sims}
In this section, we explore the finite-sample performance of our proposed estimators in simulations. We consider several different outcome regression designs and estimate the unobserved covariate $r(x)$ via the Lasso (with logistic link function).

\subsection{Setup}
In every scenario, we generate $d$-dimensional covariates $X \sim N(0,\Sigma)$, where $d = 200$ and $\Sigma$ is an \textit{AR(1)} covariance matrix with correlation parameter 0.7, i.e., $\Sigma_{ij}= 0.7^{|i-j|}$. We generate the variable $A$ by sampling from $A \mid X \sim \text{Bern}(r(X))$ and set $r(x) = \expit(\sum_j\beta_jx_j)$ with $\beta_j = (2j)^{-1}$. The outcome variable is generated as $Y= \mu(X)+\epsilon $, where $\epsilon\sim N(0,1)$ is independent noise.

We run 500 replications. In each replication, we split the sample into two sub-samples. The training sample is used to estimate $r(x)$, $\mu(x)$ and $m'(t)$. The remaining inference sample is used to estimate the target $m(t)$. The outcome model $\mu(x)$ is estimated using random forest (\texttt{ranger}) with default parameters. The derivative term $m'(t)$ is estimated using local linear regression (\texttt{lprobust}) with Gaussian kernel and default parameters. We estimate $r(x)$ using the Lasso with regularization parameter $\lambda$ chosen by 5-fold cross-validation. We also consider calibrating $\rhat$ using isotonic regression (\texttt{isoreg}); that is, $\rhat_{\rm cali}$ is obtained by regressing A on $\rhat$ under the monotonicity constraint on the main sample used for inference. All the code is available at \href{}{https://github.com/Jiaqi0987/Estimated-covariate.git}.

We investigate the performance of three different estimators:
\begin{itemize}
    \item $\widehat{m}_{\text{pi}}(t)$: plug-in estimator regressing $Y$ onto $\rhat(X)$;
    \item $\widehat{m}_{\text{bc}}(t)$: influence function-based bias-corrected estimator (Sections 2.1 and 3.1);
    \item $\widehat{m}_{\text{cpi}}(t)$: corrected plug-in estimator (Sections 2.2 and 3.2).
\end{itemize}
For benchmarking, we also report the performance of their oracle counterparts:
\begin{itemize}
    \item $\widehat{m}_{\text{or}}(t)$: regressing $Y$ onto true $r(X)$;
    \item $\widehat{m}_{\text{bc.or}}(t)$: the oracle version of $\widehat{m}_{\text{bc}}(t)$, which uses true regression function $\mu(X)$ while retaining the estimated covariate $\rhat(X)$;
    \item $\widehat{m}_{\text{cpi.or}}(t)$: the oracle version of $\widehat{m}_{\text{cpi}}(t)$, which uses true derivative $m'(t)$ while retaining the estimated covariate $\rhat(X)$.
\end{itemize}
And also the corresponding estimator using calibrated $\rhat_{\rm cali}$:
\begin{itemize}
    \item $\widehat{m}_{\text{pi.cali}}(t)$: the calibrated version of $\widehat{m}_{\text{pi}}(t)$, computed using $\rhat_{\rm cali}(X)$;
    \item $\widehat{m}_{\text{bc.cali}}(t)$: the calibrated version of $\widehat m_{\mathrm{bc}}(t)$, computed using $\widehat r_{\mathrm{cali}}(X)$;
    \item $\widehat{m}_{\text{bc.cali.or}}(t)$:  the calibrated oracle version of $\widehat m_{\mathrm{bc.or}}(t)$, computed using $\widehat r_{\mathrm{cali}}(X)$ and the true regression function $\mu(X)$.
\end{itemize}
Note that, because calibration enforces $\Pn\{g(\rhat_{\rm cali})(A - \rhat_{\mathrm{cali}})\} = 0$ for any function $g$,
the corrected plug-in estimator constructed with $\rhat_{\mathrm{cali}}$ (as implemented in this simulation study) coincides with plug-in estimator based on $\rhat_{\text{cali}}$; we therefore do not list it separately.
Each estimator is implemented based on three different models: (i) local linear regression, ii) regression using B-splines, and (iii) regression using cosine basis. For the local linear implementation, the bandwidth $h$ is selected by \texttt{lprobust} with default parameters, with the exception of the bandwidths entering the bias-corrected estimators, which are chosen by 5-fold cross validation. For the sieve-based estimators, the number of basis terms are chosen by 5-fold cross-validation. 

We consider six designs for $\E(Y \mid X) = \mu(X)$ of the form $\mu(X) = f\{r(X)\} + g(X)$, where $f(t)$ is (i) a sinusoidal function \textbf{(M1)}, (ii) a quasi-linear function \textbf{(M2)}, and (iii) a smooth function with a localized bump \textbf{(M3)}. The function $g(x)$ is either (i) identically zero or (ii) $g(x) = x^T \beta_g$, where $\beta_{g} $ is a fixed vector generated according to $\beta_{g,j}\sim N\left((2j)^{-1},(2j)^{-2}\right)$. The setting with $g(x) = 0$ corresponds to the case where $\rho(X) = \mu(X) - m\{r(X)\} = 0$. As a robustness check, we also considered $g(x) = x_1 x_2$ and the results were essentially unchanged.

\subsection{Results}
\subsubsection{Estimation accuracy}
We evaluate the accuracy of the estimators by averaging their squared errors across a grid of 50 equally spaced evaluation points $t_j\in[0.15, 0.85]$, i.e.,  $\text{Average squared error} = \frac 1 {50}\sum_{j=1}^{50} (\widehat{m}(t_j)-m(t_j))^2$. The points 0.15 and 0.85 correspond roughly to the 0.05- and 0.95-quantile of $r(X)$, respectively. A total sample size of $n = 5000$ is split evenly into a training sample and an inference sample. 

Figures \ref{fig:scenarioII-M1}, \ref{fig:scenarioII-M2}, \ref{fig:scenarioII-M3} report the results for the settings with $g(X) = x^T\beta_g$ (corresponding to $\rho \neq 0$). Each figure consists of 3 plots: (i) the top-left panel displays the true function $m(t)$ along with $\widehat{m}_{\text{pi}}$, $\widehat{m}_{\text{cpi}}$ and $\widehat{m}_{\text{bc}}$ estimates based on local linear smoothing from one replication; (ii) the top-right panel shows the pointwise squared error curves for the estimators based on local linear smoothing, where the gray lines come from individual simulation runs and the red line is their average; and (iii) the box-plots summarize the distribution of the average squared error for each estimator considered. The results show that any form of bias correction yields notable improvements relative to the plug-in estimator. Across all scenarios, calibrating $\widehat{r}$, in conjunction with the plug-in strategy (equivalent to the corrected plug-in strategy in this case) or the bias-correction based on the influence function of the fixed-dimensional approximating target, appears to yield the best results in terms of average squared-error. Finally, when $\rho = 0$, Figures \ref{fig:scenarioI-M1}, \ref{fig:scenarioI-M2}, \ref{fig:scenarioI-M3} show that estimators based on calibrating $\rhat$ have the best and most stable performance. Without calibrating $\rhat$, the forms of bias correction considered did not harm the performance relative to the plug-in estimator but they yield substantial gains mostly in scenario M1 (with the corrected plug-in also showing improvements in Scenario M3). 

\begin{figure}[htbp]
    \centering
    \includegraphics[width=0.8\textwidth]{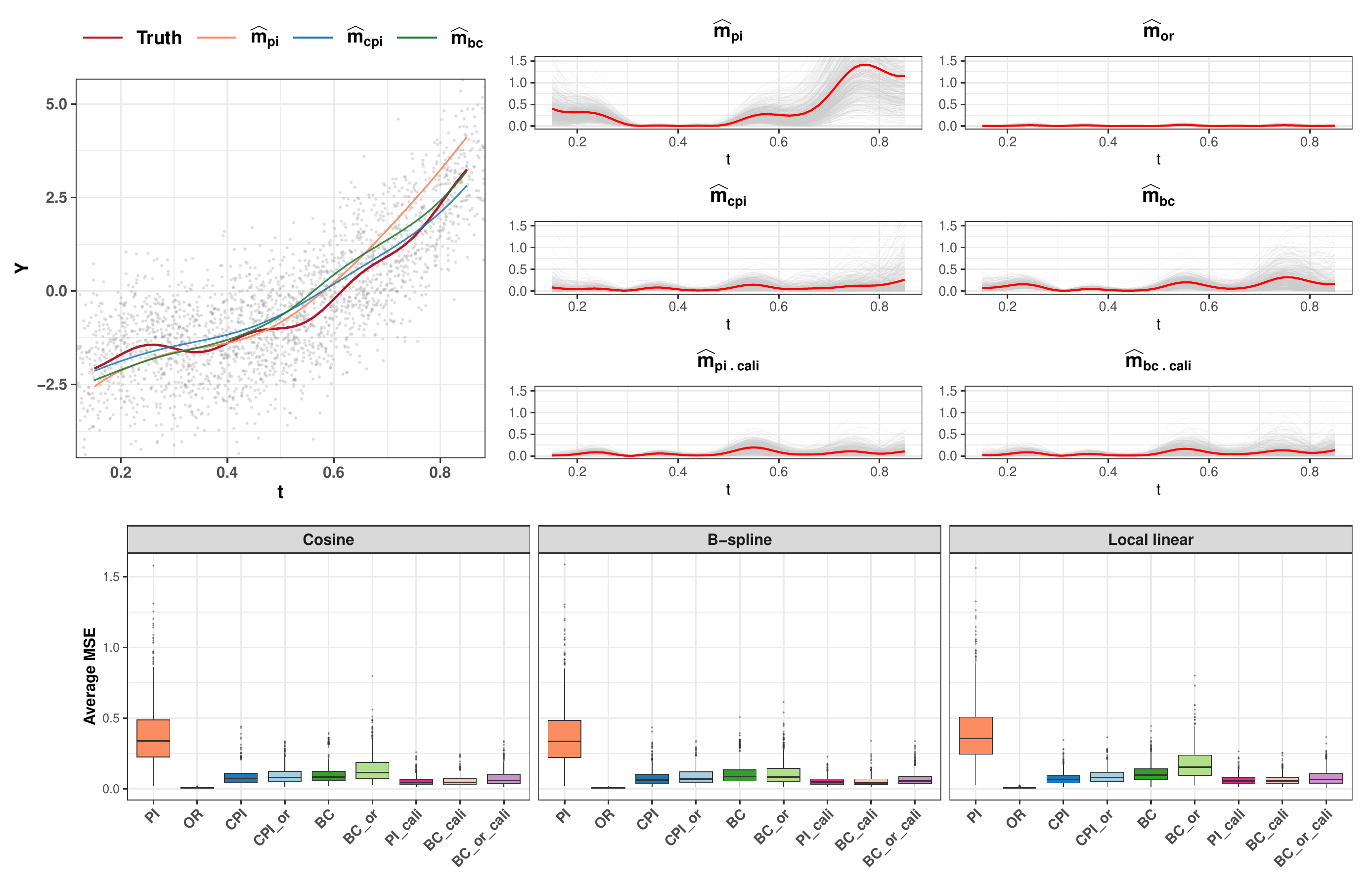}

    \caption{\textbf{M1 (sin–cosine function + linear term in X):} 
    $\mu(X) = 0.15\sin(10\pi r(X)) + \cos(2\pi r(X)) + \beta_g^T X$, 
    where $\beta_g$ is a fixed coefficient vector.
    }
    \label{fig:scenarioII-M1}
\end{figure}

\subsubsection{Coverage and CI length}

We evaluate the coverage of the confidence intervals for the proposed estimators. Because the estimators are nonparametric, the tuning parameters minimizing the mean-square-error (such as cross-validated choices) balance smoothing bias and standard error; the confidence intervals should therefore be interpreted as targeting the corresponding smoothed estimand (for example $\Phi^\intercal_k(t)\beta_k$ for sieves), rather than the true regression function $m(t)$. To make the coverage comparison meaningful, we use a two-stage simulation
procedure. First, for each simulation setting, we run a pilot simulation with
50 replications and record the selected bandwidth or the basis dimension for each estimator.
We then fix the bandwidth or basis dimension at the median value from the pilot study,
separately for each estimator. Using these fixed tuning parameters, we evaluate whether the confidence interval for each estimator covers its corresponding smoothed target.

For our coverage analysis, we fix the inference sample size at $2500$ and vary the training sample size $n_{tr}\in\{2500, 10000\}$, so that the nuisance estimates are increasingly more accurate. For scenario M1 with $\rho \neq 0$, Figures \ref{fig:coverage-ntr} and \ref{fig:cilength-ntr} report the pointwise coverage and the length of intervals, respectively. For all other settings, the results are reported in Appendix B. As expected, as $n_{tr}$ increases, the coverage moves closer to the nominal 95\% level for estimators based on some form of bias correction, while the coverage for plug-in estimators remains far below the target. In some scenarios, we notice a drop in coverage in certain regions (for example:  evaluation points near $t = 0.5$ in Figure \ref{fig:coverage-ntr_S1M1}). This is likely due to the fact that the true function $m(t)$ changes slope there, and the regularization in $\widehat{r}$ leads to more observations falling in that region, thus reducing the standard errors in a difficult landscape for estimation. This is
consistent with the condition that the first-step estimation error should be smaller
relative to the smoothing scale, for example $\|\widehat r-r\|_\infty=o(h)$ for local linear estimator. In regions where the target function varies more rapidly, smaller bandwidths may be needed to reduce smoothing bias, which makes the required first-stage accuracy condition more stringent.

\begin{figure}[htbp]
    \centering
    \includegraphics[width=\textwidth]{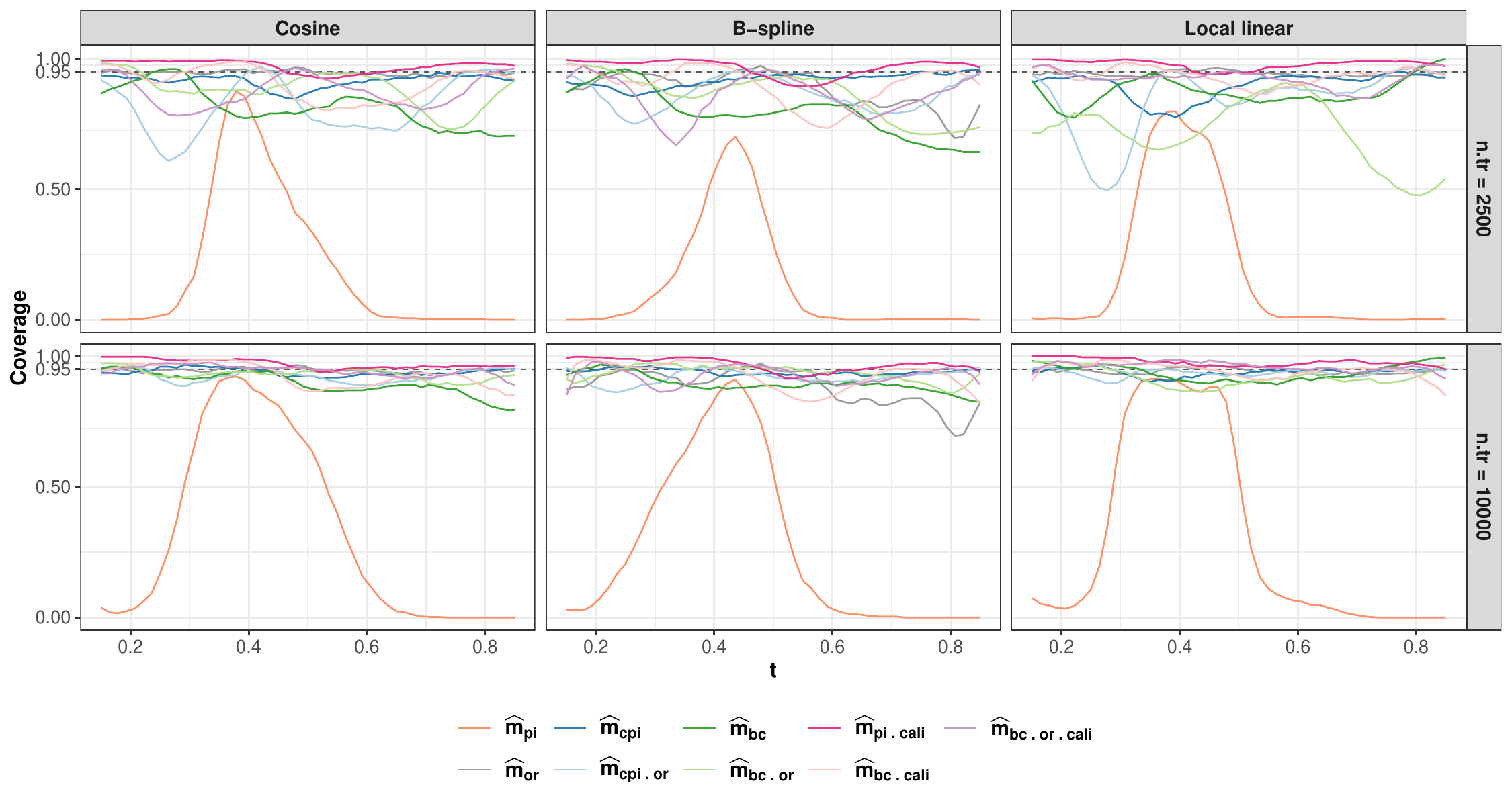}
    \caption{Pointwise coverage relative to the matched same-replication semi-oracle target for \textbf{M1 (sine–cosine function + linear term in $X$):} 
    $\mu(X) = 0.15\sin\{10\pi r(X)\} + \cos\{2\pi r(X)\}+\beta_g^TX$, where $\beta_g$ is a fixed coefficient vector.}
    \label{fig:coverage-ntr}
\end{figure}

\begin{figure}[htbp]
    \centering
    \includegraphics[width=\textwidth]{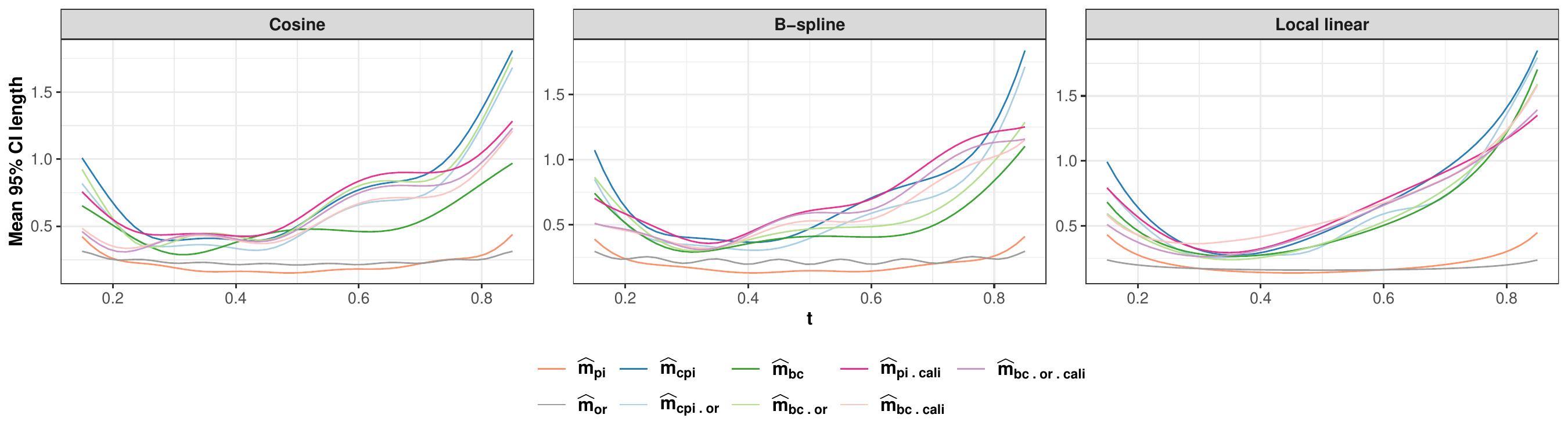}
    \caption{Pointwise confidence interval length for \textbf{M1 (sine–cosine function + linear term in $X$):} 
    $\mu(X) = 0.15\sin\{10\pi r(X)\} + \cos\{2\pi r(X)\}+\beta_g^TX$, where $\beta_g$ is a fixed coefficient vector.  Training sample uses 2500 observations.}
    \label{fig:cilength-ntr}
\end{figure}

\section{Application}\label{sec:application}\label{sec:da}
In this section, we investigate the effect of college completion on cumulative unemployment in early to mid-career. We use data from the National Longitudinal Survey of Youth 1997 (NLSY97), which is a nationally representative sample of 8,984 men and women of age 12--16 as of December 31, 1996. We restrict the sample to those who completed at least a high school degree ($n = 7,753$), those falling within a region of common support ($n = 7,626$) and to cases with no missing values on key variables including unemployment ($n = 6,978$) (but where several precollege variables were imputed) . Table \ref{tab:summary statistic} reports summary statistics for the precollege covariates used in our analysis. Our goal is to evaluate the causal effect of college completion on cumulative unemployment as a function of the true probability of completely college (propensity score). Previous analyses have shown that college completion reduces time spent unemployed over the career, and the effect is greater for those less likely to complete college \parencite{brand2023overcoming}. This finding, termed negative selection, is inconsistent with a rational-behavior model whereby agents decide to complete college based on their expected future career returns. One possible explanation for this finding is that students from advantaged backgrounds may consider college attendance and completion as culturally expected regardless of whether it is advantageous from a future career standpoint. On the contrary, attending college may not be the norm for less advantaged students, and thus it may be justified under more stringent economic gains. Workers from disadvantaged backgrounds who do not complete college face particularly poor labor market outcomes, such that we expect large benefits to degree completion, especially in reducing cumulative unemployment \parencite{brand2023overcoming}. By contrast, students from advantaged backgrounds are likely to have stronger labor market opportunities regardless of whether they complete college, suggesting smaller returns. Overall, the theory suggests that the effect of college completion on unemployment should decrease as the propensity score increases.

Mathematically, we let $A \in \{0, 1\}$ denote college completion and $(Y^1, Y^0)$ the potential cumulative unemployment rates if the unit completes versus does not complete college, respectively. A rich set of background covariates $X \in \R^p$ are observed, including sociodemographics, parents' education and income, family structure, students' high-school achievement, and school characteristics. A summary is presented in Table \ref{tab:summary statistic}; certain variables are constructed, we refer to \cite{brand2023overcoming} for their definition. We invoke (i) no-unmeasured-confounding, $A \ind (Y^1, Y^0) \mid X$, (ii) positivity, $\epsilon \leq r(X) = P(A = 1 \mid X) \leq 1 - \epsilon$ with probability 1 for some $\epsilon > 0$, and (iii) consistency, $A = a \implies Y^a = Y$. Then,
\begin{align*}
    \tau(t) = \E\{Y^1 - Y^0 \mid r(X) = t\} = \E\{\mu_1(X) - \mu_0(X) \mid r(X) = t\}, 
\end{align*}
where $\mu_a(X) = \E(Y \mid A = a, X)$. To estimate $\tau(t)$ using the methods described in the previous sections, we propose following the principle of DR-Learning \parencite{kennedy2020optimal} and regressing the following pseudo-outcome on $r(X)$:
\begin{align*}
    \varphi(Z) = \frac{A}{r(X)}\{Y - \mu_1(X)\} - \frac{1-A}{1 - r(X)}\{Y - \mu_0(X)\} + \mu_1(X) - \mu_0(X).
\end{align*}
As $\E\{\varphi(Z) \mid r(X) = t\} = \tau(t)$, Propositions~\ref{proposition:bc_local}--\ref{proposition:cpi_sieve} apply directly when treating $\varphi(Z)$ as known. In fact, the results translate by viewing  $\varphi(Z)$ as $Y$ and $\E\{\varphi(Z) \mid X\} = \mu_1(X) - \mu_0(X)$ as $\mu(X)$. In practice, $\varphi(Z)$ needs to be estimated, but because $\varphi(Z)$ is an orthogonal signal (i.e., regressing its estimate on $X$ results in estimators of the $X$-conditional CATE that have second-order dependence on the nuisance errors under mild conditions \parencite{kennedy2020optimal, foster2023orthogonal}), we expect the conclusions from Propositions \ref{proposition:bc_local}-\ref{proposition:cpi_sieve} continue to hold in this setting as well.

In \cite{brand2023overcoming}, the propensity score is estimated by a logit regression. The model is based on an iterative procedure considering all possible higher order and interaction terms \parencite{imbens2015causal} to produce a flexible specification. Heterogeneous treatment effects are estimated by plotting the treated and control unemployment outcomes along a continuous representation of the propensity score using local polynomial smoothing and taking the difference in the nonparametric curves \parencite{xie2012estimating}. In other analyses, effects are summarized within propensity score strata and effect heterogeneity is considered across the strata.

In our analysis, we adopt a stacked learning approach to estimate $r$, $\mu_0$ and $\mu_1$. The ensemble model combines predictions from Random Forests (\texttt{ranger}), Lasso with 2-knots, natural-splined continuous covariates and up to two-degree interactions (\texttt{glmnet}), Generalized Additive Models (\texttt{mgcv}), Gradient Boosted Trees (\texttt{gbm}), parametric linear / logistic models (\texttt{glm}), with cross-validated tuning parameters. For the meta learner used for stacking, we employ parametric linear regression for $\mu_0$ and $\mu_1$ and parametric logistic regression for $r$. Figure \ref{fig:fold-wise loss} reports the risk estimates for each learner and across nuisances.
\begin{figure}
    \centering
    \includegraphics[width=0.6\linewidth]{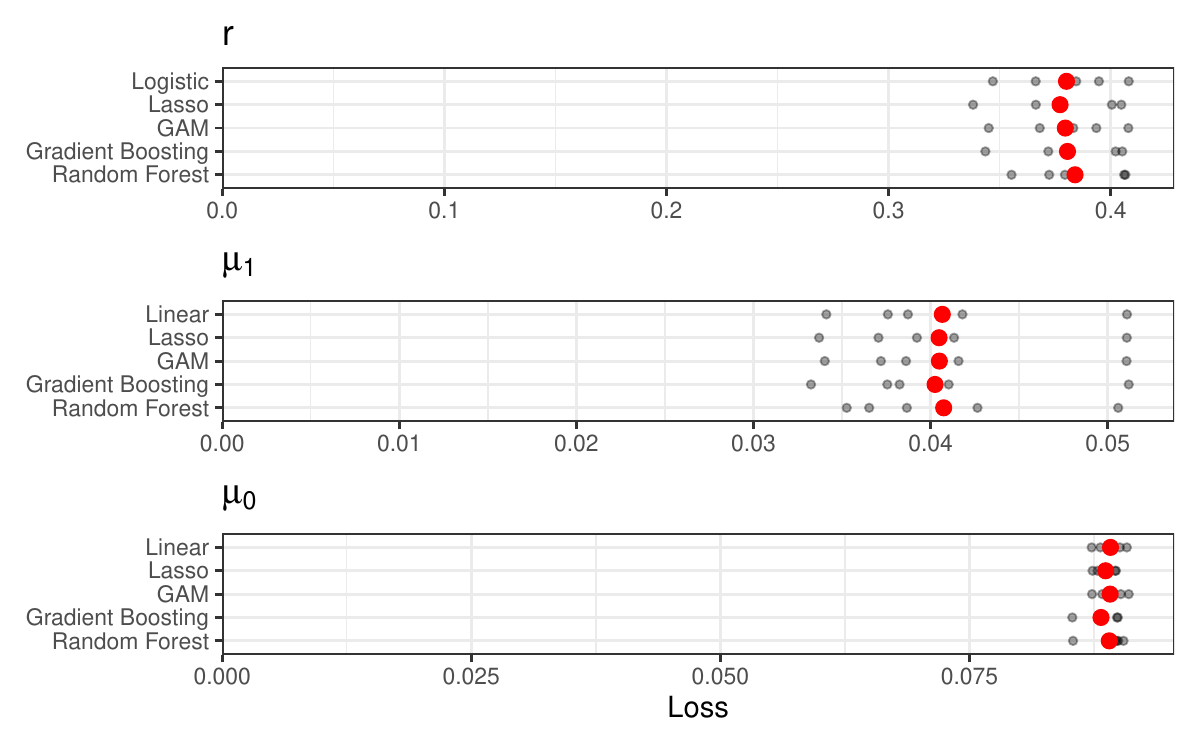}
    \caption{Loss estimates for each learner across nuisances. Gray points represent losses in each fold, and red points indicate the mean loss across folds.}
    \label{fig:fold-wise loss}
\end{figure}

Figure \ref{fig:unempprop effect} compares the estimated conditional effects as a function of the estimated propensity score when using the plug-in, corrected plug-in, and influence-function-based bias-corrected estimators. All estimates are negative (corresponding to a beneficial effect of college completion on unemployment rates). Importantly, the estimates suggest that individuals who are least likely to obtain college education benefit the most from completing college.
\begin{figure}
    \centering
    \includegraphics[width=0.8\linewidth]{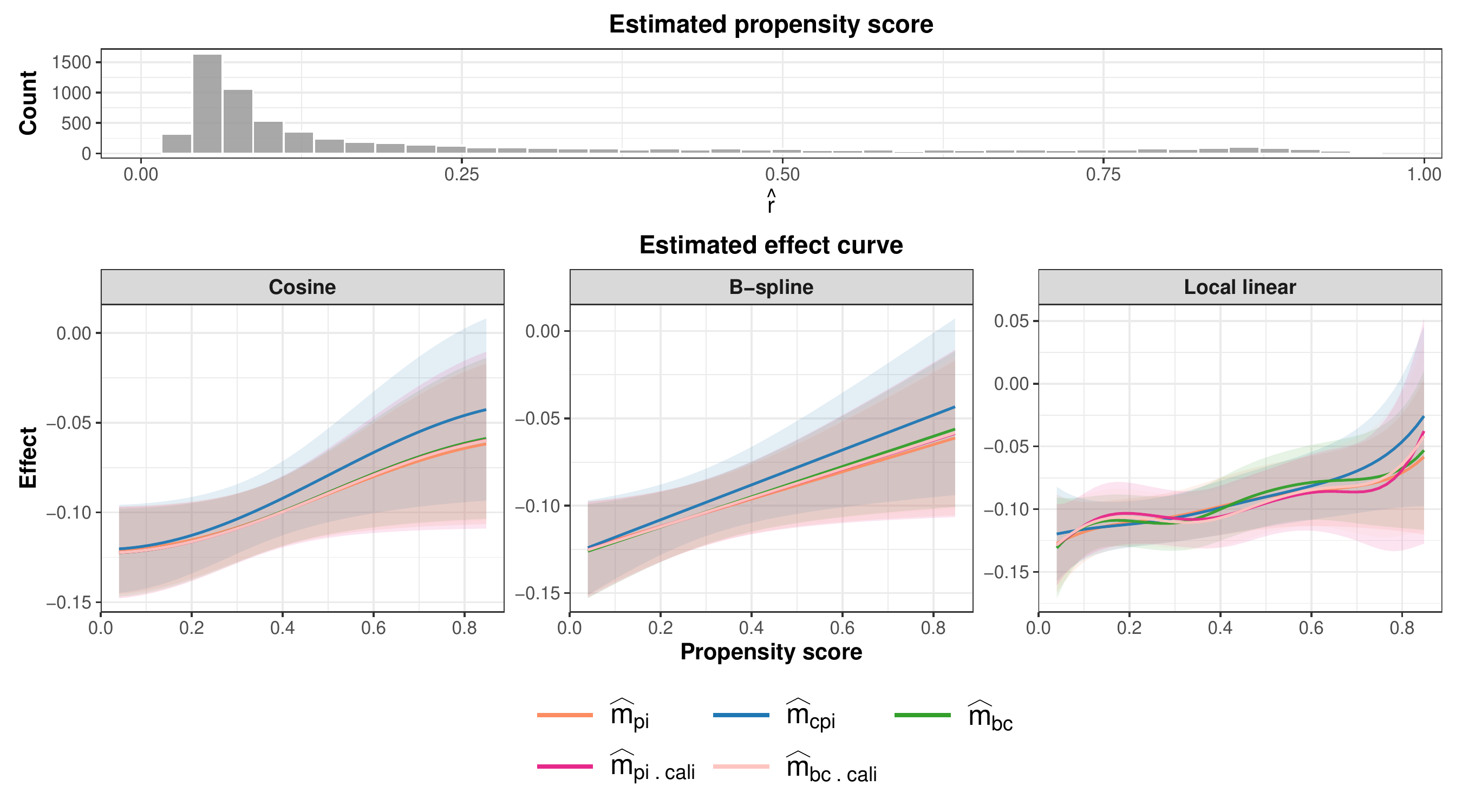}
    \caption{Distribution of propensity score (top) and estimated effect curve of college completion on cumulative unemployment with 95\% confidence band (bottom). }
    \label{fig:unempprop effect}
\end{figure}
\begin{table}[htbp]
\centering
\begin{threeparttable}
\caption{Summary Statistics for Precollege Covariates}
\small
\renewcommand{\arraystretch}{1.25} 

\begin{tabular}{lcccc}
\toprule
& \multicolumn{2}{c}{\textbf{Men} (N = 3486)} & \multicolumn{2}{c}{\textbf{Women} (N = 3492)} \\
\cmidrule(lr){2-3} \cmidrule(lr){4-5}
& Non-college & College
& Non-college & College \\

&Graduate & Graduate & Graduate & Graduate\\
Variables &(N = 2758)& (N = 728) & (N = 2500) & (N = 992)\\
\midrule

\textbf{Race} \\[0.15em]
\quad Black (0/1)       & 0.28 (0.45)& 0.14 (0.35)& 0.32 (0.47)& 0.19 (0.39)\\
\quad Hispanic (0/1)        & 0.22 (0.41)&  0.12 (0.32)& 0.24 (0.43)& 0.13 (0.34)\\
\textbf{Social Background} \\[0.15em]
\quad Father's education (years, 0--20)    & 12.12 (2.84)& 14.58 (2.91)& 11.86 (2.82)& 14.01 (3.01) \\
\quad Mother's education (years, 0--20)   & 12.30 (2.62)& 14.31 (2.44)& 11.98 (2.61)& 13.90 (2.78) \\
\quad Parent's income (\$1000s) & 42.97 (33.34)&71.54 (46.63)& 38.95 (30.80)& 64.95 (43.02)\\
\quad Intact family (0/1) & 0.41 (0.49)&0.72 (0.45)& 0.37 (0.48)& 0.67 (0.47)\\
\quad Number of siblings &4.54 (2.96)& 3.69 (2.53)& 4.72 (2.99)& 3.76 (2.65)\\
\quad Rural residence (0/1)&0.17 (0.38)& 0.18 (0.38)& 0.16 (0.37)& 0.18 (0.39)\\
\quad Southern residence (0/1) &0.33 (0.47) &0.30 (0.46) &0.34 (0.47) &0.32 (0.47)\\
\quad Married by age 18 (0/1)  &0.01 (0.07) & 0.00 (0.04)& 0.03 (0.17)&0.00 (0.05)\\
\quad Had children by age 18 (0/1) &0.01 (0.11) &0.00 (0.04) &0.07 (0.25) &0.01 (0.07)\\
\quad Delinquency  (0--10) & 0.64 (0.48)& 0.51 (0.50)& 0.46 (0.50)&0.33 (0.47) \\
\textbf{Ability and Academics}\\[0.15em]
\quad ASVAB scale (-3--3) &-0.25 (0.57) &0.23 (0.59) & -0.27 (0.62)& 0.16 (0.56)\\
\quad GPA (0--4) & 2.56 (0.71)& 3.33 (0.60)& 2.82 (0.70)& 3.45 (0.54)\\
\quad College-prep (0/1) &0.21 (0.41)& 0.58 (0.49)& 0.25 (0.43)&0.62 (0.48)\\
\textbf{Social-Psychological}\\[0.15em]
\quad Good teachers (0/1)&0.17 (0.38)&0.25 (0.43)& 0.16 (0.37)&0.22 (0.41) \\
\quad Friends aspire college (0/1)& 0.96 (0.19)& 0.99 (0.10)& 0.96 (0.19)& 1.00 (0.06)\\
\quad School safety (0/1)& 0.30 (0.46)&0.45 (0.50)&0.26 (0.44)&0.45 (0.50)\\
\bottomrule
\end{tabular}
\label{tab:summary statistic}
\begin{tablenotes}
\footnotesize
\item \textit{Notes:} Entries report mean values with standard deviations in parentheses. 
\end{tablenotes}
\end{threeparttable}
\end{table}
\section{Conclusion \& practical considerations}
In this work, we proposed and analyzed two ways of debiasing plug-in estimators of a regression on an unobserved, but estimable, covariate. One is based on directly subtracting off an estimate of the plug-in estimator's bias, while the other one is based on the influence function of the parameter, derived by viewing the second-stage tuning parameter (either the bandwidth in local smoothing or the basis dimension in the sieves regression) as fixed. For each approach, we have considered both local smoothing and regression onto a finite-dimensional space of increasing dimension. Recalling that the generic estimand considered is $\E(Y \mid r(X) = t)$, for $r(X) = \E(A \mid X)$, the upper bounds on the estimators' errors that we have derived distinguish between the case where the index $r(X)$ is sufficient, i.e., $\mu(X):= \E(Y \mid X) = \E\{Y \mid r(X)\}$, and the general case where it is not. However, in the application we considered, the general case is much more plausible.

The two approaches differ in several ways, which can be summarized as follows (with some simplifications for easier exposition, please refer to the technical propositions). First, the way we propose directly subtracting off an estimate of the plug-in's bias requires estimating the derivative of the regression function given the unobserved covariate. On the other hand, the influence-function-based estimators do not require estimating the derivative function but they rely on estimating $\mu(X) = \E(Y \mid X)$, which is a standard regression problem since $X$ is observed. Furthermore, while we consider $U$-statistic-based estimators for the former approach (coming from a leave-one-out estimator of the derivative), we find that simply regressing $Y - \widehat{m}'(\rhat)(A - \widehat{r})$ onto $\widehat{r}$ using off-the-shelf software works reliably in our simulations; here $\widehat{m}'(\rhat)$ denotes a plug-in estimator of the derivative function treating $\widehat{r}(x)$ as known, which can too be computed using available software. On the contrary, we found the construction of the influence-function-based estimators in simulations to be more delicate because of the presence of the derivative of either the kernel or the basis functions, leading to additional divisions by the bandwidth or multiplications by the number of basis terms, respectively.

In addition, our theoretical analysis finds that, in the likely case where the index $r(X)$ is not a sufficient statistic, the debiasing approach based on the influence function is unable to return an estimator that can converge as fast as the oracle estimator with access to the true $r(x)$. However, the approach based on directly estimating the plug-in bias can return such estimator even if $\|\widehat{r} - r\|_\infty$ is converging somewhat slower than the oracle rate (with a local linear second-stage regression, our required rate is faster than $n^{-3/10}$ when the oracle rate is $n^{-2/5}$). In simulations, selecting the bandwidth that is optimal for estimating the derivative itself (using off-the-shelf software such as \texttt{lprobust} \parencite{calonico2019nprobust}) leads to a reasonable performance. The influence-function-based estimators, in the general case, converge to the truth at a slower rate than the oracle. Yet, inference can be carried out at these slower regimes under weaker conditions on the accuracy of $\widehat{r}$ relative to those needed for inference, at the faster oracle rate, for the approach based on direct plug-in debiasing.

In the simulations considered, both approaches yielded estimators whose performance was at least as good as, and often much better than, that of the plug-in estimators, albeit worse than that of the oracle estimators. Notably, using a plug-in estimator after calibrating the covariate's estimate, for instance using isotonic regression, resulted in stable and reliable performance. Our theory, however, does not immediately cover calibration of $\widehat{r}(X)$ because the calibrating step occurs on the same sample used to compute the second-stage regression, while our results assume that $\widehat{r}(X)$ is estimated on a separate sample. A careful analysis of the plug-in estimator with calibrated covariates is thus an important avenue for future work. Overall, our analysis suggests that because the estimators considered have different strengths, it is advisable to try both debiasing approaches to assess the robustness of the inferences. To implement the direct debiasing of plug-in estimators, a promising way is to simply regress the outcome on the calibrated covariate.
 \section{Acknowledgments}
JW and MB gratefully acknowledge support from NSF DMS Grant 2413891. MB thanks Stijn Vansteelandt, Georgi Baklicharov and Oliver Dukes (Ghent University) for very insightful conversations. Large language models were used to assist with proofreading and identifying typographical errors during the revision of this manuscript.
\printbibliography
\newpage
\appendix
\section{Estimated outcomes versus estimated covariates}\label{appendix:sec_estimated_outcomes_vs_covariates}
At the surface level, the problem of estimating a regression function with unobserved, but estimable, covariates resembles that of estimating a regression with an unobserved, but estimable, pseudo-outcome. However, our analysis suggests that this is not the case. Here, we briefly outline a conceptual difference between the two settings in a stylized example. Consider the conditional average treatment effect function (identified under no-unmeasured-confounding):
    \begin{align*}
        \theta(t) = \E\{\mu_1(X) - \mu_0(X) \mid V = t\}, \qquad \mu_A(X) = \E(Y \mid A, X), \quad V \subset X.
    \end{align*}
    One key aspect of this problem is that, for any fixed functions $\overline{r}(x)$ and $\overline\mu_A(x)$, letting
    \begin{align*}
        \varphi_a(Z) = \frac{I(A = a)\{Y - \mu_a(X)\}}{r(X)^a\{1 - r(X)\}^{1-a}} + \mu_a(X),
    \end{align*}
   it holds that
    \begin{align*}
        & \E\{\overline\varphi_1(Z) - \overline\varphi_0(Z) \mid V = t\} - \theta(t) \\
        & = \E\left[ \left\{\frac{r(X)}{\overline{r}(X)}-1\right\} \{\mu_1(X) - \overline\mu_1(X)\} - \left\{\frac{1 - r(X)}{1 - \overline{r}(X)} - 1\right\}\{\mu_0(X) - \overline\mu_0(X)\} \mid V = t\right].
    \end{align*}
    This crucial observation then suggests that an estimator regressing the estimated pseudo-outcome $\widehat\varphi(Z)$ onto $V$ can potentially behave like an oracle estimator with access to $\varphi(Z)$ as long as the \textit{product of the errors} for estimating $r(x)$ and $\mu_A(x)$ (measured by some suitable norm) is smaller than the oracle rate. A careful analysis would need to take into consideration the properties of the second-stage regression estimator, but this intuition has been formalized and considerably generalized in recent influential works, including \cite{foster2023orthogonal} and \cite{kennedy2020optimal}. In particular, regressing $\widehat\varphi(Z)$ on $V$, the so-called DR-Learner, was analyzed in \cite{kennedy2020optimal}; we build upon this approach in Section \ref{sec:application} where $\widehat{r}(X)$ plays the role of $V$ here.

    If instead $V$ itself has to be estimated, our results suggest that one should generally not expect the same clean second-order nuisance remainder. To see one obstacle to obtain pure second-order nuisance errors, consider estimating $\E\{\mu_1(X)\}$ versus $\E[K_{ht}\{\mu_1(X)\}]$. The latter parameter would naturally appear when considering a regression onto $\mu_1(X)$ via local linear smoothing, while the former is the treatment-specific mean functional under no-unmeasured-confounding. The uncentered influence function of the former parameter is $\varphi_1(Z)$, and one has
    \begin{align*}
        \E\{\overline\varphi_1(Z)\} - \E\{\mu_1(X)\} = \E\left[ \left\{\frac{r(X)}{\overline{r}(X)}-1\right\} \{\mu_1(X) - \overline\mu_1(X)\}\right],
    \end{align*}
    paving the way again for the existence of an estimator (the so-called augmented-inverse-probability-weighted estimator \parencite{robins1994estimation}) exhibiting a pure second-order nuisance error. On the contrary, for a fixed bandwidth $h$, the (uncentered) influence function of $\E[K_{ht}\{\mu_1(X)\}]$ is
    \begin{align*}
    & \kappa_h(Z) = K'_{ht}\{\mu_1(X)\}\frac{A}{r(X)}\{Y - \mu_1(X)\} + K_{ht}\{\mu_1(X)\},\quad \text{  which satisfies } \\
        & \E\{\overline\kappa_h(Z)\} - \E[K_{ht}\{\mu_1(X)\}] \\
        & = \E\left[ K'_{ht}\{\overline\mu_1(X)\}\left\{\frac{r(X)}{\overline{r}(X)}-1\right\} \{\mu_1(X) - \overline\mu_1(X)\}\right] - \frac{1}{2} \E\left[ K''_{ht}\{\widetilde\mu_1(X)\}\{\mu_1(X) - \overline\mu_1(X)\}^2\right],
    \end{align*}
    for some intermediate value $\widetilde\mu_1(X)$ between $\mu_1(X)$ and $\overline\mu_1(X)$. Because of the presence of $K'_{ht}$ and $K''_{ht}$, by a standard smoothing argument and under mild conditions, these two apparent second-order terms are inflated by $h^{-1}$ and $h^{-2}$, respectively. When $h$ goes to zero, the nuisance errors are thus no longer purely second-order in a bandwidth-free sense. Nevertheless, they may still be asymptotically smaller than $|\E[K_{ht}\{\overline\mu_1(X)\} - K_{ht}\{\mu_1(X)\}]|$; for example, if $\|\mu_1-\overline\mu_1\|_\infty=o(h)$, then the quadratic term $h^{-2}\|\mu_1-\overline\mu_1\|_\infty^2$ is smaller than $h^{-1}\|\mu_1-\overline\mu_1\|_\infty$. Building upon \cite{mammen2012nonparametric}, in this work, we aimed to provide a more nuanced analysis than this crude calculation would imply in this stylized setting. However, the difficulty in obtaining pure second-order nuisance errors when $h$ goes to zero remains.
\section{Main proofs}
\subsection{Proof of Proposition \ref{proposition:bc_local}}
Recall that $Y = \mu(X) + \epsilon$ and, for shorthand notation, let us write $H_{ht}(r) = K_{ht}(r) g_{ht}(r)$ and $G_{ht}(r) = K_{ht}(r) g_{ht}(r) g_{ht}(r)^\intercal = H_{ht}(r) g_{ht}(r)^\intercal$. We interpret all derivatives with respect to the scalar $r$ component-wise. Thus, we have
\begin{align*}
    &\varphi_1 (Z,\rhat)= H'_{ht}(\rhat) \cdot (A - \rhat) \cdot \widehat\mu + H_{ht}(\rhat) \cdot Y, \\
   &\varphi_2 (Z,\rhat)  = G'_{ht}(\rhat) \cdot (A - \rhat)  + G_{ht}(\rhat) = \left\{H'_{ht}(\rhat)g_{ht}(\rhat)^\intercal + H_{ht}(\rhat) g'_{ht}(\rhat)^\intercal\right\} \cdot (A - \rhat) + G_{ht}(\rhat),
\end{align*}
and
\begin{align*}
    \varphi_1(Z,\rhat) & =  \left\{H_{ht}'(\rhat) \cdot(A - \rhat) + H_{ht}(\rhat) \right\} \cdot \mu+H'_{ht}(\rhat) \cdot (A - \rhat) \cdot (\widehat\mu - \mu)  + H_{ht}(\rhat) \cdot \epsilon \\
    &=: R_1 + R_2 + H_{ht}(\rhat) \cdot \epsilon.
\end{align*}
Let $\beta_{ht} = \begin{bmatrix} m(t) & h \cdot m'(t) \end{bmatrix}^\intercal$. By a Taylor expansion and definition of $\rho(X) := \mu(X) - m\{r(X)\}$, we have
\begin{align*}
    \mu & = g_{ht}(\rhat)^\intercal \beta_{ht} + m'(t)(r - \rhat) + \frac{1}{2}m''(\overline{t})(r-t)^2 + \rho,
\end{align*}
for some intermediate value $\overline{t}$ between $t$ and $r(X)$. Notice that
\begin{align*}
    \left\{H_{ht}'(\rhat) \cdot(A - \rhat) + H_{ht}(\rhat) \right\} \cdot g_{ht}(\rhat)^\intercal \beta_{ht} & = \varphi_2(Z, \rhat) \beta_{ht} - H_{ht}(\rhat)g'_{ht}(\rhat)^\intercal \beta_{ht} \cdot (A - \rhat)  \\
    & = \varphi_2(Z, \rhat)\beta_{ht} - m'(t) \cdot H_{ht}(\rhat) \cdot (A - \rhat) \\
    & = \varphi_2(Z, \rhat) \beta_{ht} - H_{ht}(\rhat) \cdot m'(t) \cdot \left\{(A - r)+ (r - \rhat)\right\}
\end{align*}
Therefore, we have
\begin{align*}
    R_1 & = \varphi_2(Z, \rhat) \beta_{ht} - m'(t) H_{ht}(\rhat) (A - r) + H'_{ht}(\rhat) \cdot (A - \rhat) \cdot m'(t)  \cdot (r - \rhat) \\
    & \hphantom{=} +\left\{H_{ht}'(\rhat) \cdot (A - \rhat) + H_{ht}(\rhat)\right\} \cdot \left\{\frac{1}{2}m''(\overline{t})(r-t)^2 + \rho\right\}
\end{align*}
Thus, we can write
\begin{align*}
    & \Pn\{\varphi_1(Z, \rhat)- \varphi_2(Z, \rhat) \beta_{ht}\} \\
    & = \Pn(R_2) + \Pn\left[H_{ht}(\rhat)\left\{\epsilon - m'(t)(A - r)\right\}\right] + \Pn\left\{H_{ht}'(\rhat) \cdot (A - \rhat) \cdot m'(t) \cdot  (r - \rhat)\right\} \\
    & \hphantom{=} + \Pn\left[\left\{H_{ht}'(\rhat) \cdot (A - \rhat) + H_{ht}(\rhat)\right\} \cdot \left\{\frac{1}{2}m''(\overline{t})(r-t)^2 + \rho\right\}\right].
\end{align*}
Next, we decompose $\Pn\{\varphi_1(Z, \rhat)- \varphi_2(Z, \rhat) \beta_{ht}\}$ into terms of three types: leading terms, terms that have mean zero given $X$ and $D^n$ and nuisance bias terms. 

To start, we first bound the term $\|\Pn(R_2)\|_2$. Note that\begin{align*}
    &\|H'_{ht}(\rhat)\|_2\lesssim\frac 1 {h^2} \cdot 1(|\rhat-t|\leq h)\lesssim\frac 1 {h^2}\cdot 1(|r-t|\leq h+\|r-\rhat\|_\infty).
\end{align*} So we have
\begin{align*}
    \|\Pb (R_2)\|_2 &= \|\Pb \{H_{ht}'(\rhat)(r-\rhat)(\muhat-\mu)\}\|_2\\&\lesssim\|r-\rhat\|_\infty\|\mu-\muhat\|_\infty\Pb\|H'_{ht}(\rhat)\|_2\\
    &\lesssim\|r-\rhat\|_\infty\|\mu-\muhat\|_\infty \cdot\frac 1{h^2}\cdot (h+\|r-\rhat\|_\infty)
\end{align*}
And 
\begin{align*}
    \Pb(\|(\Pb-\Pn) R_2\|_2^2)&\lesssim\frac 1n\Pb\|R_2\|_2^2\\&\lesssim \frac 1n\|\mu-\muhat\|_\infty^2\Pb\|H'_{ht}(\rhat)\|_2^2\\
    &\lesssim\frac 1n \|\mu-\muhat\|^2_\infty\frac 1 {h^4}\cdot(h+\|r-\rhat\|_\infty)
\end{align*}
Therefore,
\begin{align*}
    \|\Pn(R_2)\|_2
    & = O_{\Pb}\left(\frac{\|\muhat - \mu\|_{\infty}}{ \sqrt{n}h^2} \cdot(h+\|r-\rhat\|_\infty)^{1/2}+ \frac{\|\muhat - \mu\|_{\infty}\|\rhat - r\|_{\infty}}{h^2}\cdot(h+\|r-\rhat\|_\infty)\right)
\end{align*}

Similarly, we have
\begin{align*}
    & \left\|\Pn\left\{H_{ht}'(\rhat) \cdot (A - \rhat) \cdot m'(t) \cdot  (r - \rhat)\right\}\right\|_2 \\
    & =O_{\Pb}\left(\frac{\|r-\rhat\|_{\infty}}{ \sqrt{n}h^2} \cdot(h+\|r-\rhat\|_\infty)^{1/2}+ \frac{\|\rhat - r\|_{\infty}^2}{h^2}\cdot(h+\|r-\rhat\|_\infty)\right)
\end{align*}
In addition, since
\begin{align}
     \left\|\int_0^1H'_{ht}(r+u(\rhat-r))du\right\|_2 \lesssim\frac 1 {h^2}\cdot 1(|r-t|\leq h+\|r-\rhat\|_\infty),
     \label{eq:cpi_local_2}
\end{align} we have
\begin{align}
    &\Pn\left[H_{ht}(\rhat)\left\{\epsilon - m'(t)(A - r)\right\}\right]\nonumber\\ &\quad = \Pn\left[H_{ht}(r)\left\{\epsilon - m'(t)(A - r)\right\}\right]\nonumber \\
    & \hphantom{\quad =} + \Pn\left[\int_0^1 H'_{ht}(r + u(\rhat - r))(\rhat - r) du \cdot \left\{\epsilon - m'(t)(A - r)\right\}\right]\nonumber \\
    & \quad= \Pn\left[H_{ht}(r)\left\{\epsilon - m'(t)(A - r)\right\}\right] + O_{\Pb}\left(\frac{\|\rhat - r\|_{\infty}}{\sqrt{n}h^2}\cdot(h+\|r-\rhat\|_\infty)^{1/2}\right)
    \label{eq:cpi_local_1}
\end{align}
Furthermore, we have
\begin{align*}
H_{ht}'(\rhat) \cdot (A - \rhat) + H_{ht}(\rhat) & = H_{ht}(r) + H_{ht}'(\rhat) \cdot (A - r) + \int_0^1 \{H'(r + u(\rhat - r)) - H'_{ht}(\rhat)\}(\rhat - r) du \\
& = H_{ht}(r) + H_{ht}'(r) \cdot (A - r) + \int_0^1H_{ht}''(r + u(\rhat - r)) (\rhat - r) du \cdot (A - r) \\
& \hphantom{=} + \int_0^1 \{H'(r + u(\rhat - r)) - H'_{ht}(\rhat)\}(\rhat - r) du\\
&=: H_{ht}(r) + H_{ht}'(r) \cdot (A - r)+I_1+I_2
\end{align*}

we will have
\begin{align*}
    \left\|\Pn\left\{H'_{ht}(r)(A-r)\cdot \frac 12m''(\bar t)(r-t)^2\right\}\right\|_2 = O_{\Pb}\left(\sqrt{\frac{\Pb\|H'_{ht}(r)(r-t)^2\|_2^2}{ n}}\right) = O_{\Pb}\left(\sqrt{\frac hn}\right).
\end{align*}

Since $\left\|\int_0^1H''_{ht}(r+u(\rhat-r))du\right\|_2\lesssim \frac 1 {h^3}\cdot 1(|r-t|\leq h+\|r-\rhat\|_\infty)$,
\begin{align*}
    \Pb\|I_1\cdot(r-t)^2\|_2^2 &\lesssim\Pb\left\{(r-\rhat)^2\left\|\int_0^1H''_{ht}(r+u(\rhat-r))du\right\|_2^2\cdot(r-t)^4\right\}\\
    &\lesssim\|r-\rhat\|_{\infty}^2\cdot\frac{1}{h^6}\left(h+\|r-\rhat\|_\infty\right)^5,
\end{align*}
we have
\begin{align*}
    \left\|\Pn\left\{I_1\cdot \frac 12m''(\bar t)(r-t)^2\right\}\right\|_2 = O_{\Pb}\left(\sqrt{\frac{\Pb\|I_1\cdot(r-t)^2\|_2^2}{ n}}\right) = O_{\Pb}\left(\frac{\|r-\rhat\|_{\infty} }{\sqrt n h^3}\cdot (h+\|r-\rhat\|_\infty)^{5/2}\right).
\end{align*}

And for the term $I_2\cdot\frac 12m''(\bar t)(r-t)^2$,
\begin{align*}
    \left\|\Pb\left\{I_2\cdot\frac 12m''(\bar t)(r-t)^2\right\}\right\|_2&\lesssim\|r-\rhat\|_\infty\Pb\left\| \left\{\int_0^1 H'_{ht}(r+u(\rhat-r))du- H'_{ht} (\rhat)\right\}(r-t)^2\right\|_2\\
    &\lesssim \frac{\|r-\rhat\|_\infty}{h^2}\cdot(h+\|r-\rhat\|_\infty)^3.
\end{align*}
Also
\begin{align*}
    \Pb\left\|I_2\cdot \frac 12 m''(\bar t)(r-t)^2\right\|_2^2\lesssim\frac{\|r-\rhat\|_\infty^2}{h^4}\cdot(h+\|r-\rhat\|_\infty)^5.
\end{align*}
Then we have
\begin{align*}
    \left\|\Pn\left\{I_2\cdot \frac 12 m''(\bar t)(r-t)^2\right\}\right\|_2 = O_{\Pb}\left(\frac{\|r-\rhat\|_\infty}{\sqrt n h^2}\cdot(h+\|r-\rhat\|_\infty)^{5/2}+\frac{\|r-\rhat\|_\infty}{h^2}\cdot(h+\|r-\rhat\|_\infty)^3\right)
\end{align*}

In this light, we have
\begin{align*}
 &\Pn\left[\left\{H_{ht}'(\rhat) \cdot (A - \rhat) + H_{ht}(\rhat)\right\} \cdot \frac{1}{2} m''(\overline{t})(r - t)^2\right] 
 = \frac{1}{2} m''(t) \cdot \Pn\{ H_{ht}(r) (r - t)^2\} + o_\Pb(h^2)\\&\hphantom{=} + O_{\Pb}\left(\sqrt{\frac hn}+\frac{\|r-\rhat\|_\infty}{\sqrt n h^3}\cdot(h+\|r-\rhat\|_\infty)^{5/2}+\frac{\|r-\rhat\|_\infty}{h^2}\cdot(h+\|r-\rhat\|_\infty)^3\right)
\end{align*}
Because $\rho$ is not necessarily vanishing as $h \to 0$ and $n \to \infty$, we have an additional leading term in the following expression:
\begin{align*}
   &  \Pn\left[\left\{H_{ht}'(\rhat) \cdot (A - \rhat) + H_{ht}(\rhat)\right\} \cdot \rho\right]   \\
   &= \Pn\left[ \left\{H_{ht}(r) + H_{ht}'(r) \cdot (A - r)\right\} \cdot \rho \right] \\
& \hphantom{=} + O_{\Pb}\left(\frac{\|\rhat - r\|_\infty\cdot\{\sup_{t_1,t_2}\E(\rho^2 \mid \rhat = t_1, r = t_2, D^n)\}^{1/2} }{\sqrt{n}h^3} \cdot (h+\|r-\rhat\|_\infty)^{1/2}\right)\\
&\hphantom{=} + O_{\Pb}\left( \frac{\|\rhat - r\|_\infty\sup_{t_1,t_2}|\E(\rho\mid\rhat = t_1,r = t_2, D^n) | }{h^2}\cdot(h+\|r-\rhat\|_\infty)\right)
\end{align*}
Putting everything together, we have that
\begin{align*}
&\Pn\{\varphi_1(Z,\rhat) - \varphi_2(Z,\rhat) \beta_{ht}\} \\& = \Pn \left[H_{ht}(r)\left\{\epsilon - m'(t)(A - r) \right\}\right] + \frac{1}{2} m''(t)\Pb\{H_{ht}(r) (r - t)^2 \} + o_\Pb(h^2) \\
& \hphantom{=} + \Pn\left[ \left\{H_{ht}(r) + H_{ht}'(r) \cdot (A - r)\right\} \cdot \rho \right] \\
& \hphantom{=} +O_{\Pb}\left(\frac{\|\muhat - \mu\|_{\infty}}{ \sqrt{n}h^2} \cdot(h+\|r-\rhat\|_\infty)^{1/2}+ \frac{\|\muhat - \mu\|_{\infty}\|\rhat - r\|_{\infty}}{h^2}\cdot(h+\|r-\rhat\|_\infty)\right)\\
& \hphantom{=} +O_{\Pb}\left(\frac{\|r-\rhat\|_{\infty}}{ \sqrt{n}h^2} \cdot(h+\|r-\rhat\|_\infty)^{1/2}+ \frac{\|\rhat - r\|_{\infty}^2}{h^2}\cdot(h+\|r-\rhat\|_\infty)\right)\\
& \hphantom{=} +  O_{\Pb}\left(\sqrt{\frac hn}+\frac{\|r-\rhat\|_\infty}{\sqrt n h^3}\cdot(h+\|r-\rhat\|_\infty)^{5/2}+\frac{\|r-\rhat\|_\infty}{h^2}\cdot(h+\|r-\rhat\|_\infty)^3\right) \\
& \hphantom{=} + O_{\Pb}\left(\frac{\|\rhat - r\|_\infty\cdot\{\sup_{t_1,t_2}\E(\rho^2 \mid \rhat = t_1, r = t_2, D^n)\}^{1/2} }{\sqrt{n}h^3} \cdot (h+\|r-\rhat\|_\infty)^{1/2}\right)\\
&\hphantom{=} + O_{\Pb}\left( \frac{\|\rhat - r\|_\infty\sup_{t_1,t_2}|\E(\rho\mid\rhat = t_1,r = t_2, D^n) | }{h^2}\cdot(h+\|r-\rhat\|_\infty)\right)
\end{align*}

By definition and change of variable, we have
\begin{align*}
    \Pb\varphi_2(Z,r) &=\Pb G_{ht}(r)  = \int K_{ht}(r)g_{ht}(r)g_{ht}^\intercal(r) f_r(r)dr \\&= \int K(u)g(u)g(u)^\intercal f_r(t+hu)du.
\end{align*}
Since $f_r(\cdot)$ is Lipschitz continuous, $f_r(t+hu) = f_r(t) + O(h)$. Therefore,
\begin{align*}
\Pb G_{ht}(r)
= f_r(t)
\begin{pmatrix}
1 &  0\\
 0 &  \int u^2 K(u)\,du
\end{pmatrix}
+ O(h).
\end{align*} 
Since $f_r(t)$ is bounded and away from 0, $\lambda_\min(\Pb \varphi_2(Z,r)) \gtrsim 1$. Therefore, we have
\begin{align*}
    &\mhat_{\rm lbc}(t; h) - m(t) \\& = e_1^\intercal \{\Pn\varphi_2(Z,\rhat)\}^{-1}\cdot\Pn\{\varphi_1(Z,\rhat) - \varphi_2(Z,\rhat) \beta_{ht}\} \\
    & = e_1^\intercal \{\Pb\varphi_2(Z,r)\}^{-1}\{\Pb\varphi_2(Z,r) - \Pn\varphi_2(Z,\rhat)\}\{\Pn \varphi_2(Z,\rhat)\}^{-1}\cdot \Pn\{\varphi_1(Z,\rhat) - \varphi_2(Z,\rhat) \beta_{ht}\}\\&\hphantom{=} + e_1^\intercal \{\Pb\varphi_2(Z,r)\}^{-1}\cdot \Pn\{\varphi_1(Z,\rhat) - \varphi_2(Z,\rhat) \beta_{ht}\} \\
    & = e_1^\intercal \{\Pb\varphi_2(Z,r)\}^{-1}\cdot \Pn\{\varphi_1(Z,\rhat) - \varphi_2(Z,\rhat) \beta_{ht}\} + o_\Pb\left(\Pn\{\varphi_1(Z,\rhat) - \varphi_2(Z,\rhat) \beta_{ht}\}\right),
\end{align*}
where the last equality follows by lemma \ref{lemma:min_eigen_local} and the condition $nh^3 \to \infty$ and $\|\rhat - r\|_\infty = o(h)$. In addition, we have $e_1^\intercal \{\Pb\varphi_2(Z,r)\}^{-1} = \begin{bmatrix} f^{-1}_r(t) & 0 \end{bmatrix} + O(h)$ and 
\begin{align*}
\frac{1}{2} m''(t)\Pb\{H_{ht}(r) (r - t)^2 \} = \frac{1}{2} m''(t) f_r(t) h^2 \int K(u) g(u) u^2 du + o(h^2),
\end{align*}
where $f_r$ denotes the density of $r(X)$. Therefore, relying on $\|\rhat - r\|_\infty = o(h)$, we have reached
\begin{align*}
 & \mhat_{\rm lbc}(t; h) - m(t)  - \frac{1}{2} m''(t) h^2 \int K(u) u^2 du \\
 & = \left(\frac{1}{f_r(t)} \Pn \left[K_{ht}(r)\{\epsilon+\rho - m'(t) (A - r)\} + 
 K_{ht}'(r)(A - r) \rho\right] \right) + o_\Pb(h^2)\\
 & \hphantom{=} + o_\Pb\left(\| \Pn \left[K_{ht}(r)\{\epsilon+\rho - m'(t) (A - r)\} + 
 K_{ht}'(r)(A - r) \rho\right]\|_2 \right) \\
& \hphantom{=} +  O_{\Pb}\left(\frac{\|r-\rhat\|_\infty+\|\mu-\muhat\|_\infty}{\sqrt {nh^3}}+\frac{\|\rhat - r\|_\infty\cdot\{\sup_{t_1,t_2}\E(\rho^2 \mid \rhat = t_1, r = t_2, D^n)\}^{1/2} }{\sqrt{nh^5}} \right) \\
&\hphantom{=} +O_{\Pb}\left(\frac{\|\rhat - r\|_\infty\left\{\sup_{t_1,t_2}|\E(\rho\mid\rhat = t_1,r = t_2, D^n) | +\|r-\rhat\|_\infty+\|\mu-\muhat\|_\infty\right\}}{h}\right)
\end{align*}
\subsection{Proof of Proposition \ref{proposition:cpi_local}}
To enhance clarity in our proof of this proposition, we don't use the $\Pn$ notation and instead use the regular $\sum$ notation. We will repeatedly rely on the following expansion,
\begin{align}\label{eq:taylorY}
    Y & = m(t) + m'(t)(r - t) + \frac{1}{2}m''(\overline{t})(r - t)^2 + \epsilon + \rho,
\end{align}
for some intermediate value $\overline{t}$ between $r$ and $t$. Recall the notation 
\begin{align*}
    & W_{it}(X^n, \rhat; h) = e_1^\intercal \cdot \widehat{Q}_{ht}^{-1}\cdot H_{ht}(\rhat_i), \quad H_{ht}(\rhat) = K_{ht}(\rhat) g_{ht}(\rhat), \quad \widehat{Q}_{ht} = \frac{1}{n}\sum_{i = 1}^n K_{ht}(\rhat_i) g_{ht}(\rhat_i) g_{ht}(\rhat_i)^\intercal, \\
    & W_{2, it, -j}(X^n, \rhat; b) = e_2^\intercal \cdot \widehat{Q}_{bt, -j}^{-1} \cdot b^{-1} \cdot H_{bt}(\rhat_i), \quad \widehat{Q}_{bt, -j} = \frac{1}{n-1}\sum_{i = 1, i \neq j}^n K_{bt}(\rhat_i) g_{bt}(\rhat_i) g_{bt}(\rhat_i)^\intercal \\
    & \widehat\varphi(Z_i, Z_j; h, b) = W_{2, it, -j}(X^n, \rhat; b) Y_i \cdot W_{jt}(X^n, \rhat; h)(A_j - t) \\
    & \mhat_{lcpi}(t; h, b) = \frac{1}{n}\sum_{i =1}^n W_{it}(X^n, \rhat; h) Y_i - \frac{1}{n(n-1)}\mathop{\sum\sum}_{1 \leq i \neq j \leq n} \widehat\varphi(Z_i, Z_j; h, b)
\end{align*}
By the properties of local-linear weights (see, e.g., Proposition 1.12 in \cite{tsybakov2008introduction}), we have
\begin{align*}
    & \frac{1}{n}\sum_{i = 1}^n W_{it}(X^n, \rhat; h) = 1, \quad \frac{1}{n}\sum_{i = 1}^n W_{it}(X^n, \rhat; h)(\rhat_i - t) = 0, \quad \frac{1}{n-1}\sum_{i = 1, i \neq j}^n W_{2, it, -j}(X^n, \rhat; b) = 0, \\
    & \text{and } \quad \frac{1}{n-1}\sum_{i = 1, i \neq j}^n W_{2, it, -j}(X^n, \rhat; b)(\rhat_i-t) = 1.
\end{align*}
By \eqref{eq:taylorY}, we have
\begin{align*}
   \frac{1}{n}\sum_{i =1}^n W_{it}(X^n, \rhat; h) Y_i & = m(t) + m'(t) \cdot \frac{1}{n}\sum_{i = 1}^n  W_{it}(X^n, \rhat; h)(r_i - t) \\
   & \hphantom{=} + \frac{1}{n}\sum_{i = 1}^n  W_{it}(X^n, \rhat; h)\left\{\frac{1}{2}m''(\overline{t}_i)(r_i - t)^2 + \epsilon_i + \rho_i\right\}.
\end{align*}
Furthermore, we have
\begin{align*}
    & \frac{1}{n(n-1)}\mathop{\sum\sum}_{1 \leq i \neq j \leq n} \widehat\varphi(Z_i, Z_j; h, b) \\
    & = \frac{1}{n(n-1)}\mathop{\sum\sum}_{1 \leq i \neq j \leq n} W_{2, it, -j}(X^n, \rhat; b)\{m(t) + m'(t)(\rhat_i - t)\} \cdot W_{jt}(X^n, \rhat; h)(A_j - t) \\
    & \hphantom{=} + \frac{1}{n(n-1)}\mathop{\sum\sum}_{1 \leq i \neq j \leq n} W_{2, it, -j}(X^n, \rhat; b) \left\{m'(t)(r_i - \rhat_i) + \frac{1}{2}m''(\overline{t}_i)(r_i - t)^2 + \epsilon_i + \rho_i\right\} \\
    & \qquad\qquad\qquad\qquad\qquad \times \ W_{jt}(X^n, \rhat; h)(A_j - t)  \\
    & = m'(t) \cdot \frac{1}{n}\sum_{i = 1}^n  W_{it}(X^n, \rhat; h) (A_i-t)\\
     & \hphantom{=} + \frac{1}{n(n-1)}\mathop{\sum\sum}_{1 \leq i \neq j \leq n} W_{2, it, -j}(X^n, \rhat; b) \left\{m'(t)(r_i - \rhat_i) + \frac{1}{2}m''(\overline{t}_i)(r_i - t)^2 + \epsilon_i + \rho_i\right\} \\
    & \qquad\qquad\qquad\qquad\qquad \times \ W_{jt}(X^n, \rhat; h)(A_j - t)
\end{align*}
In this light, we have
\begin{align}\label{eq:main_error_derivative}
    & \mhat_{lcpi}(t; h, b) - m(t) \nonumber \\
    & = \frac{1}{n}\sum_{i = 1}^n  W_{it}(X^n, \rhat; h)\left\{\frac{1}{2}m''(\overline{t}_i)(r_i - t)^2 + \epsilon_i - m'(t)(A_i - r_i) + \rho_i\right\} \nonumber \\
    & \hphantom{=} - \frac{1}{n(n-1)}\mathop{\sum\sum}_{1 \leq i \neq j \leq n} W_{2, it, -j}(X^n, \rhat; b) \left\{m'(t)(r_i - \rhat_i) + \frac{1}{2}m''(\overline{t}_i)(r_i - t)^2 + \epsilon_i + \rho_i\right\} \nonumber \\
    & \qquad\qquad\qquad\qquad\qquad \times \ W_{jt}(X^n, \rhat; h)(A_j - t)
\end{align}
The first term can be analyzed in a way similar to that of the proof of Proposition \ref{proposition:bc_local}. In particular, 
\begin{align*}
    W_{it}(X^n, \rhat; h) = e_1^\intercal \widehat{Q}_{ht}^{-1} \{H_{ht}(\rhat) - H_{ht}(r)\} + e_1^\intercal \widehat{Q}_{ht}^{-1}(Q_{ht} - \widehat{Q}_{ht})Q_{ht}^{-1} H_{ht}(r) + e_1^\intercal Q_{ht}^{-1} H_{ht}(r).
\end{align*}
The same arguments as those from Lemma \ref{lemma:min_eigen_local}, yields that 
\begin{align*}
    \lambda_\min(\widehat{Q}_{ht}) \geq c - \|\widehat{Q}_{ht} - Q_{ht}\|_{\rm F} = c - O_{\Pb}\left(\frac{1}{\sqrt{n}h}\cdot (h+\|r-\rhat\|_\infty)^{1/2} + \frac{\|\rhat - r\|_\infty}{h^2}\cdot(h+\|r-\rhat\|_\infty)\right).
\end{align*}
Furthermore, we have already shown in Eq.(\ref{eq:cpi_local_1})
\begin{align*}
  & \frac{1}{n} \sum_{i = 1}^n \{H_{ht}(\rhat_i) - H_{ht}(r_i)\}\{\epsilon_i - m'(t) (A_i - r_i)\} =  O_{\Pb}\left(\frac{\|\rhat - r\|_\infty}{\sqrt{n}h^2}\cdot(h+\|r-\rhat\|_\infty)^{1/2}\right).
\end{align*}
Also because Eq.(\ref{eq:cpi_local_2}), we have
\begin{align*}
   &\frac{1}{n} \sum_{i = 1}^n \{H_{ht}(\rhat_i) - H_{ht}(r_i)\} \rho_i  = \frac{1}{n} \sum_{i = 1}^n \int_0^1 H'_{ht}(r_i + u(\rhat_i - r_i)) (\rhat_i - r_i) du \cdot \rho_i \\& = O_{\Pb}\left(\Pb\left\{|r-\rhat|\cdot|\E(\rho\mid r,\rhat, D^n)|\cdot\left\|\int_0^1H'_{ht}(r+u(\rhat-r))du\right\|_2\right\}\right)\\
   &\hphantom{=} +O_{\Pb}\left(\sqrt{\frac{\Pb \{(r-\rhat)^2\E(\rho^2\mid r,\rhat, D^n)\left\|\int_0^1H'_{ht}(r+u(\rhat-r))du\right\|_2^2\}}{n}}\right)\\
   & = O_{\Pb}\left(\frac{\|\rhat - r\|_\infty \sup_{t_1,t_2}|\E(\rho\mid\rhat = t_1,r = t_2, D^n) |}{h^2} \cdot(h+\|r-\rhat\|_\infty) \right. \\
   & \left. \qquad \hphantom{=} + \frac{\|\rhat - r\|_\infty\|\E(\rho^2\mid r,\rhat, D^n)\|_\infty^{1/2}}{\sqrt{n}h^2}\cdot (h+\|r-\rhat\|_\infty)^{1/2}\right)
\end{align*}
Similarly,
\begin{align*}
   & \frac{1}{n} \sum_{i = 1}^n \{H_{ht}(\rhat_i) - H_{ht}(r_i)\}\cdot\frac 12 m''(\bar t_i)(r_i-t)^2=\\
    &  \hphantom{=} O_{\Pb}\left(\frac{\|r-\rhat\|_\infty}{h^2}\cdot(h+\|r-\rhat\|_\infty)^3+\frac{\|r-\rhat\|_\infty}{\sqrt n h^2}\cdot(h+\|r-\rhat\|_\infty)^{5/2}\right).
\end{align*}
 Under the assumption the density $f_r$ is bounded above and away from zero, $Q_{ht}$ has eigenvalues bounded above and below away from zero. In addition, since $nh \to \infty$ and $\|\rhat - r\|_\infty = o(h)$, we have  $\|\widehat{Q}_{ht}^{-1}\|_{\rm op}$ is bounded with probability tending to one.  Further, \begin{align*}
     e_1^\intercal Q_{ht}^{-1} H_{ht}(r) = \{f_r(t)^{-1} + O(h)\} K_{ht}(r),
 \end{align*}
 so that
\begin{align*}
   & \frac{1}{n}\sum_{i = 1}^n  W_{it}(X^n, \rhat; h)\left\{\epsilon_i + \rho_i - m'(t)(A_i - r_i) + \frac{1}{2}m''(\overline{t}_i)(r_i - t)^2 \right\} \\
   & = \frac{1}{f_r(t)} \cdot \left[\frac{1}{n} \sum_{i = 1}^n K_{ht}(r_i) \{\epsilon_i + \rho_i - m'(t) (A_i - r_i)\}\right] + \frac{1}{2} m''(t) h^2 \int K(u)u^2 du \\
   & \hphantom{=} + O_{\Pb}\left(\frac{\|\rhat - r\|_\infty}{h} \cdot \left\{\frac{1}{\sqrt{nh}} + 
\sup_{t_1,t_2}|\E(\rho\mid\rhat = t_1,r = t_2, D^n) |\right\}\right) + o_\Pb\left(h^2 + \frac{1}{\sqrt{nh}}\right).
\end{align*}
Our proof is finished after applying Lemma \ref{lemma:extra_Op} bounding the order of the $U$-statistic term in Eq. \eqref{eq:main_error_derivative}.
\medskip
\begin{lemma}\label{lemma:extra_Op}
Under the conditions of Proposition \ref{proposition:cpi_local}, it holds that
\begin{align*}
& \frac{1}{n(n-1)}\mathop{\sum\sum}_{1 \leq i \neq j \leq n} W_{2, it, -j}(X^n,\rhat; b) \left\{m'(t)(r_i - \rhat_i) + \frac 12 m''(\bar t_i)(r_i-t)^2  + \epsilon_i + \rho_i\right\} W_{jt}(X^n,\rhat; h)(A_j - t) \\
& \quad\quad = O_{\Pb}\left(\frac{1}{\sqrt{n^2 h b^3}} + \frac{\|\rhat - r\|_\infty + \|\E(\rho\mid \rhat,r, D^n)\|_\infty}{b \sqrt{n h}} + \frac{\|\rhat - r\|_\infty}{\sqrt{nb^3}}+\frac{b}{\sqrt{nh}}\right)\\&\quad\quad \quad + O_{\Pb}\left(b \cdot \|\rhat - r\|_\infty + \frac{\|\rhat - r\|_\infty\{\|\rhat - r\|_\infty + \|\E(\rho\mid \rhat,r, D^n)\|_\infty\}}{b}\right).
\end{align*}

\end{lemma}
\begin{proof}
To prove the statement, we break the double sum into two terms according to:
\begin{align*}
W_{jt}(X^n,\rhat;h)(A_j - t) = W_{jt}(X^n,\rhat;h)(A_j - \rhat_j) + W_{jt}(X^n,\rhat;h)(\rhat_j - t).
\end{align*}
We start with the first term, i.e., 
\begin{align*}
& \frac{1}{n(n-1)}\mathop{\sum\sum}_{1 \leq i \neq j \leq n} W_{2, it, -j}(X^n,\rhat; b) \left\{m'(t)(r_i - \rhat_i) + \frac 12 m''(\bar t_i)(r_i-t)^2  + \epsilon_i + \rho_i\right\} W_{jt}(X^n,\rhat; h)(A_j - \rhat_j) \\
 & \equiv \frac{1}{n(n-1)}\mathop{\sum\sum}_{1 \leq i \neq j \leq n} T_{ij} 
\end{align*}

Under the assumption that $Q_{bt}$ has eigenvalues bounded above and below away from zero, $nb \to \infty$, $\|\rhat - r\|_\infty = o(b)$, $\|\widehat{Q}_{bt,-j}^{-1}\|_{\rm op}$ is also bounded, uniformly over $j$, with probability tending to one:

\begin{align*}
    \widehat{Q}_{bt, -j} = \frac{n}{n-1} \widehat{Q}_{bt} - \frac{K_{bt}(\rhat_j) g_{bt}(\rhat_j) g_{bt}(\rhat_j)^\intercal}{(n-1)} \implies \min_j \lambda_\min(\widehat{Q}_{bt, -j}) \gtrsim 1.
\end{align*}
Therefore,
we have
\begin{align*}
&|W_{2, it, -j}(X^n,\rhat; b)| \ \lesssim b^{-1}\|H_{bt}(\rhat_i)\|_2 \lesssim \frac 1{b^2}\cdot 1(|\rhat_i-t|\leq b)\quad \text{ and } \\&|W_{jt}(X^n,\rhat; h)| \ \lesssim \|H_{ht}(\rhat_j)\|_2\lesssim \frac 1h \cdot 1(|\rhat_j-t|\leq h).
\end{align*}
For shorthand notation let $S_i = m'(t)(r_i - \rhat_i) + \frac 12 m''(\bar t_i)(r_i-t)^2 + \epsilon_i + \rho_i$, $G_j = A_j - \rhat_j$ . Then  $T_{ij} =W_{2, it, -j}(X^n,\rhat; b)\cdot S_i\cdot W_{jt}(X^n,\rhat; h)\cdot G_j$. Let also $\rhat^n$ denote the vector of random variables $\rhat(X_1), \ldots, \rhat(X_n)$. 

Since $\E(\epsilon_i \mid X_i, D^n) = 0$ and $ |r_i - \rhat_i + \frac 12m''(\bar t_i)(r_i-t)^2| \ \lesssim \|\rhat - r\|_\infty + (\rhat_i - t)^2$, we can bound
\begin{align*}
    |\E(S_i\mid \rhat^n, D^n)| & \lesssim\|\rhat-r\|_\infty+(\rhat_i-t)^2+\sup_{t}|\E(\rho\mid \rhat = t, D^n)|,
\end{align*}
\begin{align*}
        |\E(G_j\mid \rhat^n , D^n)| & \leq\|r-\rhat\|_\infty,
\end{align*} 
and
\begin{align*}
        |\E(S_i G_i\mid \rhat^n, D^n)| & \leq \sqrt{\E(S_i^2\mid \rhat^n, D^n)\cdot\E(G_i^2\mid \rhat^n, D^n)}\lesssim 1 
\end{align*}

Under the assumption that $\|\rhat - r\|_\infty = o(b)$, we have that
\begin{align*}
    |\Pb\left[T_{ij}^2\right]| 
    &\leq \Pb\left[W_{2, it, -j}^2(X^n,\rhat; b) \cdot \E(S_i^2\mid \rhat^n, D^n) \cdot W_{jt}^2(X^n,\rhat; h) \cdot \E(G_j^2\mid \rhat^n, D^n) \right] \\
    &\lesssim \Pb\left[\|H_{ht}(\rhat_j)\|_2^2\cdot \frac 1{b^2}\|H_{bt}(\rhat_i)\|_2^2\right] 
    \lesssim\frac{1}{b^3 h},
    \end{align*}
\begin{align*}
    \left| \Pb(T_{ij}T_{ji} )\right|
    &\leq \Pb\Big[|W_{2, it, -j}(X^n,\rhat; b)| \cdot |W_{jt}(X^n,\rhat; h)| \cdot |W_{2, jt, -i}(X^n,\rhat; b)| \cdot |W_{it}(X^n,\rhat; h)| \\
    &\qquad\qquad \cdot |\E(S_iG_i\mid \rhat^n, D^n)|\cdot|\E(S_j G_j\mid \rhat^n, D^n)| \Big] \\
    &\lesssim \frac{1}{b^2 \cdot \max(b^2, h^2)},
\end{align*}
\begin{align*}
      |\Pb(T_{ij}T_{il})|
    &\leq \Pb\Big[|W_{2, it, -j}(X^n,\rhat; b)| \cdot |W_{2, it, -l}(X^n,\rhat; b)| \cdot |W_{jt}(X^n,\rhat; h)| \cdot |W_{lt}(X^n,\rhat; h)| \\
    &\qquad\qquad \cdot |\E(S_i^2\mid \rhat^n, D^n)|\cdot|\E(G_j\mid \rhat^n, D^n)|\cdot|\E(G_l\mid \rhat^n, D^n)| \Big] \\
    &\lesssim \frac{\|r-\rhat\|_\infty^2}{b^3},
\end{align*}
\begin{align*}
    |\Pb(T_{ij}T_{li})| 
    &\leq \Pb\Big[|W_{2, it, -j}(X^n,\rhat; b)| \cdot |W_{2, lt, -i}(X^n,\rhat; b)| \cdot |W_{jt}(X^n,\rhat; h)| \cdot |W_{lt}(X^n,\rhat; h)| \\
    &\qquad\qquad \cdot |\E(S_iG_i\mid \rhat^n D^n)|\cdot|\E(G_j\mid \rhat^n, D^n)|\cdot|\E(S_l\mid \rhat^n, D^n)| \Big] \\
    &\lesssim \frac{(\|\rhat - r\|_\infty + b^2 + \sup_t|\E(\rho\mid \rhat, D^n)|)\cdot \|\rhat - r\|_\infty}{b^2 \cdot \max(b, h)},
\end{align*}
\begin{align*}
    |\Pb(T_{ij}T_{lj})| 
    &\leq \Pb\Big[|W_{2, it, -j}(X^n,\rhat; b)| \cdot |W_{2, lt, -j}(X^n,\rhat; b)| \cdot |W_{jt}(X^n,\rhat; h)|^2 \\
    &\qquad\qquad \cdot |\E(S_i\mid \rhat^n, D^n)|\cdot|\E(G_j^2\mid \rhat^n, D^n)|\cdot|\E(S_l\mid \rhat^n, D^n)| \Big] \\
    &\lesssim \frac{\{\|\rhat - r\|_\infty + b^2 + \sup_t |\E(\rho\mid \rhat = t, D^n)|\}^2}{b^2h},
\end{align*}
\begin{align*}
    |\Pb(T_{ij}T_{jl})| 
    &\leq \Pb\Big[|W_{2, it, -j}(X^n,\rhat; b)| \cdot |W_{2, jt, -l}(X^n,\rhat; b)| \cdot |W_{jt}(X^n,\rhat; h)| \cdot |W_{lt}(X^n,\rhat; h)| \\
    &\qquad\qquad \cdot |\E(S_i\mid \rhat^n, D^n)|\cdot|\E(S_jG_j\mid \rhat^n, D^n)|\cdot|\E(G_l\mid \rhat^n, D^n)| \Big] \\
    &\lesssim \frac{(\|\rhat - r\|_\infty + b^2 + \|\E(\rho\mid \rhat, D^n)\|_\infty)\cdot \|\rhat - r\|_\infty}{b^2 \cdot \max(b, h)},
\end{align*}
\begin{align*}
    |\Pb(T_{ij}T_{lm})| 
    &\leq \Pb\Big[|W_{2, it, -j}(X^n,\rhat; b)| \cdot |W_{2, lt, -m}(X^n,\rhat; b)| \cdot |W_{jt}(X^n,\rhat; h)| \cdot |W_{mt}(X^n,\rhat; h)| \\
    &\qquad\qquad \cdot |\E(S_i\mid \rhat^n, D^n)|\cdot|\E(G_j\mid \rhat^n, D^n)|\cdot|\E(S_l\mid \rhat^n, D^n)| \cdot|\E(G_m\mid \rhat^n, D^n)| \Big] \\
    &\lesssim \frac{(\|\rhat - r\|_\infty + b^2 + \sup_t|\E(\rho\mid \rhat = t, D^n)|)^2\cdot \|\rhat - r\|^2_\infty}{b^2}.
\end{align*}
Next, we write
\begin{align*}
    \left(\mathop{\sum\sum}_{1 \leq i \neq j \leq n} T_{ij}\right)^2 & = \mathop{\sum\sum}_{1 \leq i \neq j \leq n} (T^2_{ij} + T_{ij}T_{ji}) + \mathop{\sum\sum\sum}_{1 \leq i \neq j \neq l \leq n} (T_{ij}T_{il} + T_{ij}T_{li} + T_{ij}T_{lj} + T_{ij}T_{jl}) \\
    & \hphantom{=} + \mathop{\sum\sum\sum\sum}_{1 \leq i \neq j \neq l \neq m \leq n} T_{ij}T_{lm}
\end{align*}
and, using the inequalities above, we bound
\begin{align*}
    & \left\{\E\left(\frac{1}{n(n-1)}\mathop{\sum\sum}_{1 \leq i \neq j \leq n} T_{ij}\right)^2\right\}^{1/2} \\
    & \lesssim \frac{1}{\sqrt{n^2 h b^3}} + \frac{\|\rhat - r\|_\infty + \sup_t|\E(\rho\mid \rhat = t, D^n)|}{b \sqrt{n h}}+ \frac{\|\rhat - r\|_\infty}{\sqrt{nb^3}}+\frac{b}{\sqrt {nh}} \\ &\quad+ b \cdot \|\rhat - r\|_\infty + \frac{\|\rhat - r\|_\infty(\|\rhat - r\|_\infty + \sup_t|\E(\rho\mid \rhat, D^n)|)}{b}
\end{align*}
The result follows since $U = O_{\Pb}(\sqrt{\E
(U^2)})$.
Finally, consider the second term, i.e.,
\begin{align*}
& \frac{1}{n(n-1)}\mathop{\sum\sum}_{1 \leq i \neq j \leq n} W_{2, it, -j}(X^n,\rhat; b) \{m'(t)(r_i - \widehat{r}_i) + \frac 12 m''(\bar t_i)(r_i-t)^2+ \epsilon_i + \rho_i\} W_{jt}(X^n,\rhat; h)(\rhat_j-t).
\end{align*}
Fix $i$ and write
\begin{align*}
    \frac{1}{n-1}\sum_{j = 1, j \neq i}^n W_{jt}(X^n,\rhat; h)(\rhat_j - t) & = e_1^T \widehat{Q}^{-1}_{ht}\frac{1}{n-1}\sum_{j = 1, j \neq i}^n K_{ht}(\rhat_j) g_{ht}(\rhat_j) g^T_{ht}(\rhat_j) \begin{bmatrix} 0 \\ h \end{bmatrix}\\
    & = e_1^T \widehat{Q}^{-1}_{ht}\left(\frac{1}{n-1}-\frac{1}{n}\right)\sum_{j = 1, j \neq i}^n K_{ht}(\rhat_j) g_{ht}(\rhat_j) g^T_{ht}(\rhat_j) \begin{bmatrix} 0 \\ h \end{bmatrix} \\
    & \hphantom{=} -e_1^T\widehat{Q}_{ht}^{-1} \frac{1}{n} K_{ht}(\rhat_i) g_{ht}(\rhat_i) g^T_{ht}(\rhat_i) \begin{bmatrix} 0 \\ h \end{bmatrix} \\
    & \hphantom{=} + e_1^T \widehat{Q}^{-1}_{ht}\frac{1}{n}\sum_{j = 1}^n K_{ht}(\rhat_j) g_{ht}(\rhat_j) g^T_{ht}(\rhat_j) \begin{bmatrix} 0 \\ h \end{bmatrix}
\end{align*}
The third term is equal to zero, while the first two terms are $O(n^{-1})$. Let 
\begin{align*}
    \widehat{Q}_{ht, -i} = (n-1)^{-1} \sum_{j = 1, j \neq i}^n K_{ht}(\rhat_j) g_{ht}(\rhat_j)g^T_{ht}(\rhat_j).
\end{align*} 
We have
\begin{align*}
& \frac{1}{n(n-1)}\mathop{\sum\sum}_{1 \leq i \neq j \leq n} W_{2, it, -j}(X^n,\rhat; b) \cdot S_i \cdot W_{jt}(X^n,\rhat; h)(\rhat_j - t)\\
& = \frac{1}{n^2}\sum_{i=1}^n e_2^\intercal\widehat{Q}_{bt}^{-1} b^{-1} H_{bt}(\rhat_i)\cdot S_i\cdot e_1^T\widehat{Q}^{-1}_{ht} \widehat{Q}_{ht, -i} \begin{bmatrix} 0 \\ h \end{bmatrix}\\
& \hphantom{=} + \frac{1}{n(n-1)^2}\mathop{\sum\sum}_{1 \leq i\neq j \leq n}e_2^\intercal\widehat{Q}^{-1}_{bt} \left\{K_{bt}(\rhat_j) g_{bt}(\rhat_j) g^\intercal_{bt}(\rhat_j) - \widehat{Q}_{bt}\right\} \widehat{Q}_{bt, -j}^{-1} \cdot b^{-1} H_{bt}(\rhat_i)\cdot S_i\\
& \hphantom{+ \frac{1}{n(n-1)^2}\mathop{\sum\sum}_{1 \leq i\neq j \leq n}e_2^\intercal\widehat{Q}^{-1}_{bt} \left\{K_{bt}(\rhat_i) g_{bt}(\rhat_i) g^\intercal_{bt}(\rhat_i) - \widehat{Q}_{bt}\right\}} \qquad \cdot W_{jt}(X^n; \rhat, h) (\rhat_j - t)\\
& \hphantom{=} - \frac{1}{n^2}\sum_{i=1}^ne_2^\intercal\widehat{Q}_{bt}^{-1} b^{-1} H_{bt}(\rhat_i)\cdot S_i\cdot e_1^T\widehat{Q}^{-1}_{ht} K_{ht}(\rhat_i) g_{ht}(\rhat_i)g^T_{ht}(\rhat_i) \begin{bmatrix} 0 \\ h \end{bmatrix} \\
& = O_{\Pb}\left(\frac{\|\rhat - r\|_\infty+\|\E(\rho\mid \rhat, D^n)\|_\infty}{n b} + \frac{1}{n} + \frac{1}{n\sqrt{nb^3}}\right)
\end{align*}
In this light, this term is negligible.
\end{proof}
\subsection{Proof of Proposition \ref{proposition:bc_seive}}
We start by proving the first statement. Consider the decompositions:
\begin{align*}
    & \mu = m(r) + \rho = \Delta_k(\rhat) + \Phi_k(\rhat)^\intercal \beta_k + m'(\overline{r})(r - \rhat) + \rho \\
    & \mu = \Delta_k(r) + \rho + \Phi_k(\rhat)^\intercal\beta_k + (A - \rhat)\dot\Phi_k(\rhat)^\intercal\beta_k + \{\Phi_k(r) - \Phi_k(\rhat) - \dot\Phi_k(\rhat)(A - \rhat)\}^\intercal\beta_k,
\end{align*}
where $\Delta_k(u) = m(u) - \Phi_k(u)^\intercal \beta_k$ and $\overline{r}$ is an intermediate value between $r$ and $\rhat$. In this light, we have
\begin{align*}
   \varphi_1(Z,\rhat,\muhat) & = \dot\Phi_k(\rhat)(A - \rhat) \muhat + \Phi_k(\rhat) Y\\
    & = \dot\Phi_k(\rhat)(A - \rhat) (\muhat - \mu) + \dot\Phi_k(\rhat)(A - \rhat)\{\Delta_{k}(\rhat) + \Phi_k(\rhat)^\intercal \beta_k + m'(\overline{r})(r - \rhat) + \rho\} \\
    & \hphantom{=} + \Phi_k(\rhat)\left\{\Delta_k(r) + \Phi_k(\rhat)^\intercal\beta_k + (A - \rhat)\dot\Phi_k(\rhat)^\intercal\beta_k + \rho + \epsilon \right\} \\
    &\hphantom{=}+\Phi_k(\rhat)\{\Phi_k(r) - \Phi_k(\rhat) - \dot\Phi_k(\rhat)(A - \rhat)\}^\intercal\beta_k\\
    & = \left[\Phi_k(\rhat) \Phi_k(\rhat)^\intercal + (A - \rhat) \{\Phi_k(\rhat) \dot\Phi_k(\rhat)^\intercal + \dot\Phi_k(\rhat) \Phi_k(\rhat)^\intercal\}\right]\beta_k\\
    & \hphantom{=} + \dot\Phi_k(\rhat)(A - \rhat) \{(\muhat - \mu) + \Delta_k(\rhat) + m'(\overline{r})(r - \rhat) + \rho\} \\
    & \hphantom{=} + \Phi_k(\rhat)\left[\Delta_k(r) +  \{\Phi_k(r) - \Phi_k(\rhat) - \dot\Phi_k(\rhat)(r - \rhat)\}^\intercal\beta_k - (A - r)\dot\Phi_k(\rhat)^\intercal \beta_k + \rho + \epsilon \right]
\end{align*}
Therefore,
\begin{align*}
    \Pn (\varphi_1(Z,\rhat,\muhat) - \varphi_2(Z,\rhat)\beta_k) & = \Pn\left[\Phi_k(\rhat)\left\{\Delta_k(r) + \epsilon - \dot\Phi_k(\rhat)^\intercal\beta_k (A - r)  \right\}\right] \\&\hphantom{=}+ \Pn\left[\left\{\Phi_k(\rhat) + \dot\Phi_k(\rhat)(A - \rhat)\right\} \rho\right] \\
    & \hphantom{=} + \Pn\left[\dot\Phi_k(\rhat)(A - \rhat)\left\{(\muhat - \mu) + \Delta_k(\rhat) + m'(\overline{r})(r - \rhat)\right\}\right] \\
    & \hphantom{=} + \Pn\left[\Phi_k(\rhat)\left\{\int_0^1 \ddot\Phi_k(r + u(\rhat - r)) u du \right\}^\intercal\beta_k(\rhat - r)^2\right]
\end{align*}
Consider the first term, we have
\begin{align*}
    & \Pn\left[\Phi_k(\rhat)\left\{\Delta_k(r) + \epsilon - \dot\Phi_k(\rhat)^\intercal\beta_k (A - r)  \right\}\right] \\
    & = \Pn\left[\Phi_k(r)\left\{\Delta_k(r) + \epsilon - \dot\Phi_k(r)^\intercal\beta_k (A - r)  \right\}\right] \\
    & \hphantom{=} - \Pn\left[\Phi_k(r) \left\{\int_0^1 \ddot\Phi_k(r + u(\rhat - r))du\right\}^\intercal (\rhat - r)\beta_k (A - r)  \right] \\
    & \hphantom{=} + \Pn\left[\int_0^1 \dot\Phi_k(r + u(\rhat - r)) du (\rhat - r) \left\{\Delta_k(r) + \epsilon - \dot\Phi_k(\rhat)^\intercal\beta_k (A - r)  \right\}\right].
\end{align*}
Similarly, we expand
\begin{align*}
   & \Pn\left[\left\{\Phi_k(\rhat) + \dot\Phi_k(\rhat)(A - \rhat)\right\} \rho\right] \\
   & = \Pn\left[\left\{\Phi_k(r) + \dot\Phi_k(r)(A - r)\right\}\rho \right] - \Pn \left[\left\{\int_0^1\ddot\Phi_k(r + u(\rhat - r)) u du (\rhat - r)^2\right\}\rho\right]  \\
   & \hphantom{=} + \Pn\left\{\int_0^1 \ddot\Phi_k(r + u(\rhat - r)) du (\rhat - r) (A - r) \rho\right\}.
\end{align*}

Putting everything together, we have reached
\begin{align*}
& \Pn (\varphi_1(Z,\rhat,\muhat) - \varphi_2(Z,\rhat)\beta_k) \\
& = \Pn\left[\Phi_k(r)\left\{\Delta_k(r) + \epsilon - \dot\Phi_k(r)^\intercal\beta_k (A - r)  \right\}\right] + \Pn\left[\left\{\Phi_k(r) + \dot\Phi_k(r)(A - r)\right\}\rho \right] \\
& \hphantom{=} - \Pn\left[\Phi_k(r) \left\{\int_0^1 \ddot\Phi_k(r + u(\rhat - r))du\right\}^\intercal (\rhat - r) \beta_k (A - r)  \right] \\
    & \hphantom{=} + \Pn\left[\int_0^1 \dot\Phi_k(r + u(\rhat - r)) du (\rhat - r) \left\{\Delta_k(r) + \epsilon - \dot\Phi_k(\rhat)^\intercal\beta_k (A - r)  \right\}\right] \\
    & \hphantom{=} - \Pn \left\{\int_0^1\ddot\Phi_k(r + u(\rhat - r)) u du (\rhat - r)^2 \rho\right\} + \Pn\left\{\int_0^1 \ddot\Phi_k(r + u(\rhat - r)) du (\rhat - r) (A - r) \rho\right\} \\
     & \hphantom{=} + \Pn\left[\dot\Phi_k(\rhat)(A - \rhat)\left\{(\muhat - \mu) + \Delta_k(\rhat) + m'(\overline{r})(r - \rhat)\right\}\right] \\
    & \hphantom{=} + \Pn\left[\Phi_k(\rhat)\left\{\int_0^1 \ddot\Phi_k(r + u(\rhat - r)) u du \right\}^\intercal\beta_k(\rhat - r)^2\right] \\
    & \equiv \Pn\left[\Phi_k(r)\left\{\Delta_k(r) + \epsilon - \dot\Phi_k(r)^\intercal\beta_k (A - r)  \right\}\right] + \Pn\left[\left\{\Phi_k(r) + \dot\Phi_k(r)(A - r)\right\}\rho \right] + \sum_{j = 1}^6 R_j.
\end{align*}
Since $\|\Pb( R_1) \|_2= 0$, and
\begin{align*}
    \left |\int_0^1\ddot\Phi_k(r+u(\rhat-r))du^\intercal\beta_k\right | \leq \ \max_t\left|\ddot\Phi_k(t)^\intercal\beta_k\right| =\|m''-\ddot\Delta_k\|_\infty\leq\|\ddot\Delta_k\|_\infty+\|m''\|_\infty,
\end{align*}
we have \begin{align*}
        \Pb\|  R_1\|_2^2 & \leq\frac 1n \Pb\left[\Phi_k(r)^\intercal\Phi_k(r)(r-\rhat)^2(r-A)^2\right](\|\ddot\Delta_k\|_\infty+\|m''\|_\infty)^2\\
        & \lesssim\frac{\xi_k^2}{n} \|r-\rhat\|^2_{2,\Pb}.
\end{align*}         Therefore,
\begin{align*}
    \|R_1\|_2 = O_{\Pb}\left(\frac{\xi_k}{\sqrt n}\|r-\rhat\|_{2,\Pb}\right).
\end{align*}
For $R_2$, similarly we can bound $\dot\Phi_k(\rhat)^\intercal\beta_k$ by
\begin{align*}
        \left |\dot\Phi_k(\rhat )^\intercal\beta_k\right | \leq \ \max_t\left|\dot\Phi_k(t)^\intercal\beta_k\right| \leq\|\dot\Delta_k\|_\infty+\|m'\|_\infty.
\end{align*}
Using lemma \ref{lemma:projection}, we can get
\begin{align*}
    \|\Pb  R_2\|_2& = \left\|\Pb\left[\int_0^1\dot\Phi_k(r+u(\rhat-r))du \cdot(\rhat-r)\Delta_k(r)\right]\right \|_2 \\
    & \leq \|\widetilde Q_{\dot\Phi,k}(r,\rhat)\|_{\rm op}^{1/2}\|(r-\rhat)\Delta_k(r)\|_{2,\Pb}.
\end{align*}
Also,
\begin{align*}
     \Pb\|R_2-\Pb R_2\|_2^2&\lesssim\frac{1}{n}\Pb\|R_2\|^2_2\lesssim \frac {\eta_k^2} {n} \Pb\left[(\rhat-r)^2\left\{\Delta_k(r)^2+\epsilon^2+(\|\dot\Delta_k\|_\infty+\|m'\|_\infty)^2\right\}\right]
\end{align*}
Thus, 
\begin{align*}
      \|  R_2\|_2 = O_{\Pb}\left(\frac{\eta_k}{\sqrt n}\|r-\rhat\|_{2,\Pb}+\|\widetilde Q_{\dot\Phi,k}(r,\rhat)\|_{\rm op}^{1/2}\|(r-\rhat)\Delta_k(r)\|_{2,\Pb}\right).
\end{align*}
For term $R_3$, similarly we have\begin{align*}
    \|\Pb R_3\|_2 &= \left\|\Pb\left[\int_0^1\ddot\Phi_k(r+u(\rhat-r))u du \cdot(\rhat-r)^2\E(\rho\mid r,\rhat, D^n)\right]\right\|_2\\
    &\lesssim\|\widetilde{Q}_{\ddot\Phi,k}(r,\rhat)\|_{\rm op}^{1/2}\|(r-\rhat)^2\E(\rho\mid r,\rhat, D^n)\|_{2,\Pb}.
\end{align*} And \begin{align*}
    \Pb\|R_3-\Pb R_3\|_2^2\lesssim \frac 1n\Pb\|R_3\|_2^2\lesssim\frac{\zeta_k^2\|(r-\rhat)^2\|_{2,\Pb}^2}{n}\sup_{t_1,t_2}|\E(\rho^2\mid r=t_1,\rhat=t_2, D^n)|.
\end{align*}
Then we have\begin{align*}
    \|R_3\|_2 &= O_{\Pb}\Big(\|\widetilde{Q}_{\ddot\Phi,k}(r,\rhat)\|_{\rm op}^{1/2}\|(r-\rhat)^2\E(\rho\mid r,\rhat, D^n)\|_{2,\Pb}\\&\hphantom{ \ = O_{\Pb}\Big(}+\frac{\zeta_k}{\sqrt{n}}\|(r-\rhat)^2\|_{2,\Pb}\cdot\sup_{t_1,t_2}|\E(\rho^2\mid r=t_1,\rhat=t_2, D^n)^{1/2}\Big).
\end{align*}
Notice that $\|\Pb R_4\|_2 = 0$, therefore\begin{align*}
    \|R_4\|_2 = O_{\Pb}\left(\sqrt{\frac 1n\Pb \|R_4\|_2^2}\right) = O_{\Pb} \left(\frac{\zeta_k}{\sqrt{n}}\|r-\rhat\|_{2,\Pb}\cdot\sup_{t_1,t_2}|\E(\rho^2\mid r=t_1,\rhat=t_2, D^n)^{1/2}\right)
\end{align*}
Similarly, we can bound $R_5$\begin{align*}
    \|R_5\|_2 = O_{\Pb}\Bigg(&\|Q_{\dot\Phi,k}(\rhat)\|_{\rm op}^{1/2}\left\{
    \|(r-\hat{r})(\hat\mu-\mu)\|_{2,\mathbb{P}}
    +\|(r-\hat{r})\Delta_k(\hat{r})\|_{2,\mathbb{P}}
    +\|(r-\hat{r})^2\|_{2,\mathbb{P}}\right\}\\
    &+\frac{\eta_k}{\sqrt{n}}\left\{
    \|\hat\mu-\mu\|_{2,\mathbb{P}}
    +\|\Delta_k(\hat{r})\|_{2,\mathbb{P}}
    +\|r-\hat{r}\|_{2,\mathbb{P}}\right\}\Bigg),
\end{align*} And by Lemma \ref{lemma:op_bound_basis} and the assumptions we have $\|Q_{\dot\Phi,k}(\rhat)\|_{\rm op} \lesssim \|Q_{\dot\Phi,k}(r)\|_{\rm op}+o_{\Pb}(\|Q_{\dot\Phi,k}(r)\|_{\rm op})$.

And $R_6$ 
\begin{align*}
    \|R_6\|_2 = O_{\Pb}\left(\|Q_{\Phi,k}(\rhat)\|_{\rm op}^{1/2}\|(r-\rhat)^2\|_{2,\Pb}+\frac{\xi_k}{\sqrt n}\|(r-\rhat)^2\|_{2,\Pb}\right),
\end{align*} where $\|Q_{\Phi,k}(\rhat)\|_{\rm op} \lesssim \|Q_{\Phi,k}(r)\|_{\rm op}+o_{\Pb}(\|Q_{\Phi,k}(r)\|_{\rm op})$ under the assumptions and Lemma \ref{lemma:op_bound_basis}.

Because of Lemma \ref{lemma:op_bound_basis}, under the assumptions, we have $\|\Pn \varphi_2(Z,\rhat) - \Pb\varphi_2(Z,r)\|_{\rm op} = o_{\Pb}(1)$. Therefore, it holds that
\begin{align*}
    & \left\|\widehat\beta_{\rm sbc, k} - \beta_k \right\|_2\lesssim \left\|\Pn\left[\Phi_k(r)\left\{\Delta_k(r) + \epsilon - \dot\Phi_k(r)^\intercal\beta_k (A - r)  \right\}+\left\{\Phi_k(r) + \dot\Phi_k(r)(A - r)\right\}\rho \right]\right\|_2 \\
    & \hphantom{ \left\|\widehat\beta_{\rm sbc, k} - \beta_k \right\|_2\lesssim } + O_{\Pb}\left(\sum_{j = 1}^6 \|R_j\|_2\right)\\
    &\hphantom{\left\|\widehat\beta_{\rm sbc, k} - \beta_k \right\|_2}\lesssim O_\Pb\left(\frac{\xi_k+\eta_k\sqrt{\E\{(A-r)^2\rho^2\}}}{\sqrt n}+\sum_{j = 1}^6 \|R_j\|_2\right)
\end{align*}

And for any unit vector $\alpha$, we have 
\begin{align*}
& \alpha^\intercal\left(\widehat\beta_{\rm sbc, k} - \beta_k \right)\\
&= \alpha^\intercal Q_{\Phi, k}^{-1} \Bigg(
\Pn\left[\Phi_k(r)\left\{\Delta_k(r) + \epsilon - \dot\Phi_k(r)^\intercal\beta_k (A - r) + \rho\right\}\right] + \Pn\left\{\dot\Phi_k(r)(A - r)\rho\right\}
\Bigg)\\
&\hphantom{=} + O_{\Pb}\left(S_n + \sum_{j = 1}^6 \|R_j\|_2\right),
\end{align*}
where
\begin{align*}
    S_n = & \ \alpha^\intercal\left[\{\Pn \widehat\varphi_2(Z, \rhat)\}^{-1} - \{\Pb\varphi_2(Z, r)\}^{-1}\right] \\
    & \qquad \cdot \left(\Pn\left[\Phi_k(r)\left\{\Delta_k(r) + \epsilon - \dot\Phi_k(r)^\intercal\beta_k (A - r) + \rho\right\}\right] + \Pn\left\{\dot\Phi_k(r)(A - r)\rho\right\}\right) 
\end{align*}
A crude bound on $\|S_n\|_2$ is thus
\begin{align*}
    \|S_n\|_2 \lesssim \|\Pn \varphi_2(Z, \rhat) - \Pb \varphi_2(Z,r)\|_{\rm op} \|T_n\|_2,
\end{align*}
where
\begin{align*}
    T_n = \Pn\left[\Phi_k(r)\left\{\Delta_k(r) + \epsilon - \dot\Phi_k(r)^\intercal\beta_k (A - r) + \rho\right\} + \dot\Phi_k(r)(A - r)\rho\right].
\end{align*}
The bound then follows by Lemma \ref{lemma:op_bound_basis} and because
\begin{align*}
    \|T_n\|_2 = O_\Pb\left(\sqrt{\frac{\xi^2_k + \eta^2_k \E\{(A - r)^2 \rho^2\}}{n}}\right).
\end{align*}
\subsection{Proof of Proposition \ref{proposition:cpi_sieve}}

Since
    \begin{align*}
   \lambda_{\rm min}(\widehat{Q}_{\Phi,k}(\rhat))\geq c-\|\widehat{Q}_{\Phi,k}(\rhat)-Q_{\Phi,k}(r)\|_{op},
    \end{align*}
    given lemma \ref{lemma:op_bound_basis} and the assumptions, we have
    \begin{align*}
        &\|\widehat{Q}_{\Phi,k}(\rhat)-Q_{\Phi,k}(r)\|_{\rm op} = o_{\Pb}(1) \quad \text{and} \quad \|\widehat{Q}_{\Phi,q}(\rhat)-Q_{\Phi,q}(r)\|_{\rm op} = o_{\Pb}(1).
    \end{align*}
Therefore the minimum eigenvalues of $\widehat{Q}_{\Phi,k}(\rhat)$ and $\widehat{Q}_{\Phi,q}(\rhat)$ are both bounded away from zero with probability tending to 1. We also have
\begin{align*}
    \min_j \lambda_\min(\widehat{Q}_{\Phi, q, -j}) \geq \frac{n}{n-1}\lambda_\min(\widehat{Q}_{\Phi, q}) - \frac{1}{n-1} \cdot \max_j \|\Phi_q(\rhat_j)\Phi_q(\rhat_j)^\intercal\| \gtrsim 1
\end{align*}
as long as $\xi_q^2 = o(n)$. And, similarly,
\begin{align*}
    \sup_i \|\widehat{Q}_{\Phi, k, -i}(\rhat)\|_{\rm op} \leq \frac{1}{n-1} \sup_{i} \|\Phi_k(\rhat_i)\Phi_k(\rhat_i)^\intercal\|_{\rm op} + \frac{n}{n-1} \|\widehat{Q}_{\Phi, k}(\rhat)\|_{\rm op} \lesssim 1
\end{align*}
with probability tending to 1 as long as $\xi_k^2 = o(n)$.

Recall that $\mu = m(r) + \rho$ and, for shorthand notation, write $\Delta_k(u) = m(u) - \Phi_k(u)^\intercal \beta_k$ and $\dot\Delta_q(u) = m'(u) - \dot\Phi_q(u)^\intercal\beta_q$, with $\beta_q = \E\{\Phi_q(r)\Phi_q(r)^\intercal\}^{-1} \E\{\Phi_q(r) m\}$. 

From the decomposition
\begin{align*}
    \mu & = \rho + \Delta_k(\rhat) + \dot\Delta_q(\rhat) (A - \rhat) - m'(\rhat) (A - r) + \{m'(\bar r)-m'(\rhat)\}(r - \rhat) + \Phi_k(\rhat)^\intercal \beta_k + (A - \rhat)\dot\Phi_q(\rhat)^\intercal \beta_q,
\end{align*}
we have
\begin{align*}
 & \widehat{Q}_{\Phi, k}(\rhat)\left(\widehat\beta_{\rm pi, k} - \beta_k\right) \\
 & = \Pn \left(\Phi_k(\rhat) \left[\Delta_k(\rhat) + \dot\Delta_q(\rhat)(r - \rhat) +\dot\Phi_q(\rhat)^\intercal\beta_q(r-A) + \{m'(\bar r)-m'(\rhat)\}(r - \rhat) + \epsilon + \rho \right]\right) \\
& \hphantom{=} + \Pn \left\{\Phi_k(\rhat)(A - \rhat)\dot\Phi_q(\rhat)^\intercal \beta_q \right\}.
\end{align*}
Then by writing $\Phi_k(\rhat) - \Phi_k(r) = \int_0^1 \dot\Phi_k(r + u(\rhat - r)) du \cdot (\rhat - r)$ and $\dot\Phi_q(\rhat) - \dot\Phi_q(r) = \int_0^1 \ddot\Phi_q(r + u(\rhat - r)) du \cdot (\rhat - r)$, the first term  can be expanded as
\begin{align*}
    &\Pn \left[\Phi_k(r) \left\{\Delta_k(r) + \dot\Phi_q(r)^\intercal\beta_q(r - A) + \epsilon + \rho\right\} \right]\\& + \Pn \left[\Phi_k(\rhat)\int_0^1\ddot\Phi_q(r+u(\rhat-r))du^\intercal\beta_q\ (\rhat - r)(r - A)\right] \\
& + \Pn\left[\int_0^1 \dot\Phi_k(r + u(\rhat - r)) du \cdot (\rhat - r) \left\{\Delta_k(r) + \dot\Phi_q(r)^\intercal\beta_q(r-A) + \epsilon + \rho\right\} \right] \\
& + \Pn\left(\Phi_k(\rhat) \left[\dot\Delta_q(\rhat)(r - \rhat) + \{m'(\bar r)-m'(\rhat)\}(r - \rhat)+\Delta_k(\rhat)-\Delta_k(r)\right]\right).  
\end{align*}
From the decomposition
\begin{align*}
    \mu = \rho + \Delta_q(\rhat) + m'(\overline{r})(r - \rhat) + \Phi_q(\rhat)^\intercal \beta_q,
\end{align*}
we have
\begin{align*}
& \frac{1}{n(n-1)}\mathop{\sum\sum}_{1 \leq i \neq j \leq n} \dot\Phi_q(\rhat_j)^\intercal \widehat{Q}^{-1}_{\Phi, q, -j}(\rhat)\Phi_q(\rhat_i) Y_i \times \widehat{Q}_{\Phi, k}(\rhat)^{-1} \Phi_k(\rhat_j)(A_j - \rhat_j) \\
& = \frac{1}{n(n-1)}\mathop{\sum\sum}_{1 \leq i \neq j \leq n} \dot\Phi_q(\rhat_j)^\intercal \widehat{Q}^{-1}_{\Phi, q, -j}(\rhat)\Phi_q(\rhat_i) \left\{ \Delta_q(\rhat_i) + m'(\overline{r}_i)(r_i - \rhat_i) + \rho_i + \epsilon_i\right\} \times
\\&\hphantom{= \frac{1}{n(n-1)}\mathop{\sum\sum}_{1 \leq i \neq j \leq n} \dot\Phi_q(\rhat_j)^\intercal \widehat{Q}^{-1}_{\Phi, q, -j}(\rhat)\Phi_q(\rhat_i)}\quad\quad\widehat{Q}_{\Phi, k}(\rhat)^{-1} \Phi_k(\rhat_j)(A_j - \rhat_j) \\
& \hphantom{=} + \widehat{Q}_{\Phi, k}(\rhat)^{-1} \Pn\left\{\Phi_k(\rhat)(A - \rhat) \dot\Phi_q(\rhat)^\intercal\beta_q\right\}.
\end{align*}
Therefore, 
\begin{align*}
    \widehat{Q}_{\Phi, k}(\rhat)( \widehat\beta_{\rm scpi, k, q} - \beta_k ) &= \Pn \left[\Phi_k(r) \left\{\Delta_k(r) + \dot\Phi_q(r)^\intercal\beta_q(r - A) + \epsilon+\rho\right\} \right] \\
   &\hphantom{=}+\sum_{i=1}^3   R_i-\frac{1}{n(n-1)} \mathop{\sum\sum}_{1 \leq i \neq j \leq n} T_{ij},
\end{align*}
where \begin{align*}
    &R_1 \equiv \Pn\left[\Phi_k(\rhat)\int_0^1\ddot\Phi_q(r+u(\rhat-r))du^\intercal\beta_q\ (\rhat - r)(r - A)\right],\\
& R_2  \equiv\Pn\left[\int_0^1 \dot\Phi_k(r + u(\rhat - r)) du \cdot (\rhat - r) \left\{\Delta_k(r) + \dot\Phi_q(r)^\intercal\beta_q(r-A) + \epsilon + \rho\right\}\right],  \\
& R_3  \equiv \Pn\left(\Phi_k(\rhat) \left[\dot\Delta_q(\rhat)(r - \rhat) + \{m'(\bar r)-m'(\rhat)\}(r - \rhat)+\Delta_k(\rhat)-\Delta_k(r)\right]\right),  \\
    &T_{ij} \equiv  \dot\Phi_q(\rhat_j)^\intercal \widehat{Q}^{-1}_{\Phi, q, -j}(\rhat)\Phi_q(\rhat_i) \left\{ \Delta_q(\rhat_i) + m'(\overline{r}_i)(r_i - \rhat_i) + \rho_i + \epsilon_i\right\}\Phi_k(\rhat_j) (A_j - \rhat_j).
\end{align*}

Since $\|\Pb( R_1) \|_2= 0$, and
\begin{align*}
    \left |\int_0^1\ddot\Phi_q(r+u(\rhat-r))du^\intercal\beta_q\right | \leq \ \max_t\left|\ddot\Phi_q(t)^\intercal\beta_q\right| =\|m''-\ddot\Delta_q\|_\infty\leq\|\ddot\Delta_q\|_\infty+\|m''\|_\infty,
\end{align*}
we have 
\begin{align*}
        \Pb\|  R_1\|_2^2 & \leq\frac 1n \Pb\left[\Phi_k(\rhat)^\intercal\Phi_k(\rhat)(r-\rhat)^2(r-A)^2\right](\|\ddot\Delta_q\|_\infty+\|m''\|_\infty)^2\\
        & \lesssim\frac{\xi_k^2}{n} \|r-\rhat\|^2_{2,\Pb}.
\end{align*} Therefore,
\begin{align*}
    \|R_1\|_2 = O_{\Pb}\left(\frac{\xi_k}{\sqrt n}\|r-\rhat\|_{2,\Pb}\right).
\end{align*}
 
For $R_2$, similarly we can bound $\dot\Phi_q(r)^\intercal\beta_q$ by
\begin{align*}
        \left |\dot\Phi_q(r)^\intercal\beta_q\right | \leq \ \max_r\left|\dot\Phi_q(r)^\intercal\beta_q\right| \leq\|\dot\Delta_q\|_\infty+\|m'\|_\infty.
\end{align*}
Using lemma \ref{lemma:projection}, we can get
\begin{align*}
    \|\Pb  R_2\|_2& = \left\|\Pb\left[\int_0^1\dot\Phi_k(r+u(\rhat-r))du \cdot(\rhat-r)\left\{\Delta_k(r)+\E(\rho|r,\rhat,D^n)\right\}\right]\right \|_2 \\
    & \leq \|\widetilde Q_{\dot\Phi,k}(r,\rhat)\|_{\rm op}^{1/2}\|(r-\rhat)\{\Delta_k(r)+\E(\rho\mid r,\rhat,D^n)\}\|_{2,\Pb}.
\end{align*}
Also,
\begin{align*}
     \Pb\|R_2-\Pb R_2\|_2^2&\lesssim\frac{1}{n}\Pb\|R_2\|^2_2\lesssim \frac {\eta_k^2} {n} \Pb\left[(\rhat-r)^2\left\{\Delta_k(r)^2+\epsilon^2+\rho^2+(\|\dot\Delta_q\|_\infty+\|m'\|_\infty)^2\right\}\right]
\end{align*}
Thus, 
\begin{align*}
      \|  R_2\|_2 = O_{\Pb}\left(\frac{\eta_k}{\sqrt n}\|r-\rhat\|_{2,\Pb}+\|\widetilde Q_{\dot\Phi,k}(r,\rhat)\|_{\rm op}^{1/2}\|(r-\rhat)\left\{\Delta_k(r)+\E(\rho|r,\rhat,D^n)\right\}\|_{2,\Pb}\right)
\end{align*}

For $R_3$, we have that $|\Delta_k(\rhat)-\Delta_k(r)|\leq\|\dot\Delta_k\|_\infty|r-\rhat|$ and $|m'(\bar r)-m'(\rhat)|\leq \|m''\|_\infty|r-\rhat|$. Applying Lemma \ref{lemma:projection}, we can show that
\begin{align*}
    \|\Pb  R_3\|_2& \leq \left\|\Pb\left[\Phi_k(\rhat)(\rhat-r)\dot\Delta_q(\rhat)+\Phi_k(\rhat)(r-\rhat)^2+\Phi_k(\rhat)|r-\rhat|\|\dot\Delta_k\|_\infty\right]\right\|_2\\
    &\lesssim\|Q_{\Phi,k}(\rhat)\|_{\rm op}^{1/2}\left\{\|(r-\rhat)\dot\Delta_q(\rhat)\|_{2,\Pb}+\|(r-\rhat)^2\|_{2,\Pb}+\|r-\rhat\|_{2,\Pb}\|\dot\Delta_k\|_\infty\right\}
\end{align*}
And \begin{align*}
    \Pb\|R_3-\Pb R_3\|_2^2&\lesssim\frac{\xi_k^2}{n}\Pb\left[(r-\rhat)^2\left\{\|\dot\Delta_k\|^2_\infty+\dot\Delta_q(\rhat)^2+(r-\rhat)^2\right\}\right].
\end{align*}
Therefore,
\begin{align*}
    \| R_3\|_2 &= O_{\Pb}\left(\left(\|Q_{\Phi,k}(\rhat)\|^{1/2}_{\rm op}+\frac{\xi_k}{\sqrt n}\right)\left\{\|(r-\rhat)\dot\Delta_q(\rhat)\|_{2,\Pb}+\|\dot\Delta_k\|_\infty\|r-\rhat\|_{2,\Pb}+\|(r-\rhat)^2\|_{2,\Pb}\right\}\right).
\end{align*}
Using the notation $\rhat^n = \rhat(X_1), \ldots, \rhat(X_n)$, we expand $T_{ij}$ as:
\begin{align*}
    \frac{1}{n(n-1)} \mathop{\sum\sum}_{1 \leq i \neq j \leq n}T_{ij} & = \frac{1}{n(n-1)} \mathop{\sum\sum}_{1 \leq i \neq j \leq n}U_{ij}+ \frac{1}{n} \sum_{j = 1}^n V_j+ \frac{1}{n} \sum_{i = 1}^n W_i+ \overline{U} , 
\end{align*}
where \begin{align*}
   & U_{ij} = T_{ij} - \E(T_{ij} \mid A_j, \rhat^n, D^n) -\E(T_{ij} \mid A_i, \epsilon_i, X_i, \rhat^n, D^n) + \E(T_{ij}\mid \rhat^n, D^n), \\
    &V_{j} = \frac{1}{n-1} \sum_{i \neq j}^n \left\{\E(T_{ij} \mid A_j, \rhat^n, D^n) - \E(T_{ij} \mid \rhat^n, D^n)\right\}, \\
    &W_{i}  = \frac{1}{n-1}\sum_{j \neq i}^n \left\{\E(T_{ij} \mid A_i, X_i, \epsilon_i, \rhat^n, D^n) -  \E(T_{ij} \mid \rhat^n, D^n)\right\}, \\
    & \overline{U} = \frac{1}{n(n-1)} \mathop{\sum\sum}_{1 \leq i \neq j \leq n} \E(T_{ij} \mid \rhat^n, D^n).
\end{align*}
By the law of total expectation, all three terms are mean 0 \begin{align*}
    \E(U_{ij}\mid \rhat^n, D^n)= \E(V_j\mid \rhat^n, D^n) = \E(W_i\mid \rhat^n, D^n) = 0.
\end{align*}

Note that,
 \begin{align*}
     \E(U_{ij}\mid A_j, \rhat^n, D^n) &= \E(T_{ij}\mid A_j,\rhat^n, D^n) - \E(T_{ij} \mid A_j, \rhat^n, D^n)\\
     &\ \ \ \ -\E\left\{\E(T_{ij} \mid A_i, X_i, \epsilon_i, \rhat^n, D^n)\mid A_j, \rhat^n, D^n\right\} + \E(T_{ij}\mid \rhat^n, D^n) \\&= 0.
 \end{align*}
 And similarly $\E(U_{ij}\mid A_i, X_i, \epsilon_i, \rhat^n, D^n) = 0$. By construction, $U_{ij}$ are uncorrelated with $V_j$ and $W_i$.
And for $k\neq j$,
\begin{align*}
    \E(U_{ij}^\intercal U_{ik}\mid \rhat^n, D^n) = \E\left\{\E(U_{ij}\mid A_i, X_i, \epsilon_i, \rhat^n, D^n)^\intercal\E(U_{ik}\mid A_i, X_i, \epsilon_i, \rhat^n, D^n)\mid \rhat^n, D^n\right\} = 0
\end{align*} 
Following the same reasoning,  $\E(U_{ij}^\intercal U_{kj}\mid \rhat^n, D^n) = 0$ for $i \neq k$.
Therefore,
\begin{align*}
   &\E\left\{\left\| \frac{1}{n(n-1)} \mathop{\sum\sum}_{1 \leq i \neq j \leq n} T_{ij}\right\|^2_2 \mid \rhat^n, D^n\right\} \\
   & = \E\left\{\left\| \frac{1}{n(n-1)} \mathop{\sum\sum}_{1 \leq i \neq j \leq n} U_{ij} + \frac{1}{n}\sum_{i = 1}^n (V_i + W_i) + \overline{U} \right\|^2_2 \mid \rhat^n, D^n\right\}   \\
   & = \E\left\{\left\|\frac{1}{n(n-1)} \mathop{\sum\sum}_{1 \leq i \neq j \leq n} U_{ij}\right\|_2^2 \mid \rhat^n, D^n\right\} +\E\left\{\left\|\frac{1}{n}\sum_{i = 1}^n (V_i + W_i) \right\|_2^2 \mid \rhat^n, D^n\right\} +\|\overline{U}\|_2^2\\
   & = \frac{1}{n^2(n-1)^2}\mathop{\sum\sum}_{1 \leq i \neq j \leq n} \E(\|U_{ij}\|_2^2 + U_{ij}^\intercal U_{ji}  \mid \rhat^n, D^n) + \frac{1}{n^2} \sum_{i = 1}^n \E(\|V_i+W_i\|_2^2 \mid \rhat^n, D^n) + \|\overline{U}\|_2^2
\end{align*}
Taking the expectation with respect to $\rhat^n$ given $D^n$, we have
\begin{align*}
    & \E\left\{\left\| \frac{1}{n(n-1)} \mathop{\sum\sum}_{1 \leq i \neq j \leq n} T_{ij}\right\|^2_2 \mid D^n\right\} \\
    & = \frac{1}{n(n-1)}\E(\|U_{ij}\|_2^2 + U_{ij}^\intercal U_{ji}  \mid D^n) + \frac{1}{n} \E(\|V_i+W_i\|_2^2 \mid D^n) + \E(\|\overline{U}\|_2^2 \mid D^n).
\end{align*}
Let $s_i = \Delta_q(\rhat_i)+ \E\left\{m'(\bar r_i)(r_i-\rhat_i) \mid \rhat_i, D^n\right\} +\E(\rho_i|\rhat_i,D^n)$ and 
\begin{align*}
    B_{-j, n} = \widehat{Q}^{-1}_{\Phi,q, -j} \frac{1}{n-1}\sum_{i = 1, i \neq j}^n \Phi_{q}(\rhat_i) s_i.
\end{align*}
We have
\begin{align*}
    \overline{U} = \frac{1}{n(n-1)} \mathop{\sum\sum}_{1 \leq i \neq j \leq n} \E(T_{ij}\mid \rhat^n, D^n) = \frac{1}{n} \sum_{j = 1}^n \Phi_k(\rhat_j) \E(r_j - \rhat_j \mid \rhat_j, D^n) \dot\Phi_q(\rhat_j)^\intercal B_{-j,n}
\end{align*}
Notice that
\begin{align*}
  B^\intercal_{-j, n} \widehat{Q}_{\Phi, q, -j} B_{-j, n} =  \frac{1}{n-1}\sum_{i = 1, i \neq j}^n \{ \Phi_q(\rhat_i)^\intercal B_{-j, n}\}^2 \leq \frac{1}{n-1}\sum_{i = 1, i \neq j}^n s_i^2,
\end{align*}
so that on the event that $\min_j \lambda_\min (\widehat{Q}_{\Phi, q, -j}) \gtrsim 1$, we have
\begin{align*}
    \|B_{-j, n}\|_2^2 \ \lesssim \frac{1}{n-1}\sum_{i = 1, i \neq j}^n s_i^2 \implies \max_{1 \leq j \leq n} \|B_{-j, n}\|_2^2 \ \leq \frac{n}{n-1} \Pb_n s^2 \lesssim \Pn s^2.
\end{align*}
Furthermore,
\begin{align*}
    \|\overline{U}\|_2^2 & \lesssim \|\widehat{Q}_{\Phi, k}(\rhat)\|_{\rm op} \frac{1}{n}\sum_{j = 1}^n \{\E(r_j - \rhat_j \mid \rhat_j, D^n)\}^2B_{-j, n}^\intercal\dot\Phi_{q}(\rhat_j) \dot\Phi_q(\rhat_j)^\intercal B_{-j, n} 
\end{align*}
This then yields
\begin{align*}
    & \E(\|\overline{U}\|_2^2 \mid D^n) \\
    & \lesssim \min \left\{ \|\rhat - r\|^2_\infty \|Q_{\dot\Phi, q}(\rhat)\|_{\rm op}, \|\rhat - r\|^2_{2, \Pb} \eta^2_q\right\} \cdot \|Q_{\Phi, k}(\rhat)\|_{\rm op} \cdot \{\|\Delta_q(\rhat)\|^2_{2, \Pb} + \|\rhat - r\|^2_{2, \Pb} + \|E(\rho \mid \rhat, D^n)\|^2_{2, \Pb}\}.
\end{align*}
Similarly, we have
\begin{align*}
   V_j = \Phi_k(\rhat_j)\cdot \{(A_j - \rhat_j) - \E(r_j - \rhat_j \mid \rhat_j, D^n)\} \cdot \dot\Phi_q(\rhat_j)^\intercal B_{-j, n}
\end{align*}
so that
\begin{align*}
    \E(\|V_j\|_2^2 \mid \rhat^n_{-j}, D^n) \lesssim \xi^2_k \cdot \|Q_{\dot\Phi, q}(\rhat)\|_{\rm op} \cdot \|B_{-j, n}\|^2_2 
\end{align*}
and
\begin{align*}
    \E(\|V_j\|^2 \mid D^n) \lesssim \xi^2_k \cdot \|Q_{\dot\Phi, q}(\rhat)\|_{\rm op} \cdot \{\|\Delta_q(\rhat)\|^2_{2, \Pb} + \|\rhat - r\|^2_{2, \Pb} + \|E(\rho \mid \rhat, D^n)\|^2_{2, \Pb}\}.
\end{align*}
Next, letting $w_i = \Delta_q(\rhat_i) + m'(\overline{r}_i)(r_i - \rhat_i) + \rho_i + \epsilon_i$ and $\overline{w}_i = w_i - \E(w_i \mid \rhat^n, D^n)$, we have
\begin{align*}
    W_i = \frac{1}{n-1} \sum_{j = 1, j \neq i}^n \E(r_j - \rhat_j \mid \rhat_j, D^n) \Phi_k(\rhat_j) \dot\Phi_q(\rhat_j)^\intercal \widehat{Q}_{\Phi, q, -j}^{-1}(\rhat) \Phi_q(\rhat_i) \overline{w}_i.
\end{align*}
Notice that, for any $f_j$, we have, on the event that $\|\widehat{Q}_{\Phi, k, -i}(\rhat)\|_{\rm op} \lesssim 1$,
\begin{align*}
   \left\|\frac{1}{n-1} \sum_{j = 1, j \neq i}^n\Phi_k(\rhat_j) f_j\right\|_2^2 \leq \|\widehat{Q}_{\Phi, k, -i}(\rhat)\|_{\rm op} \frac{1}{n-1} \sum_{j = 1, j \neq i}^n f_j^2 \lesssim  \frac{1}{n-1} \sum_{j = 1, j \neq i}^n f_j^2.
\end{align*}
Applying the inequality above with $f_j = \E(r_j - \rhat_j \mid \rhat_j, D^n) \dot\Phi_q(\rhat_j)^\intercal \widehat{Q}_{\Phi, q, -j}^{-1}(\rhat) \Phi_q(\rhat_i) \overline{w}_i$, we have
\begin{align*}
    \|W_i\|_2^2 \lesssim \ \overline{w}^2_i \cdot \Phi_q(\rhat_i)^\intercal \left(\frac{1}{n-1} \sum_{j = 1, j \neq i}^n \left[\{\E(r_j - \rhat_j \mid \rhat_j, D^n)\}^2 \widehat{Q}_{\Phi, q, -j}^{-1}(\rhat) \dot\Phi_q(\rhat_j) \dot\Phi_q(\rhat_j)^\intercal \widehat{Q}_{\Phi, q, -j}^{-1}(\rhat)\right] \right) \Phi_q(\rhat_i).
\end{align*}
Therefore, on the even that $\min_j \lambda_\min (\widehat{Q}_{\Phi, q, -j}) \gtrsim 1$, 
\begin{align*}
    \E(\|W_i\|_2^2 \mid D^n) \lesssim \xi^2_q \cdot \min\left\{\|Q_{\dot\Phi, q}(\rhat)\|_{\rm op} \|\rhat - r\|^2_\infty, \eta^2_q \|\rhat - r\|^2_{2, \Pb}\right\}.
\end{align*}
Finally, conditioning on the event that $\min_j \lambda_\min(\widehat{Q}_{\Phi, q, -j}) \gtrsim 1$,
\begin{align*}
\E(\|U_{ij}\|_2^2 + U_{ij}^\intercal U_{ji} \mid D^n) & \lesssim \E(\|U_{ij}\|_2^2 + \|U_{ji}\|_2^2 \mid D^n) \\
& \lesssim \E(\|T_{ij}\|_2^2 + \|T_{ji}\|_2^2 \mid D^n) \\
& \lesssim \E\left\{\frac{\xi_k^2}{n}\sum_{j = 1}^n \dot\Phi_{q}(\rhat_j)^\intercal \widehat{Q}_{\Phi, q, -j}^{-1} \dot\Phi_q(\rhat_j) \mid D^n\right\} \\
& \lesssim \xi_k^2 \cdot \eta_q^2.
\end{align*}
By Assumption \ref{assump:sieve_eigen} and Lemma \ref{lemma:op_bound_basis}, we have
$
    \|Q_{\Phi,k}(\rhat)\|_{\rm op} = \|Q_{\Phi,k}(r)\|_{\rm op} + o_\Pb(1)$, $
        \|Q_{\Phi,q}(\rhat)\|_{\rm op} = \|Q_{\Phi,q}(r)\|_{\rm op} + o_\Pb(1)$ and $
         \|Q_{\dot\Phi,q}(\rhat)\|_{\rm op} = \|Q_{\dot\Phi,q}(r)\|_{\rm op}+o_{\Pb}(\|Q_{\dot\Phi,q}(r)\|_{\rm op}).
$

Putting everything together, and invoking Lemma \ref{lemma:op_bound_basis} for the second term below, we have reached
\begin{align*}
\widehat\beta_{\rm scpi, k, q} - \beta_k & = Q_{\Phi, k}(r)^{-1} \Pn \left[ \Phi_k(r)\left\{\Delta_k(r) + \dot\Phi_q(r)^\intercal\beta_q(r - A) + \epsilon + \rho\right\} \right] \\
& \qquad + \ \{\widehat{Q}^{-1}_{\Phi, k}(\rhat) - Q_{\Phi, k}(r)^{-1}\} \Pn \left[ \Phi_k(r)\left\{\Delta_k(r) + \dot\Phi_q(r)^\intercal\beta_q(r - A) + \epsilon + \rho\right\} \right] \\
& \qquad + \ \widehat{Q}_{\Phi, k}(\rhat)^{-1} \left\{ \sum_{i=1}^3   R_i-\frac{1}{n(n-1)} \mathop{\sum\sum}_{1 \leq i \neq j \leq n} T_{ij} \right\}\\
& = Q_{\Phi, k}(r)^{-1} \Pn \left[ \Phi_k(r)\left\{\Delta_k(r) + \dot\Phi_q(r)^\intercal\beta_q(r - A) + \epsilon + \rho\right\} \right] \\
& \qquad + O_\Pb\left(\frac{\xi_k}{\sqrt{n}} \cdot \left(\sqrt{\frac{\xi_k^2\cdot\log k}{n}} +\|\rhat-r\|_\infty \|\widetilde{Q}_{\dot\Phi,k}(r,\rhat)\|^{1/2}_{\rm op}\right)\right) \\
& \qquad + \ O_{\Pb}\Bigg(
\frac{\xi_k}{\sqrt n}\|r-\rhat\|_{2,\Pb} \\
& \hphantom{\widehat\beta_{\rm scpi, k} - \beta O_{\Pb}\Bigg[\,} + \frac{\eta_k}{\sqrt n}\|r-\rhat\|_{2,\Pb}+\|\widetilde Q_{\dot\Phi,k}(r,\rhat)\|_{\rm op}^{1/2}\|(r-\rhat)\left\{\Delta_k(r)+\E(\rho|r,\rhat,D^n)\right\}\|_{2,\Pb} \\
&\hphantom{\widehat\beta_{\rm scpi, k} - \beta O_{\Pb}\Bigg[\,}
+ \left(1 + \frac{\xi_k}{\sqrt n} \right) \left\{\|(r-\rhat)\dot\Delta_q(\rhat)\|_{2,\Pb}+\|\dot\Delta_k\|_\infty\|r-\rhat\|_{2,\Pb}+\|(r-\rhat)^2\|_{2,\Pb}\right\}\\
&\hphantom{\widehat\beta_{\rm scpi, k} - \beta O_{\Pb}\Bigg[\,}
+ \frac{\xi_k \eta_q}{n} + \frac{\xi_k}{\sqrt{n}} \cdot \|Q_{\dot\Phi, q}(\rhat)\|^{1/2}_{\rm op} \cdot \{\|\Delta_q(\rhat)\|_{2, \Pb} + \|\rhat - r\|_{2, \Pb} + \|E(\rho \mid \rhat, D^n)\|_{2, \Pb}\}\\
&\hphantom{\widehat\beta_{\rm scpi, k} - \beta O_{\Pb}\Bigg[\,}
+ \frac{\xi_q}{\sqrt{n}} \cdot \min\left\{\|Q_{\dot\Phi, q}(\rhat)\|^{1/2}_{\rm op} \|\rhat - r\|_\infty, \eta_q \|\rhat - r\|_{2, \Pb}\right\}
\\
&\hphantom{\widehat\beta_{\rm scpi, k} - \beta O_{\Pb}\Bigg[\,}
+ \min \left\{ \|\rhat - r\|_\infty \|Q_{\dot\Phi, q}(\rhat)\|^{1/2}_{\rm op}, \|\rhat - r\|_{2, \Pb} \eta_q\right\} \\
& \left. \vphantom{\frac{\xi_k}{\sqrt{n}}} \hphantom{\widehat\beta_{\rm scpi, k} - \beta O_{\Pb}\Bigg[\,} \qquad \qquad \cdot \{\|\Delta_q(\rhat)\|_{2, \Pb} + \|\rhat - r\|_{2, \Pb} + \|E(\rho \mid \rhat, D^n)\|_{2, \Pb}\} \right)
\end{align*}
\section{Supporting Lemmas}
\begin{lemma}\label{lemma:min_eigen_local}
    Let $\lambda_\min(M)$ denote the smallest eigenvalue of a symmetric matrix $M$. Suppose that $\lambda_\min(\E\{K_{ht}(r) g_{ht}(r) g_{ht}(r)^\intercal\}) \geq c$, for some constant $c > 0$. Then,
    \begin{align*}
    & \|\Pn \varphi_2(Z,\rhat) - \Pb\varphi_2(Z,r)\|_{\rm op} = O_{\Pb}\left(\frac{1}{\sqrt{n}h^2} \cdot(h+\|r-\rhat\|_\infty)^{1/2}+ \frac{\|\rhat - r\|_\infty^2}{h^3}\cdot(h+\|r-\rhat\|_\infty)\right), \text{ and } \\
    & \lambda_\min\left\{\Pn \varphi_2(Z,\rhat)\right\} \geq c - O_{\Pb}\left(\frac{1}{\sqrt{n}h^2} \cdot(h+\|r-\rhat\|_\infty)^{1/2}+ \frac{\|\rhat - r\|_\infty^2}{h^3}\cdot(h+\|r-\rhat\|_\infty)\right).
    \end{align*}
\end{lemma}
\begin{proof}
    For any matrices $M$ and $\widehat{M}$, we have
    \begin{align*}
        \lambda_\min(\widehat{M}) \geq \lambda_\min(M) - \|\widehat{M} - M\|_{\rm op} \geq \lambda_\min(M) - \|\widehat{M} - M\|_{\rm F}
    \end{align*}
    where $\|M\|^2_{\rm F} = \sum_{j, l} M^2_{jl}$ denotes (the square of) the Frobenius norm. Thus, it is sufficient to bound
    \begin{align*}
\|\Pn \varphi_2(Z,\rhat) - \Pb\varphi_2(Z,r)\|_{\rm F} = \|(\Pn - \Pb)\varphi_2(Z,\rhat)  + \Pb\{\varphi_2(Z,\rhat) - \varphi_2(Z,r)\}\|_{\rm F}.
    \end{align*}

        By direct calculation, we have
    \begin{align*}
        G'_{ht}(s) = K'_{ht}(s) \begin{bmatrix}
            1 & \frac{s - t}{h} \\
            \frac{s - t}{h} & \left(\frac{s - t}{h}\right)^2
        \end{bmatrix} + K_{ht}(s) \cdot \frac{1}{h}  \begin{bmatrix}
            0 & 1 \\
            1 & \frac{2(s - t)}{h}
        \end{bmatrix},
    \end{align*}

    and  \begin{align*}
        G''_{ht}(s) = K''_{ht}(s) \begin{bmatrix}
            1 & \frac{s - t}{h} \\
            \frac{s - t}{h} & \left(\frac{s - t}{h}\right)^2
        \end{bmatrix} + 2\cdot K'_{ht}(s) \cdot \frac{1}{h}  \begin{bmatrix}
            0 & 1 \\
            1 & \frac{2(s - t)}{h}
        \end{bmatrix}+K_{ht}(s)\cdot
         \begin{bmatrix}
            0 & 0 \\
            0 & \frac 2 {h^2}
        \end{bmatrix}.
    \end{align*}
 So component-wise, we have $G'_{ht}(s)\lesssim h^{-2}\cdot 1 (|s-t|\leq h)$ and $G''_{ht}(s)\lesssim\frac 1 {h^3} \cdot1(|s-t|\leq h)$.
    
    We have
    \begin{align*}
        \Pb \{\varphi_2(Z,\rhat) - \varphi_2(Z,r)\}&=\Pb\left\{ G_{ht}'(\rhat) (r - \rhat) + G_{ht}(\rhat) - G_{ht}(r)\right\} \\& = -\Pb\left\{\int_0^1 G''(r + u(\rhat - r))(r - \rhat)^2 u du\right\} \\
        & = O_{\Pb}\left(\frac{\|\rhat - r\|^2_\infty}{h^3}\cdot (h+\|r-\rhat\|_\infty)\right),
    \end{align*}
    where the $O_{\Pb}(\cdot)$ statement holds component-wise. Similarly, we have
    \begin{align*}
       \E(\|(\Pn - \Pb)\varphi_2(Z,\rhat)\|^2_{\rm F} \mid D^n) \leq \frac {\Pb(\|\varphi_2(Z,\rhat)\|^2_{\rm F})}{n}\lesssim \frac {\E\left[\|G'_{ht}(\rhat)\|_{\rm F}^2 \mid D^n\right]}{n}
    \end{align*}
 Thus, we have reached that
    \begin{align*}
        \|\Pn \varphi_2(Z,\rhat) - \Pb \varphi_2(Z,r)\|_{\rm F} = O_{\Pb}\left(\frac{1}{\sqrt{n}h^2} \cdot(h+\|r-\rhat\|_\infty)^{1/2}+ \frac{\|\rhat - r\|_\infty^2}{h^3}\cdot(h+\|r-\rhat\|_\infty)\right).
    \end{align*}
    and the two statements of the lemma follow.
\end{proof}
\begin{lemma}\label{lemma:op_bound_basis}
Suppose $\sup_t\|\Phi_k(t)\|_2 \lesssim \xi_k$,  $\sup_t\|\dot\Phi_k(t)\|_2 \lesssim \eta_k$ and $\sup_t\|\ddot\Phi_k(t)\|_2 \lesssim \zeta_k$. Then, it holds that:
\begin{align*}
  &\text{(1) }  \|Q_{\Phi,k}(\rhat)-Q_{\Phi,k}(r)\|_{\rm op} 
  = O_{\Pb}\Big(
  \|\widetilde{Q}_{\dot\Phi,k}(r,\rhat)\|_{\rm op}^{1/2}\cdot\min\left\{\|r-\rhat\|_{\infty}\|Q_{\Phi,k}(r)\|_{\rm op}^{1/2}, \xi_k\|r-\rhat\|_{2,\Pb}\right\}\\
&\hphantom{\text{(1) }  \|Q_{\Phi,k}(\rhat)-Q_{\Phi,k}(r)\|_{\rm op} 
  = O_{\Pb}\Big(}+\min\left\{\|r-\rhat\|^2_\infty \|\widetilde Q_{\dot\Phi,k}(r,\rhat)\|_{\rm op},\eta_k^2\|r-\rhat\|_{2,\Pb}^2\right\}
  \Big),\\
 &\text{(2)  }  \|Q_{\dot\Phi,k}(\rhat)-Q_{\dot\Phi,k}(r)\|_{\rm op} 
  =  O_{\Pb}\Big(
  \|\widetilde{Q}_{\ddot\Phi,k}(r,\rhat)\|_{\rm op}^{1/2}\cdot\min\left\{\|r-\rhat\|_{\infty}\|Q_{\dot\Phi,k}(r)\|_{\rm op}^{1/2}, \eta_k\|r-\rhat\|_{2,\Pb}\right\}\\
 &\hphantom{\text{(2) }  \|Q_{\Phi,k}(\rhat)-Q_{\Phi,k}(r)\|_{\rm op} 
  = O_{\Pb}\Big(}+\min\left\{\|r-\rhat\|^2_\infty \|\widetilde Q_{\ddot\Phi,k}(r,\rhat)\|_{\rm op},\zeta_k^2\|r-\rhat\|_{2,\Pb}^2\right\}
  \Big),\\
   &\text{(3) } 
    \|\widehat Q_{\Phi,k}(\rhat)-Q_{\Phi,k}(r)\|_{\rm op}  
    =O_{\Pb}\Bigg(
    \frac{\xi_k^2\log k}{n}
    +\sqrt{\frac{\xi_k^2\cdot\|Q_{\Phi,k}(\rhat)\|_{\rm op}\cdot\log k}{n}} \\
   &\hphantom{   \|\widehat Q_{\Phi,k}(\rhat)-Q_{\Phi,k}(r)\|_{\rm op}  =O_{\Pb}\Bigg(}
   +\|\rhat-r\|_\infty \|\widetilde{Q}_{\dot\Phi,k}(r,\rhat)\|^{1/2}_{\rm op}
   +\|\rhat-r\|^2_\infty \|\widetilde Q_{\dot\Phi,k}(r,\rhat)\|_{\rm op}
   \Bigg),\\
&\text{(4) } \|\Pn\varphi_2(Z,\rhat) - \Pb \varphi_2(Z,r)\|_{\rm op} 
= O_{\Pb}\Bigg(
    \frac{(\xi_k \eta_k+\xi_k^2)\log k}{n} 
    + \sqrt{\frac{\xi_k^2 \|Q_{\Phi,k}(\rhat)\|_{\rm op}\log k}{n}} \\
&\hphantom{\|\Pn\varphi_2(Z,\rhat) - \Pb \varphi_2(Z,r)\|_{\rm op} 
= O_{\Pb}\Bigg(}
    + \sqrt{\frac{\big(\eta_k^2 \|Q_{\Phi, k}(\rhat)\|_{\rm op} 
    + \xi_k^2 \|Q_{\dot\Phi, k}(\rhat)\|_{\rm op}\big)\log k}{n}} \\
&\hphantom{\|\Pn\varphi_2(Z,\rhat) - \Pb \varphi_2(Z,r)\|_{\rm op} 
= O_{\Pb}\Bigg(}
    + \min\Bigg\{
        \|\rhat - r\|_\infty^2 
        \Big(
            \|\widetilde{Q}_{\dot\Phi,k}(r,\rhat)\|_{\rm op} 
            + \|\widetilde{Q}_{\Phi,k}(r,\rhat)\|_{\rm op}^{1/2}
              \|\widetilde{Q}_{\ddot\Phi,k}(r,\rhat)\|_{\rm op}^{1/2}
        \Big), \\
&\hphantom{\|\Pn\varphi_2(Z,\rhat) - \Pb \varphi_2(Z,r)\|_{\rm op} 
= O_{\Pb}\Bigg(+ \min\Bigg\{}
        \|r-\rhat\|_{2,\Pb}^2(\eta_k^2+\xi_k\zeta_k)
    \Bigg\}
\Bigg).
\end{align*}

where
\begin{align*}& \varphi_2(Z,r) = \left\{\dot\Phi_k(r) \Phi_k(r)^\intercal + \Phi_k(r) \dot\Phi_k(r)^\intercal\right\}(A - r) + \Phi_k(r)\Phi_k(r)^\intercal,\\
   &\widetilde{Q}_{\Phi,k}(r,\rhat) = \Pb\left\{\int_0^1 \Phi_k(r + u(\rhat - r))\Phi^\intercal_k(r + u(\rhat - r)) du \right\}. 
   \\& \widetilde{Q}_{\dot\Phi,k}(r,\rhat) = \Pb\left\{\int_0^1 \dot\Phi_k(r + u(\rhat - r))\dot\Phi^\intercal_k(r + u(\rhat - r)) du \right\}\text{, and} \\ 
    &\widetilde{Q}_{\ddot\Phi,k}(r,\rhat) = \Pb\left\{\int_0^1 \ddot\Phi_k(r + u(\rhat - r))\ddot\Phi^\intercal_k(r + u(\rhat - r)) du \right\}.
\end{align*}
\end{lemma}
\begin{proof}
First, we notice that, for any $v$ and $s_u = r + u(\rhat - r)$, we have
\begin{align*}
    \int_0^1 \left\{v^\intercal \Phi_k(s_u) \right\}^2 du \geq \left\{\int_0^1 v^\intercal \Phi_k(s_u) du\right\}^2 \implies \|\widetilde{Q}_{\Phi, k}(r, \rhat)\|_{\rm op} \geq \left\|\Pb\left\{\int_0^1 \Phi_k(s_u) du \int_0^1 \Phi_k(s_u)^\intercal du\right\}\right\|_{\rm op}.
\end{align*}
By lemma \ref{lemma:projection} and the inequality above, we have
\begin{align*}
&\left\|\Pb\left[\left(\int_0^1\dot\Phi_k(\rhat+u(r-\rhat))du \right)\Phi_k(r)^\intercal(\rhat-r)\right]\right\|_{\rm op}\\ =& \sup _{v:\|v\|_2 = 1}\left\|\Pb\left[\left(\int_0^1\dot\Phi_k(\rhat+u(r-\rhat))du \right)\Phi_k(r)^\intercal(\rhat-r)v\right]\right\|_2\\\lesssim &\sup _{v:\|v\|_2 = 1}\|\widetilde{Q}_{\dot\Phi,k}(r,\rhat)\|_{\rm op}^{1/2} \sqrt{\Pb\{(r-\rhat)^2 v^\intercal\Phi_k(r) \Phi_k(r)^\intercal v\}}\\
    \lesssim&\|\widetilde{Q}_{\dot\Phi,k}(r,\rhat)\|_{\rm op}^{1/2}\cdot\min\left\{\|r-\rhat\|_{\infty}\|Q_{\Phi,k}(r)\|_{\rm op}^{1/2}, \xi_k\|r-\rhat\|_{2,\Pb}\right\}.
\end{align*}
Then by the expansion of $\Phi_k(\rhat)$, we have that
\begin{align*}\|Q_{\Phi,k}(\rhat)-Q_{\Phi,k}(r)\|_{\rm op} &\leq 2\left\|\Pb\left[\left(\int_0^1\dot\Phi_k(\rhat+u(r-\rhat))du \right)\Phi_k(r)^\intercal(\rhat-r)\right]\right\|_{\rm op}\\
&\hphantom{\leq}+\min\left\{\|r-\rhat\|^2_\infty \|\widetilde Q_{\dot\Phi,k}(r,\rhat)\|_{\rm op},\eta_k^2\|r-\rhat\|_{2,\Pb}^2\right\} \\ &\lesssim\|\widetilde{Q}_{\dot\Phi,k}(r,\rhat)\|_{\rm op}^{1/2}\cdot\min\left\{\|r-\rhat\|_{\infty}\|Q_{\Phi,k}(r)\|_{\rm op}^{1/2}, \xi_k\|r-\rhat\|_{2,\Pb}\right\}\\
&\hphantom{\leq}+\min\left\{\|r-\rhat\|^2_\infty \|\widetilde Q_{\dot\Phi,k}(r,\rhat)\|_{\rm op},\eta_k^2\|r-\rhat\|_{2,\Pb}^2\right\}.\label{equ:Q(rhat)-Q(r)}
  \end{align*}

 Similarly, we can show the second statements
 \begin{align*}
    \|Q_{\dot\Phi,k}(\rhat)-Q_{\dot\Phi,k}(r)\|_{\rm op} &=O_{\Pb}\Big(\|\widetilde{Q}_{\ddot\Phi,k}(r,\rhat)\|_{\rm op}^{1/2}\cdot\min\left\{\|r-\rhat\|_{\infty}\|Q_{\dot\Phi,k}(r)\|_{\rm op}^{1/2}, \eta_k\|r-\rhat\|_{2,\Pb}\right\}
 \\&\hphantom{\|Q_{\dot\Phi,k}(\rhat)-Q_{\dot\Phi,k}(r)\|_{\rm op} =O_{\Pb}\Big(}\hphantom{\leq}+\min\left\{\|r-\rhat\|^2_\infty \|\widetilde Q_{\ddot\Phi,k}(r,\rhat)\|_{\rm op},\zeta_k^2\|r-\rhat\|_{2,\Pb}^2\right\}\Big)
 \end{align*}

 Using Lemma 6.2 in \cite{belloni2015some}, we obtain
\begin{align*}
    &\|\widehat{Q}_{\Phi,k}(\rhat)-Q_{\Phi,k}(\rhat)\|_{\rm op} = O_{\Pb}\left(\frac{\xi_k^2\log k}{n} +\sqrt{\frac{\xi_k^2\|Q_{\Phi,k}(\rhat)\|_{\rm op}\log k}{n}}\right)
\end{align*}
Then, combine the above together, we get the third statement.

For the fourth statement, notice that
    \begin{align*}
    & \|\Pn \varphi_2(Z,\rhat) - \Pb\varphi_2(Z,r)\|_{\rm op}\\
    & \leq \|(\Pn - \Pb) \Phi_k(\rhat)\Phi_k(\rhat)^\intercal\|_{\rm op} + \left\|(\Pn{-\Pb)}\left[\{\dot\Phi_k(\rhat)\Phi_k(\rhat)^\intercal + \Phi_k(\rhat)\dot\Phi_k(\rhat)^\intercal\}(A - \rhat)\right]\right\|_{\rm op} \\
    & \qquad + \|\Pb\{\varphi_2(Z,\rhat) - \varphi_2(Z,r)\}\|_{\rm op}.
    \end{align*}

    The first term is just $\|\widehat{Q}_{\Phi,k}(\rhat)-Q_{\Phi,k}(\rhat)\|_{\rm op}$. To bound the second term, we rely on Bernstein inequality for matrices (Theorem 5.4.1 in \cite{Vershynin_2026}). 
   We have
    \begin{align*}
        \left\|\left\{\dot\Phi_k(\rhat)\Phi_k(\rhat)^\intercal + \Phi_k(\rhat)\dot\Phi_k(\rhat)^\intercal\right\}(A - \rhat) \right\|_{\rm op} \lesssim \|\dot\Phi_k(\rhat)\|_2 \|\Phi_k(\rhat)\|_2 \ \lesssim \xi_k\cdot\eta_k.
    \end{align*}
    and
    \begin{align*}
    & \left\{\dot\Phi_k(\rhat)\Phi_k(\rhat)^\intercal + \Phi_k(\rhat)\dot\Phi_k(\rhat)^\intercal\right\}^2 \\
    & = \|\dot\Phi_k(\rhat)\|^2_2 \cdot \Phi_k(\rhat)\Phi_k(\rhat)^\intercal + \|\Phi_k(\rhat)\|^2_2 \cdot \dot\Phi_k(\rhat)\dot\Phi_k(\rhat)^\intercal + \{\dot\Phi_k(\rhat)^\intercal\Phi_k(\rhat)\}\{\Phi_k(\rhat)\dot\Phi_k(\rhat)^\intercal + \dot\Phi_k(\rhat)\Phi_k(\rhat)^\intercal\}.
    \end{align*}
Notice that, for any $v \in \R^k$, we have
\begin{align*}
  & v^\intercal \{\dot\Phi_k(\rhat)^\intercal\Phi_k(\rhat)\}\{\Phi_k(\rhat)\dot\Phi_k(\rhat)^\intercal + \dot\Phi_k(\rhat)\Phi_k(\rhat)^\intercal\} v \\
  & = 2  \{\dot\Phi_k(\rhat)^\intercal\Phi_k(\rhat)\} \{v^\intercal \Phi_k(\rhat)\} \{v^\intercal \dot\Phi_k(\rhat)\} \\
  & \leq 2\|\dot\Phi_k(\rhat)\|_2\|\Phi_k(\rhat)\|_2\cdot |v^\intercal \Phi_k(\rhat)| \cdot |v^\intercal \dot\Phi_k(\rhat)| \\
  & \leq \|\dot\Phi_k(\rhat)\|_2^2 \{v^\intercal\Phi_k(\rhat)\}^2 + \|\Phi_k(\rhat)\|_2^2\{v^\intercal \dot\Phi_k(\rhat)\}^2 \\
  & = v^\intercal \left\{\|\dot\Phi_k(\rhat)\|^2_2 \cdot \Phi_k(\rhat)\Phi_k(\rhat)^\intercal + \|\Phi_k(\rhat)\|^2_2 \cdot \dot\Phi_k(\rhat)\dot\Phi_k(\rhat)^\intercal \right\} v
\end{align*}
so that
\begin{align*}
  & \left\| \Pb \left\{\dot\Phi_k(\rhat)\Phi_k(\rhat)^\intercal + \Phi_k(\rhat)\dot\Phi_k(\rhat)^\intercal\right\}^2 \right\|_{\rm op} \\
  & \lesssim \left\|\Pb\left\{\|\dot\Phi_k(\rhat)\|^2_2 \cdot \Phi_k(\rhat)\Phi_k(\rhat)^\intercal + \|\Phi_k(\rhat)\|^2_2 \cdot \dot\Phi_k(\rhat)\dot\Phi_k(\rhat)^\intercal \right\} \right\|_{\rm op} \\
  & \lesssim \eta_k^2 \cdot \left\| \Pb\left\{\Phi_k(\rhat)\Phi_k(\rhat)^\intercal  \right\}\right\|_{\rm op} + \xi_k^2 \cdot \left\| \Pb\left\{\dot\Phi_k(\rhat)\dot\Phi_k(\rhat)^\intercal \right\}\right\|_{\rm op}
\end{align*}
Thus, we have reached
\begin{align*}
& \left\|(\Pn - \Pb)\left[\{\dot\Phi_k(\rhat)\Phi_k(\rhat)^\intercal + \Phi_k(\rhat)\dot\Phi_k(\rhat)^\intercal\}(A - \rhat)\right]\right\|_{\rm op} \\
& = O_{\Pb}\left(\frac{\xi_k \cdot \eta_k \cdot \log k}{n} + \sqrt{\frac{\left(\eta_k^2 \cdot \left\| Q_{\Phi, k}(\rhat) \right\|_{\rm op}  + \xi_k^2 \cdot \left\| Q_{\dot\Phi, k}(\rhat)\right\|_{\rm op} \right) \log k}{n}}\right).
\end{align*}

Finally, we have
\begin{align*}
    \|\Pb\{\varphi_2(Z,\rhat) - \varphi_2(Z,r)\}\|_{\rm op} = \left\|- \E\left\{\int_0^1\frac{\partial^2}{\partial s^2} \Phi_k(s)\Phi^\intercal_k(s)\Big|_{s = r + u(\rhat - r)}u du \cdot (\rhat - r)^2 \mid D^n\right\}\right\|_{\rm op}
\end{align*}
We have
\begin{align*}
   \frac{\partial^2}{\partial s^2} \Phi_k(s)\Phi^\intercal_k(s) = 2\dot\Phi_k(s)\dot\Phi_k(s)^\intercal + \Phi_k(s)\ddot\Phi_k(s)^\intercal + \ddot\Phi_k(s)\Phi_k(s)^\intercal
\end{align*}
yielding that
\begin{align*}
    \|\Pb\{\varphi_2(Z,\rhat) - \varphi_2(Z,r)\}\|_{\rm op} & \leq 2\|\rhat - r\|^2_\infty \cdot \left\|\E\left\{ \int_0^1\dot\Phi_k(s_u)\dot\Phi^\intercal_k(s_u)u du\mid D^n\right\}\right\|_{\rm op} \\
    & \hphantom{=} + 2 \|\rhat - r\|_\infty^2 \left\|\E\left\{ \int_0^1\ddot\Phi_k(s_u)\ddot\Phi^\intercal_k(s_u) udu\mid D^n\right\}\right\|^{1/2}_{\rm op} \\
    & \qquad\qquad\qquad\quad \times \left\|\E\left\{ \int_0^1\Phi_k(s_u)\Phi^\intercal_k(s_u)u du\mid D^n\right\}\right\|^{1/2}_{\rm op}
\end{align*}
where $s_u = r + u(\rhat - r)$.
Another bound also holds:
\begin{align*}
    \|\Pb\{\varphi_2(Z,\rhat) - \varphi_2(Z,r)\}\|_{\rm op} & \leq 2\|\rhat - r\|^2_{2,\Pb}\cdot (\eta_k^2  +\xi_k\zeta_k)
\end{align*}
\end{proof}
\begin{lemma}\label{lemma:projection} Let $\psi$ be a random vector and assume that $Q_\psi = \Pb(\psi\psi^\intercal)$ is positive definite. For any function $f$, define the projection operator onto the space spanned by $\psi$:
\begin{align*}
    \Pi_{\psi}(f) = \psi^\intercal Q_{\psi}^{-1}\Pb(\psi f).
\end{align*} Then it holds that
\begin{align*}
    \|\Pb(\psi f)\|_2 \leq \|Q_{\psi}\|^{1/2}_{\rm op}\|f\|_{2,\Pb}
\end{align*}
    \begin{proof}    
        By inserting $Q_{\psi}^{1/2} Q_{\psi}^{-1/2}$ and applying the operator norm bound, we have:
        \begin{align*}
            \|\Pb (\psi f)\|_2 &= \|Q_{\psi}^{1/2} Q_{\psi}^{-1/2}\Pb (\psi f)\|_2\\
            &\leq\|Q_{\psi}\|^{1/2}_{\rm op}\cdot\|Q_{\psi}^{-1/2}\Pb(\psi f)\|_2\\
            &=\|Q_{\psi}\|^{1/2}_{\rm op}\cdot\sqrt{\Pb( f\psi^\intercal) Q_{\psi}^{-1}\Pb(\psi f)}\\
            & = \|Q_{\psi}\|^{1/2}_{\rm op}\cdot\sqrt{\Pb[\{\Pi_{\psi}(f)\}^2]}\\
            &\leq\|Q_{\psi}\|^{1/2}_{\rm op}\cdot\|f\|_{2,\Pb}.
        \end{align*}

        The last inequality is because 
    \begin{align*}
        \Pb[\psi\{f-\Pi_\psi(f)\}] = \Pb(\psi f) - \Pb(\psi\psi^\intercal)Q_\psi^{-1}\Pb(\psi f) = 0,
    \end{align*}
    $\Pi_\psi(f)$ is an orthogonal projection operator onto the space spanned by $\psi$. Therefore, by Pythagorean theorem $\|f\|_{2,\Pb}^2 = \|f-\Pi_\psi(f)\|_{2,\Pb}^2+\|\Pi_\psi(f)\|_{2,\Pb}^2\geq \|\Pi_\psi(f)\|_{2,\Pb}^2$.
    \end{proof}
\end{lemma}
\section{Additional Simulation results}
\begin{figure}[htbp]
    \centering
    \includegraphics[width=0.8\textwidth]{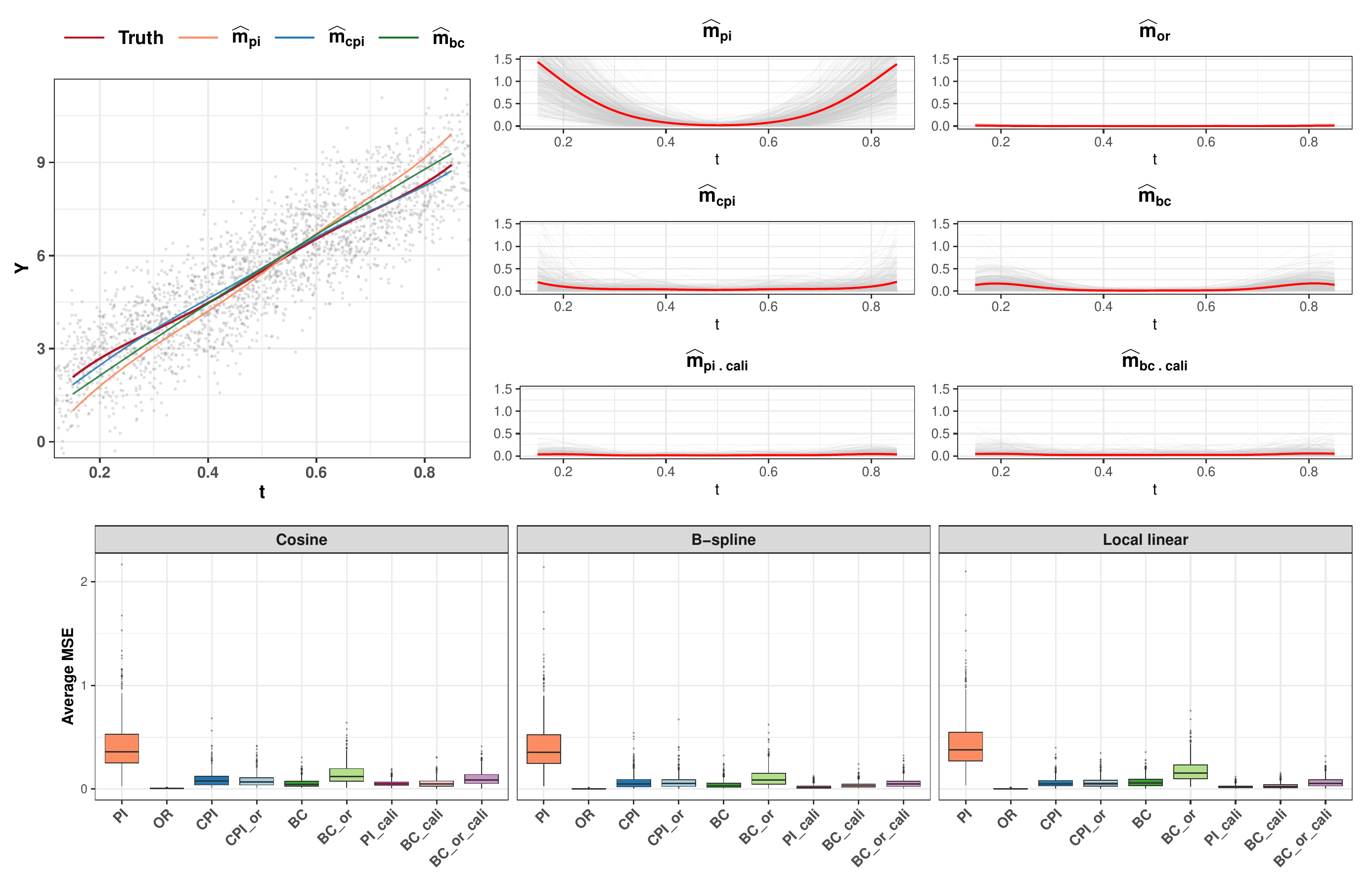}

    \caption{
    \textbf{M2 (linear model with small sinusodal perturbation + linear term in X):} $\mu(X) = 4+ 3r(X)+0.15\sin(4\pi r(X))+\beta_g^T X$, where $\beta_g$ is a fixed coefficient vector. 
    }
    \label{fig:scenarioII-M2}
\end{figure}

\begin{figure}[htbp]
    \centering
    \includegraphics[width=0.8\textwidth]{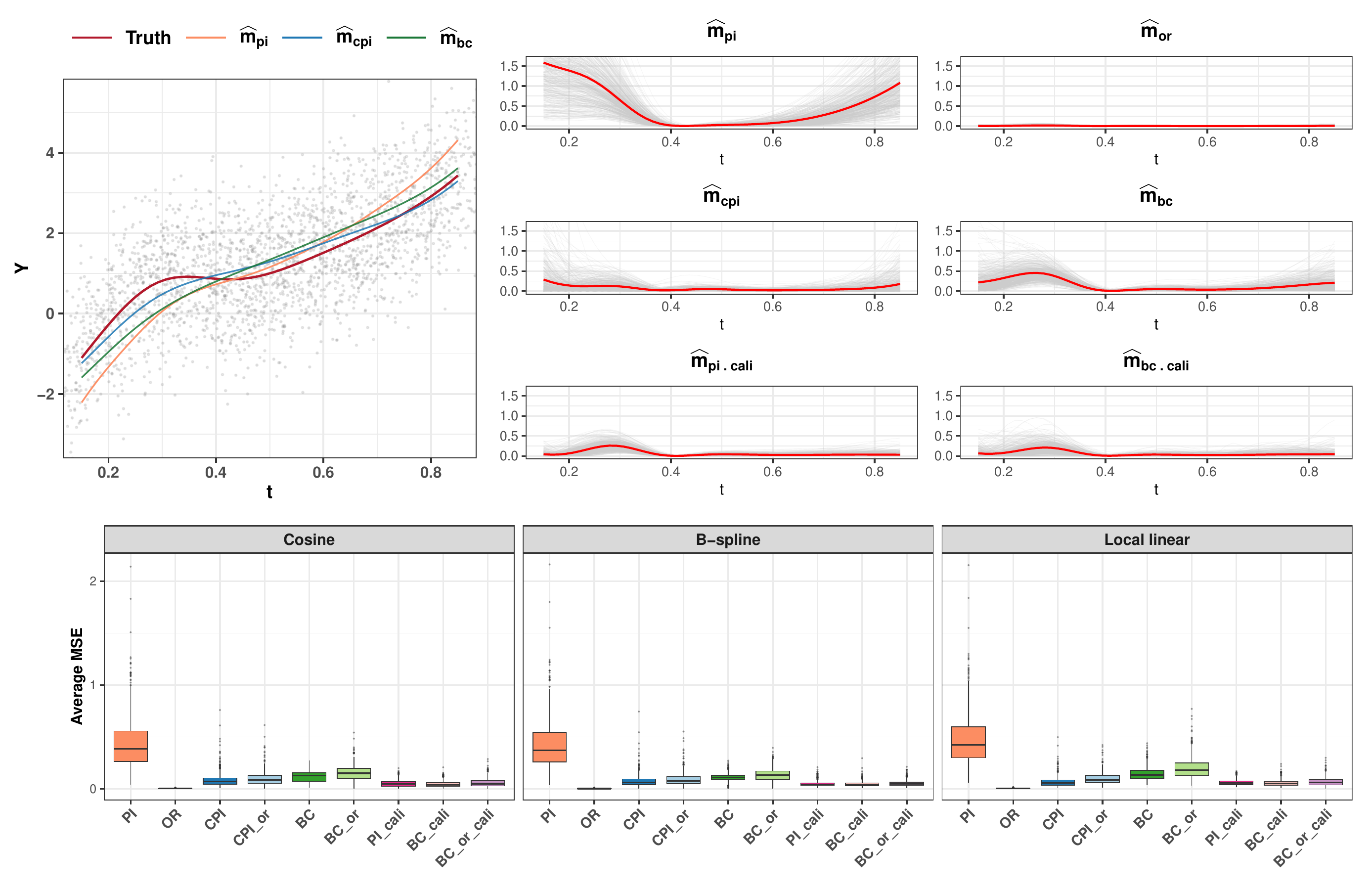}

    \caption{
    \textbf{M3 (Sinusoidal function with a localized bump + linear term in $X$):} $\mu(X) = 1 + 0.1\sin\{2\pi r(X)\} +   \sin\{2\pi r(X)\} *\exp[-40\{r(X) - 0.3\}^2]+\beta_g^T X$, where $\beta_g$ is a fixed coefficient vector. 
    }
    \label{fig:scenarioII-M3}
\end{figure}

\begin{figure}[htbp]
    \centering
\includegraphics[width=0.8\textwidth]{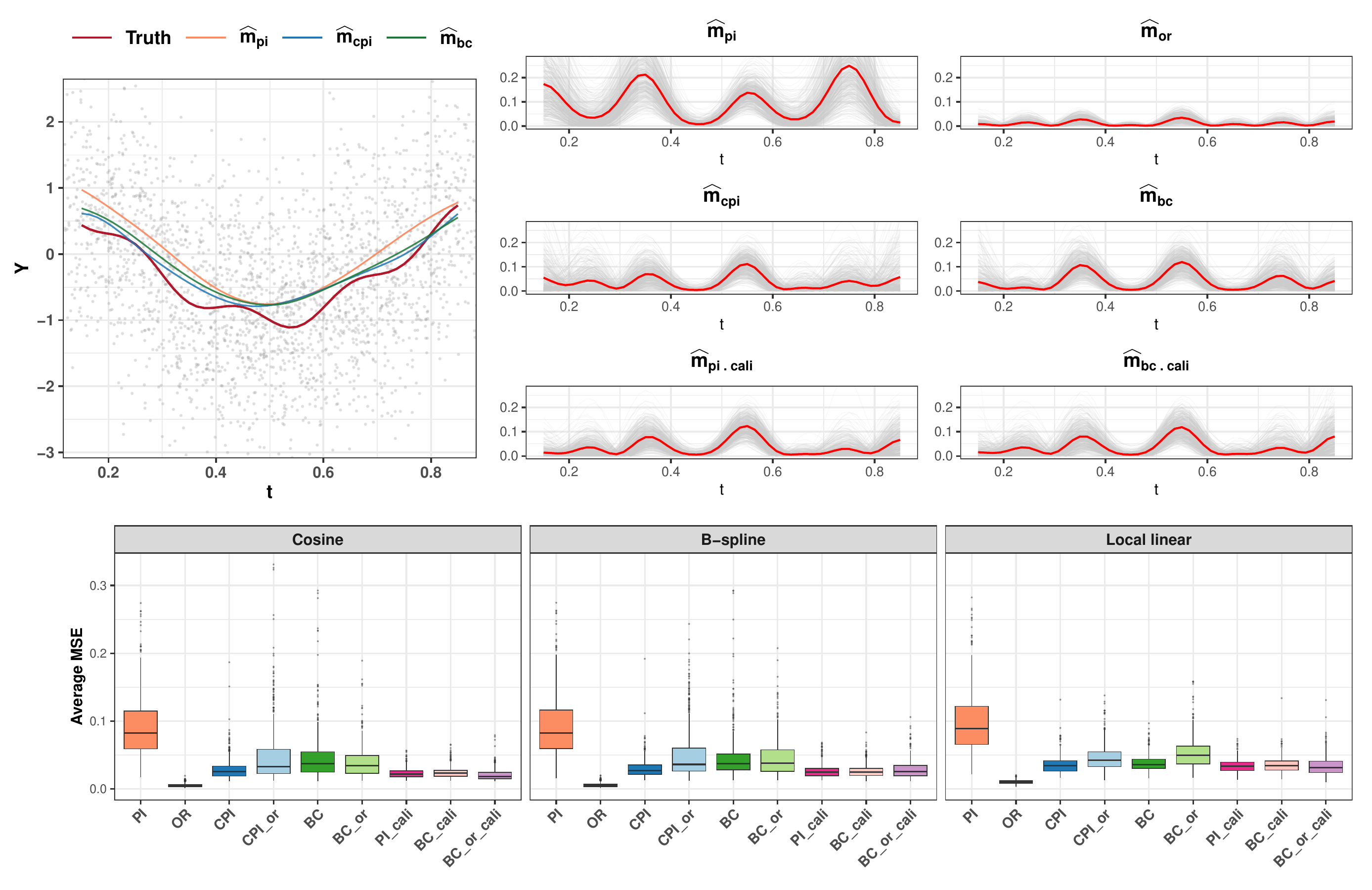}

    \caption{
    \textbf{M1 (sine–cosine function):} 
    $\mu(X) = 0.15\sin\{10\pi r(X)\} + \cos\{2\pi r(X)\}$
    }
    \label{fig:scenarioI-M1}
\end{figure}

\begin{figure}[htbp]
    \centering
    \includegraphics[width=0.8\textwidth]{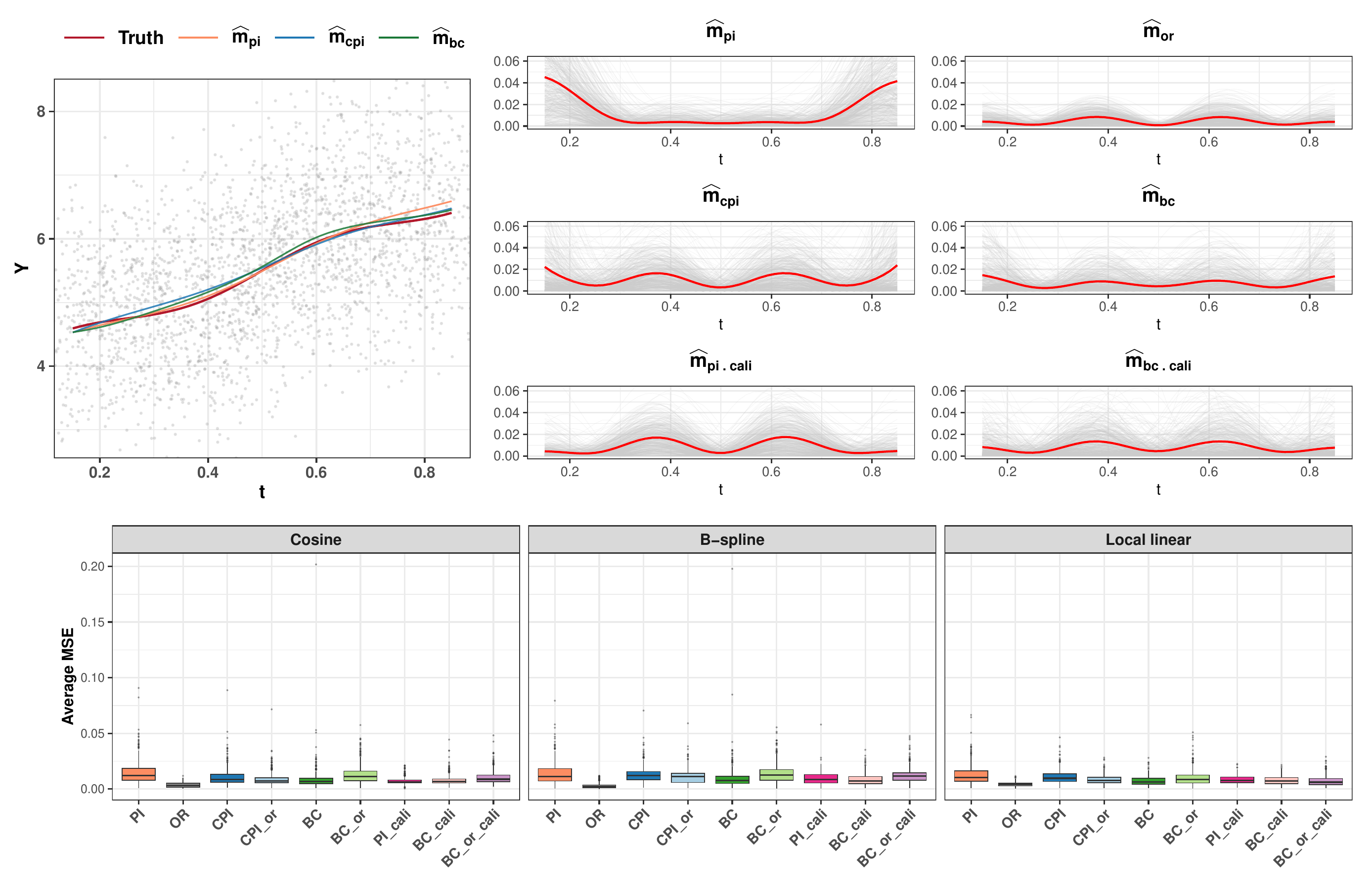}

    \caption{
    \textbf{M2 (linear model with small sinusodal perturbation): $\mu(X) = 4+ 3r(X)+0.15\sin\{4\pi r(X)\}$} 
    }
    \label{fig:scenarioI-M2}
\end{figure}

\begin{figure}[htbp]
    \centering
    \includegraphics[width=0.8\textwidth]{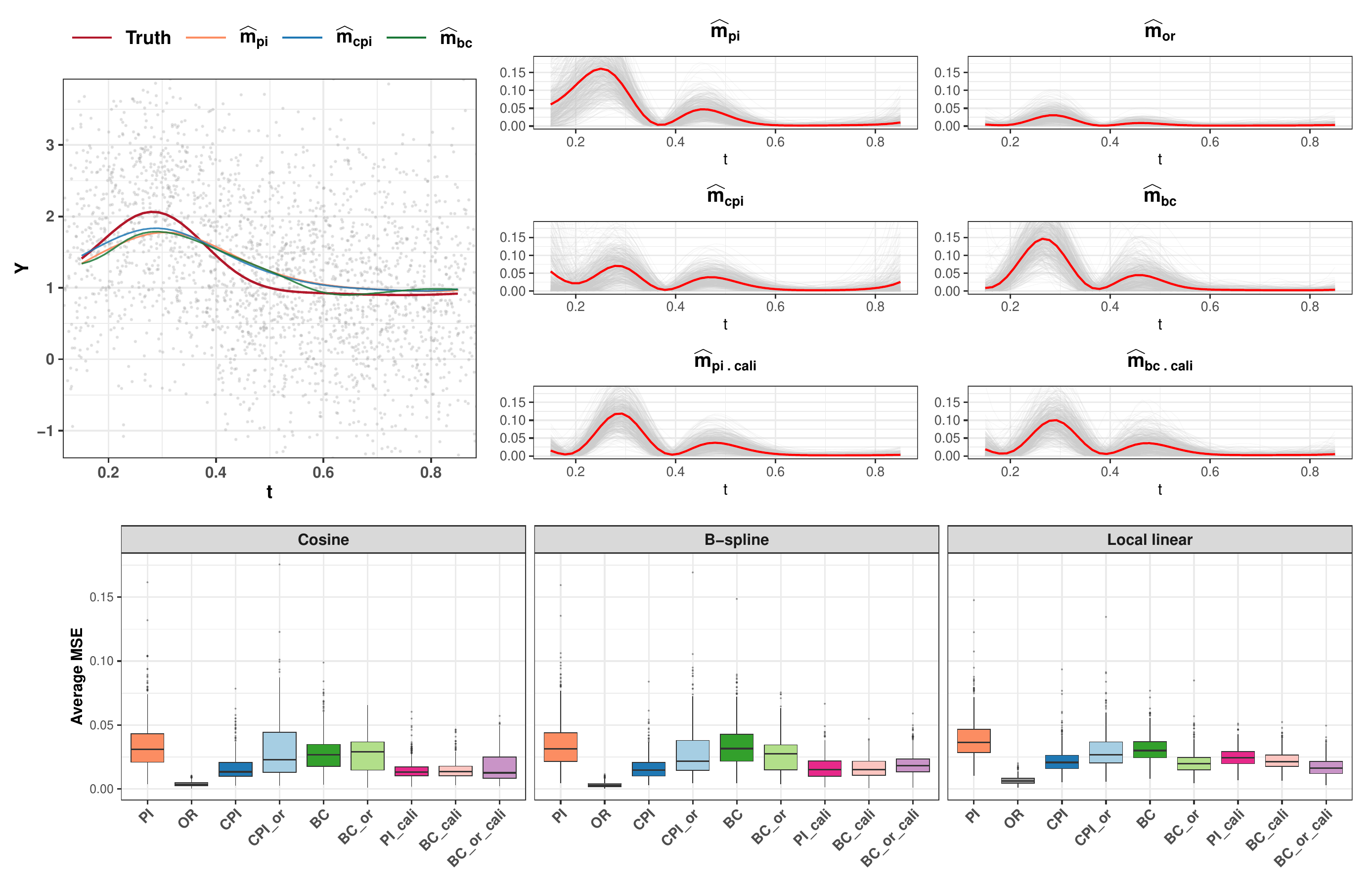}

    \caption{
    \textbf{M3 (Sinusoidal function with a localized bump): $\mu(X) = 1 + 0.1\sin\{2\pi r(X)\} +   \sin\{2\pi r(X)\} *\exp[-40\{r(X) - 0.3\}^2]$} 
    }
    \label{fig:scenarioI-M3}
\end{figure}

\begin{figure}[htbp]
    \centering
    \includegraphics[width=\textwidth]{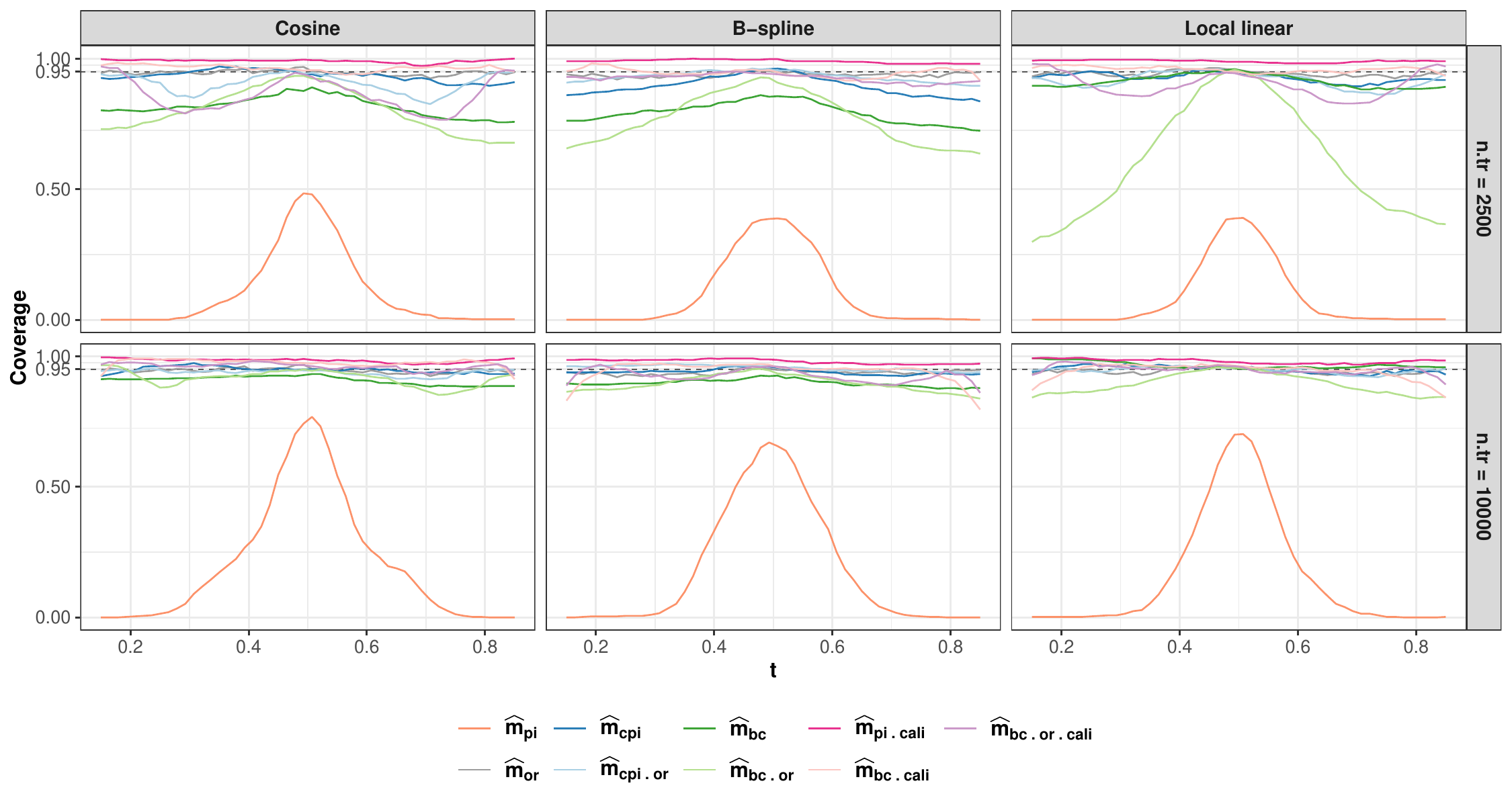}
    \caption{Pointwise coverage relative to the matched same-replication semi-oracle target for  \textbf{M2 (linear model with small sinusodal perturbation + linear term in X):} $\mu(X) = 4+ 3r(X)+0.15\sin(4\pi r(X))+\beta_g^T X$, where $\beta_g$ is a fixed coefficient vector. }
    \label{fig:coverage-ntr_S2M2}
\end{figure}

\begin{figure}[htbp]
    \centering
    \includegraphics[width=\textwidth]{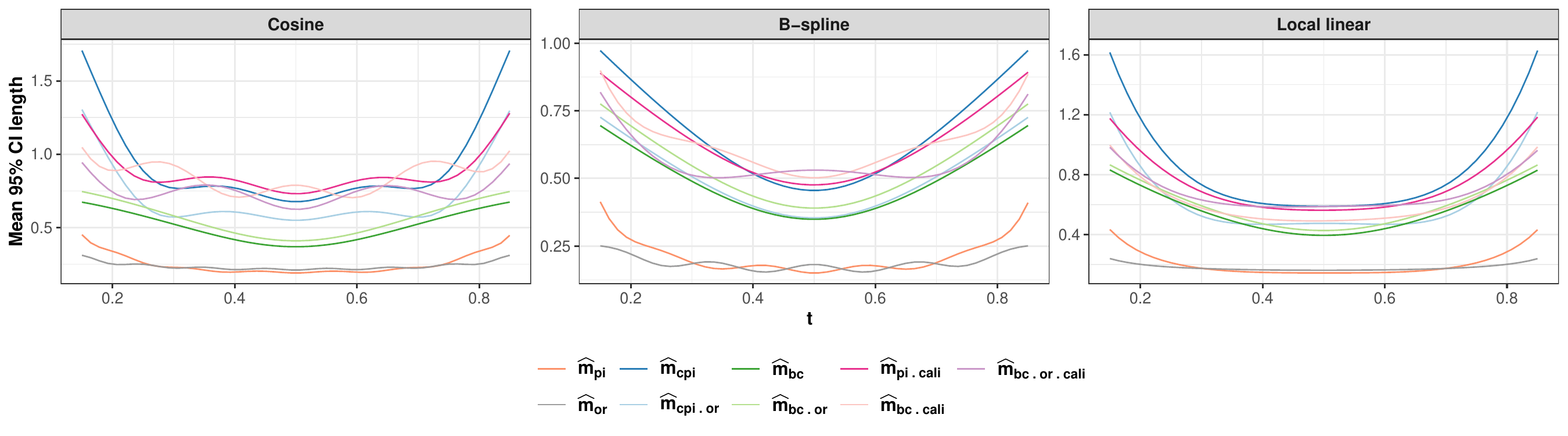}
    \caption{Pointwise confidence interval length for  \textbf{M2 (linear model with small sinusodal perturbation + linear term in X):} $\mu(X) = 4+ 3r(X)+0.15\sin(4\pi r(X))+\beta_g^T X$, where $\beta_g$ is a fixed coefficient vector. }
    \label{fig:cilength-ntr_S2M2}
\end{figure}

\begin{figure}[htbp]
    \centering
    \includegraphics[width=\textwidth]{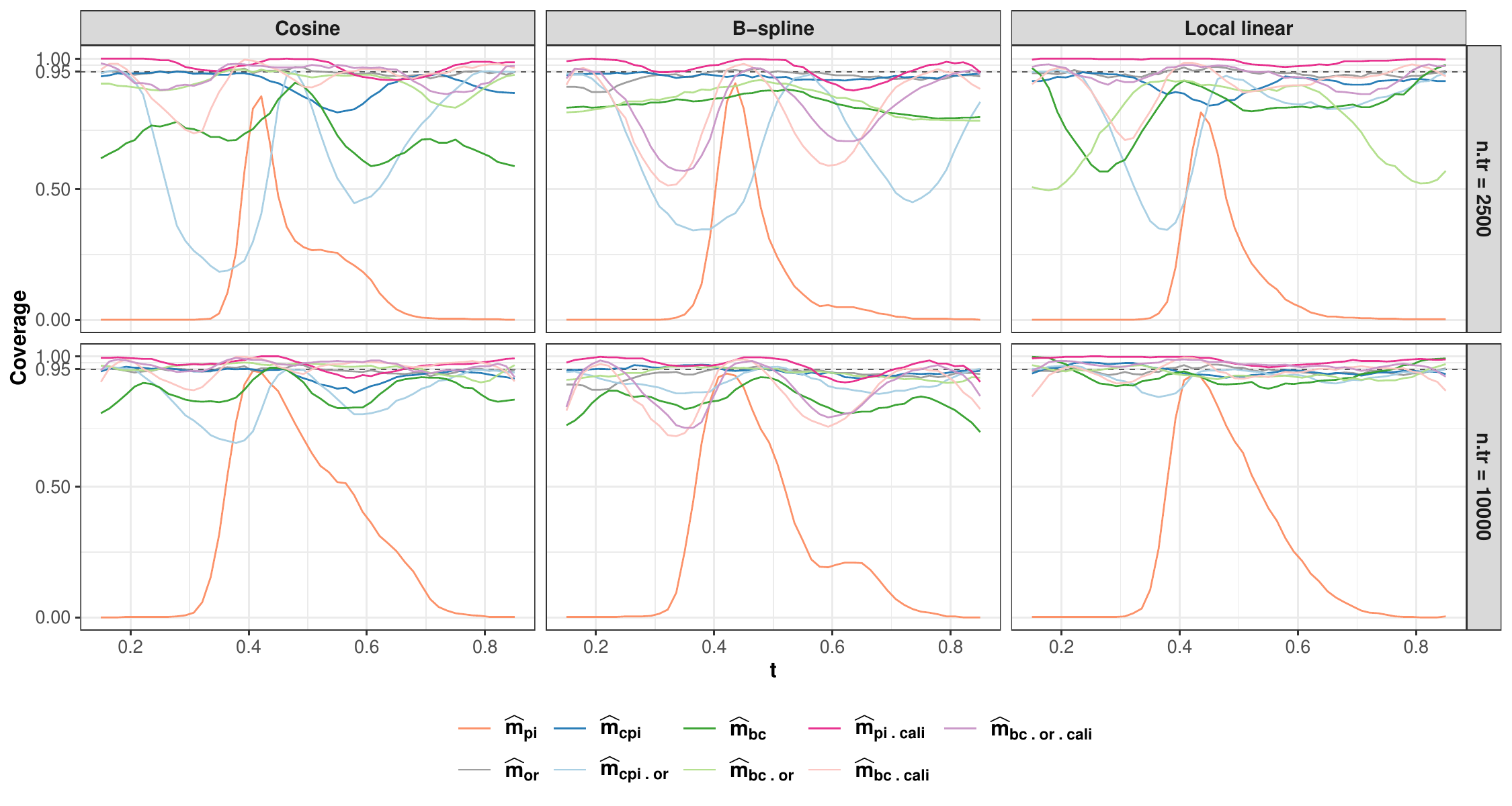}
    \caption{Pointwise coverage relative to the matched same-replication semi-oracle target for     \textbf{M3 (Sinusoidal function with a localized bump + linear term in $X$):} $\mu(X) = 1 + 0.1\sin\{2\pi r(X)\} +   \sin\{2\pi r(X)\} *\exp[-40\{r(X) - 0.3\}^2]+\beta_g^T X$, where $\beta_g$ is a fixed coefficient vector.  }
    \label{fig:coverage-ntr_S2M3}
\end{figure}

\begin{figure}[htbp]
    \centering
    \includegraphics[width=\textwidth]{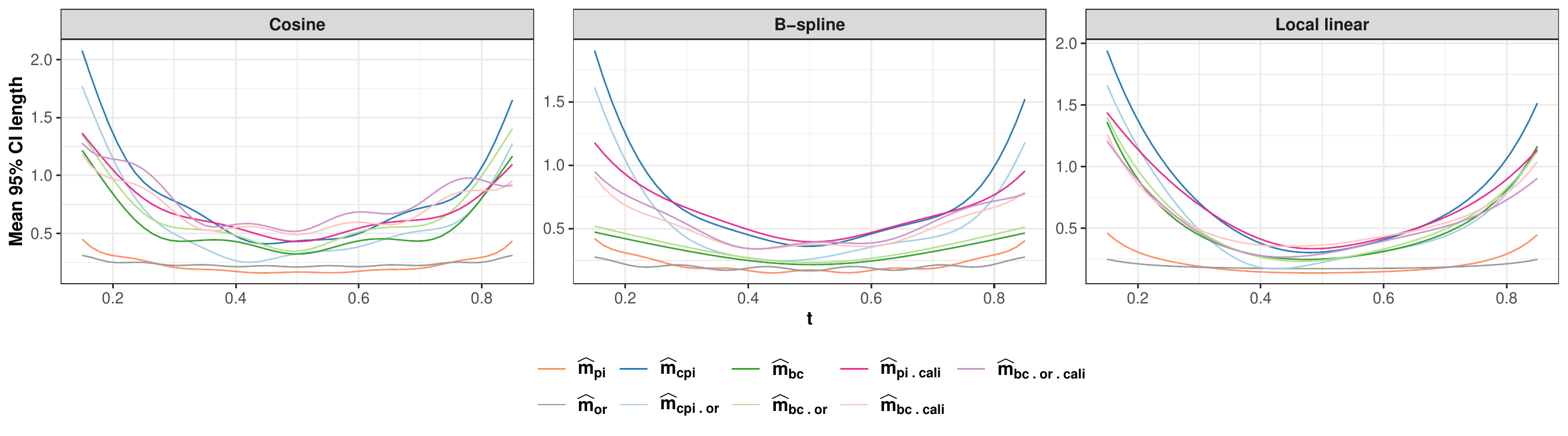}
    \caption{Pointwise confidence interval length for     \textbf{M3 (Sinusoidal function with a localized bump + linear term in $X$):} $\mu(X) = 1 + 0.1\sin\{2\pi r(X)\} +   \sin\{2\pi r(X)\} *\exp[-40\{r(X) - 0.3\}^2]+\beta_g^T X$, where $\beta_g$ is a fixed coefficient vector. }
    \label{fig:cilength-ntr_S2M3}
\end{figure}

\begin{figure}[htbp]
    \centering
    \includegraphics[width=\textwidth]{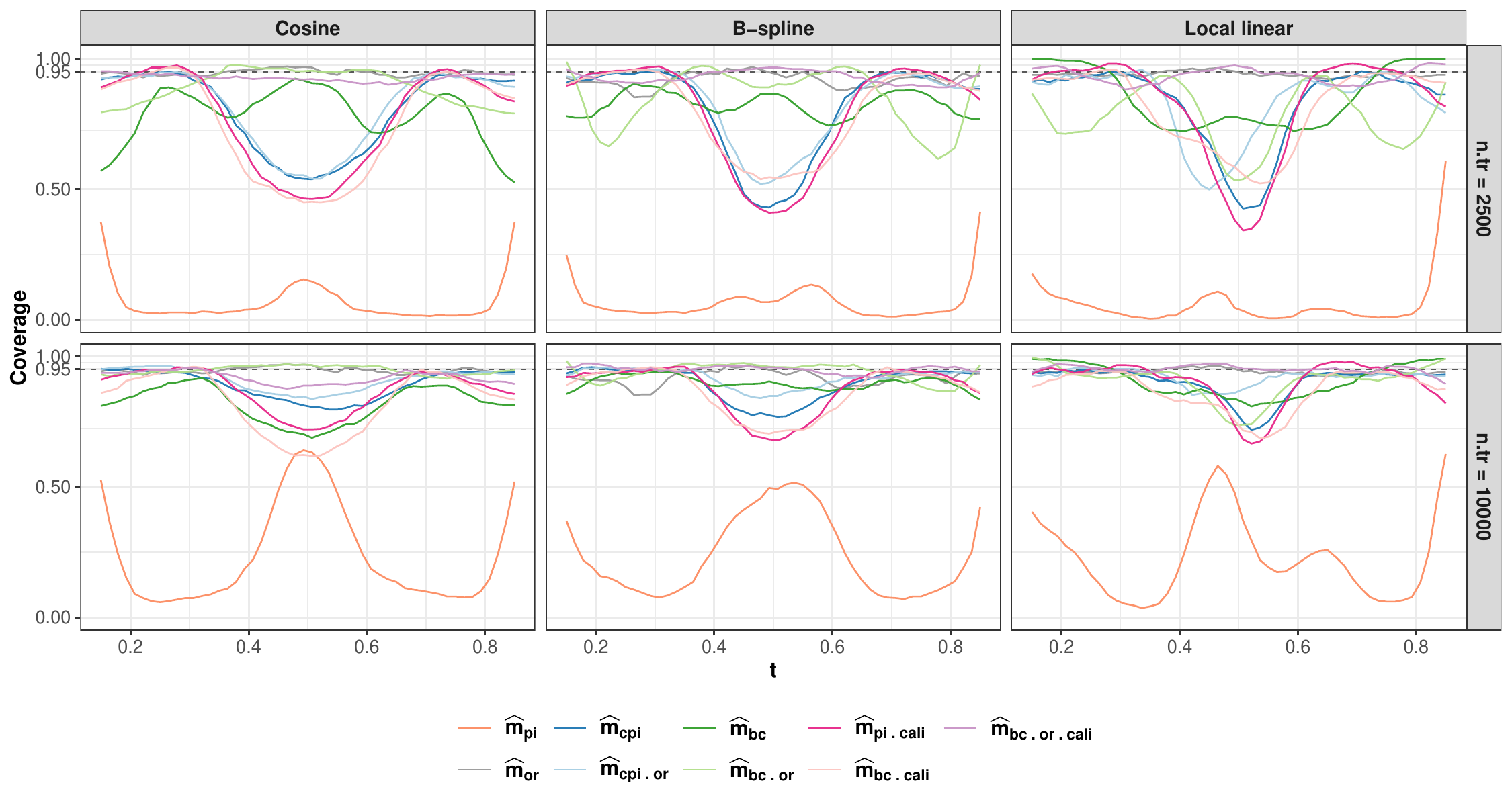}
    \caption{Pointwise coverage relative to the matched same-replication semi-oracle target for \textbf{M1 (sin–cosine function):} 
    $\mu(X) = 0.15\sin(10\pi r(X)) + \cos(2\pi r(X))$.}
    \label{fig:coverage-ntr_S1M1}
\end{figure}

\begin{figure}[htbp]
    \centering
    \includegraphics[width=\textwidth]{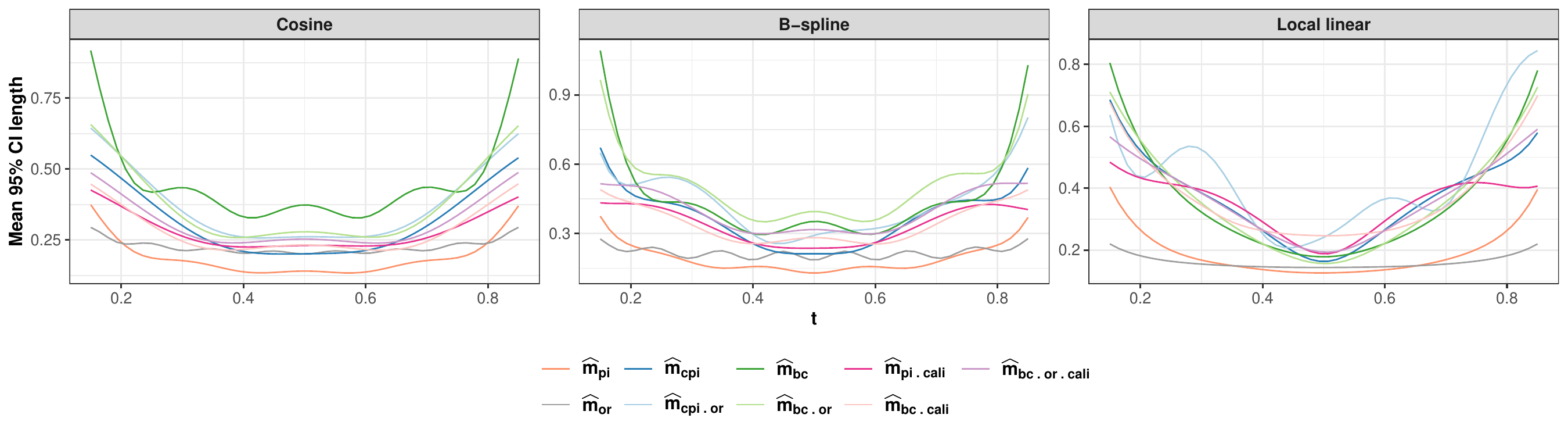}
    \caption{Pointwise confidence interval length for \textbf{M1 (sin–cosine function):} 
    $\mu(X) = 0.15\sin(10\pi r(X)) + \cos(2\pi r(X))$. }
    \label{fig:cilength-ntr_S1M1}
\end{figure}
\begin{figure}[htbp]
    \centering
    \includegraphics[width=\textwidth]{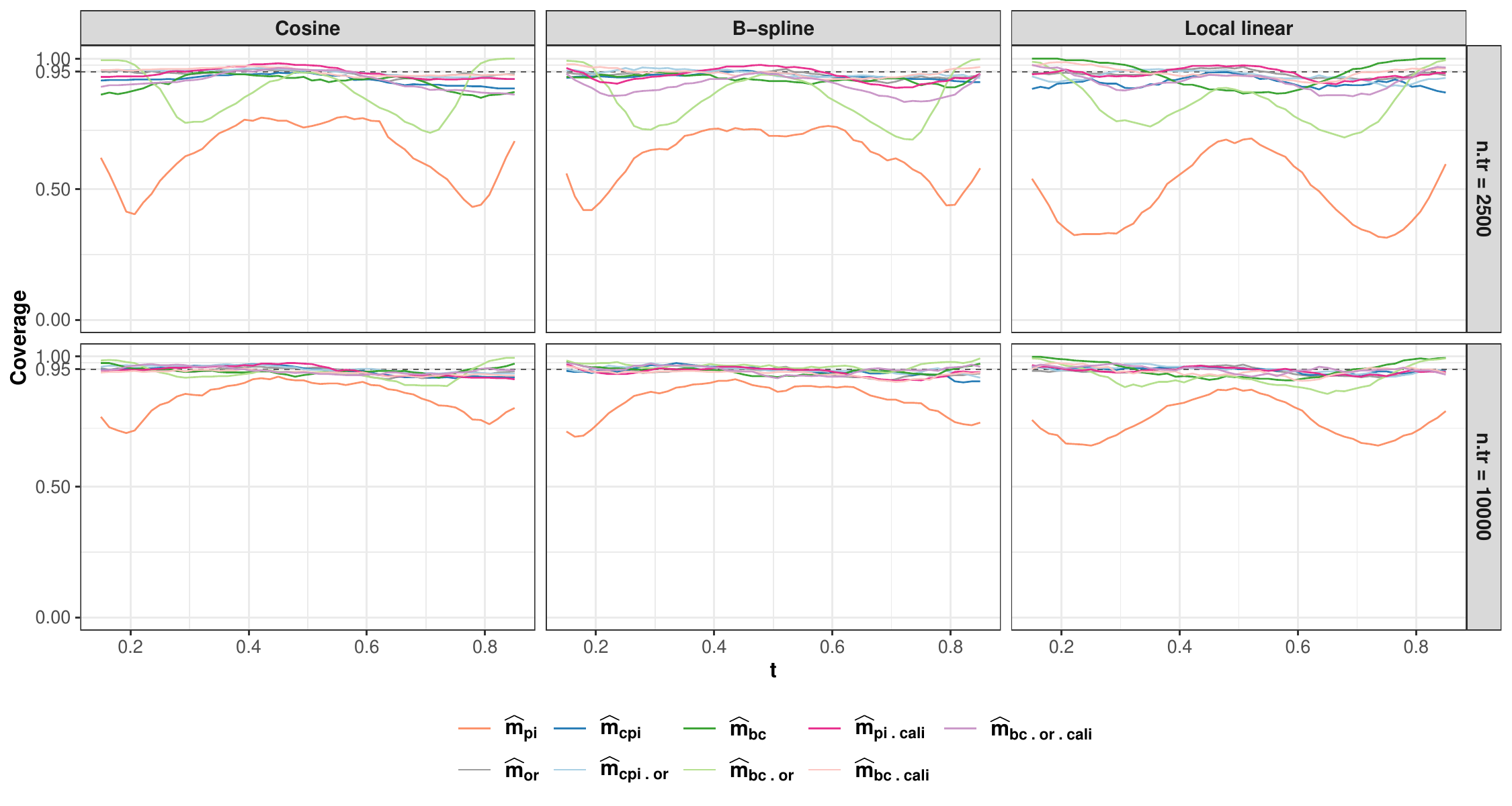}
    \caption{Pointwise coverage relative to the matched same-replication semi-oracle target for \textbf{M2 (linear model with small sinusodal perturbation): $\mu(X) = 4+ 3r(X)+0.15\sin\{4\pi r(X)\}$}}
    \label{fig:coverage-ntr_S1M2}
\end{figure}

\begin{figure}[htbp]
    \centering
    \includegraphics[width=\textwidth]{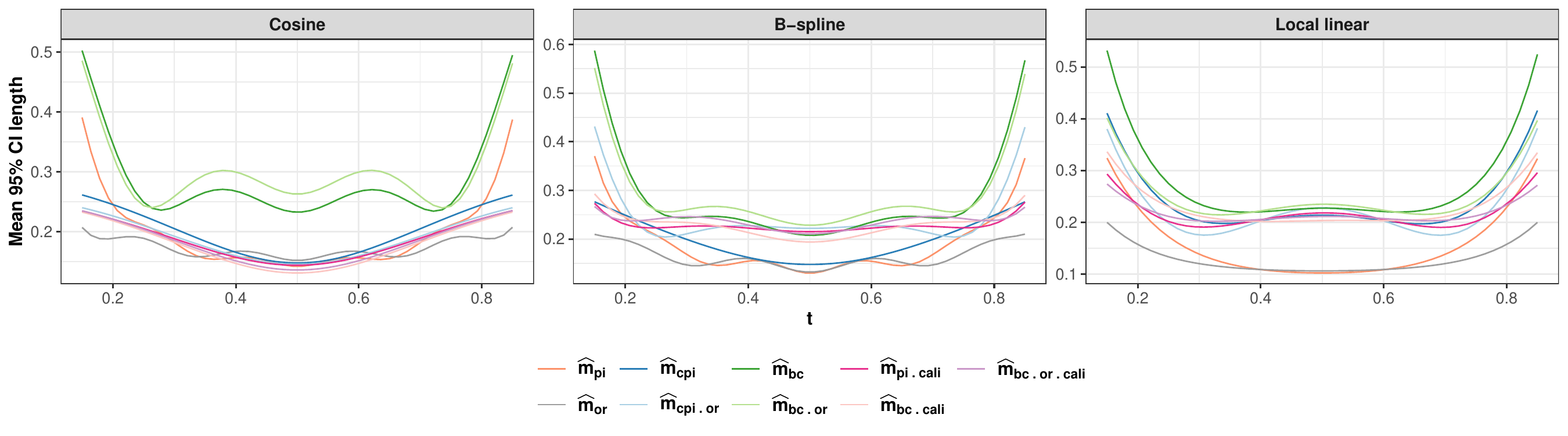}
    \caption{Pointwise confidence interval length for  \textbf{M3 (Sinusoidal function with a localized bump): $\mu(X) = 1 + 0.1\sin\{2\pi r(X)\} +   \sin\{2\pi r(X)\} *\exp[-40\{r(X) - 0.3\}^2]$} }
    \label{fig:cilength-ntr_S1M2}
\end{figure}

\begin{figure}[htbp]
    \centering
    \includegraphics[width=\textwidth]{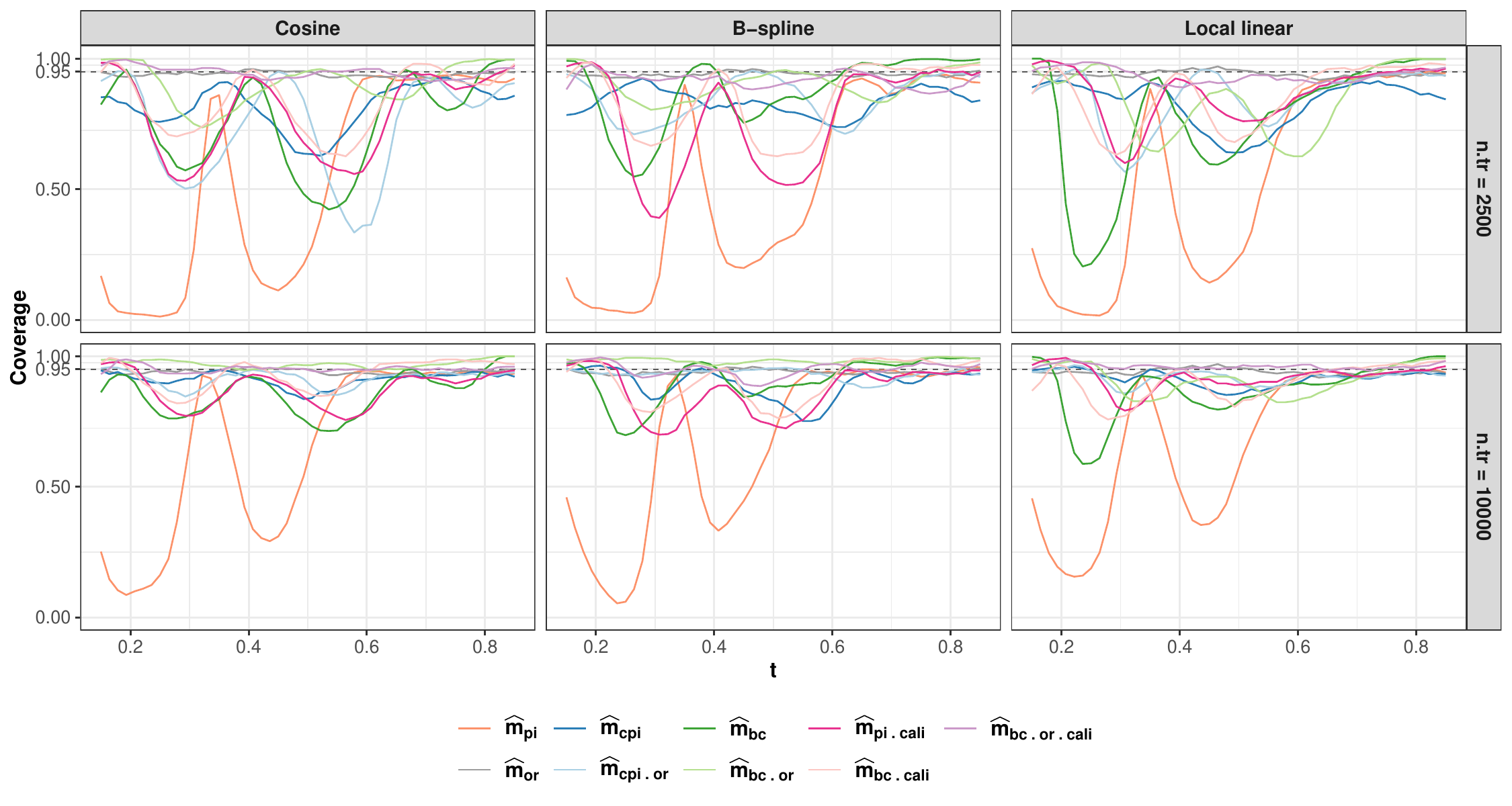}
    \caption{Pointwise coverage relative to the matched same-replication semi-oracle target for     \textbf{M3 (Sinusoidal function with a localized bump): $\mu(X) = 1 + 0.1\sin\{2\pi r(X)\} +   \sin\{2\pi r(X)\} *\exp[-40\{r(X) - 0.3\}^2]$}.}
    \label{fig:coverage-ntr_S1M3}
\end{figure}

\begin{figure}[htbp]
    \centering
    \includegraphics[width=\textwidth]{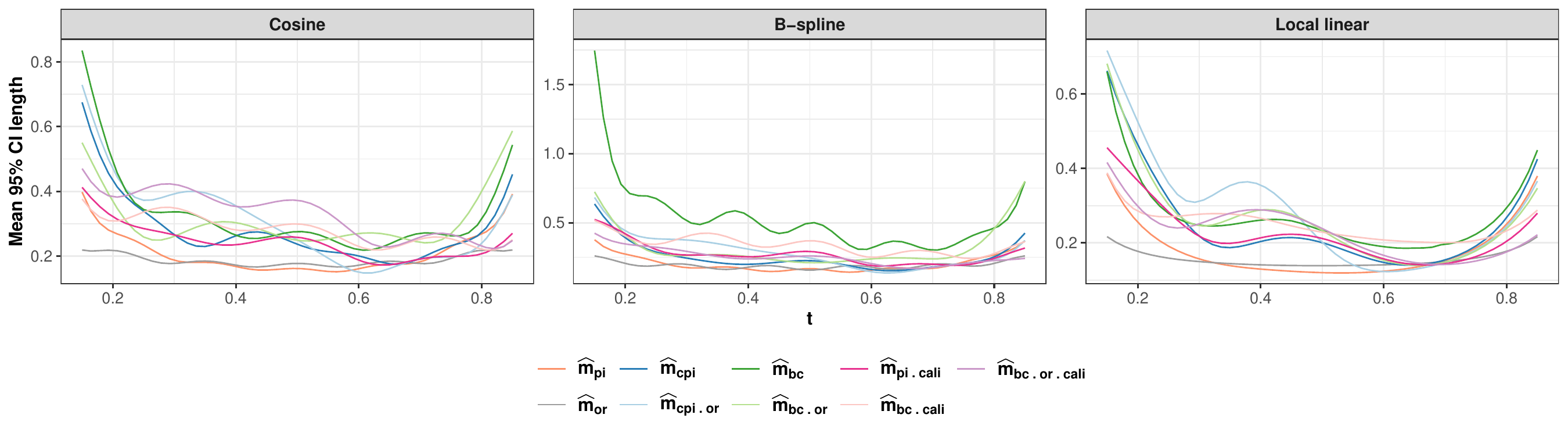}
    \caption{Pointwise confidence interval length for     \textbf{M3 (Sinusoidal function with a localized bump): $\mu(X) = 1 + 0.1\sin\{2\pi r(X)\} +   \sin\{2\pi r(X)\} *\exp[-40\{r(X) - 0.3\}^2]$}. }
    \label{fig:cilength-ntr_S1M3}
\end{figure}
\end{document}